\documentclass[aps,prd,showpacs,twocolumn,superscriptaddress,floatfix]{revtex4-2}  
\usepackage[T1]{fontenc}
\usepackage{lmodern} 
\usepackage{graphicx}  
\usepackage{dcolumn}   
\usepackage{bm}        
\usepackage{amssymb}   
\usepackage{amsmath}
\usepackage{bbold}
\usepackage{array}
\usepackage{makecell}
\usepackage{braket}
\usepackage{supertabular,booktabs}
\usepackage{listings}
\usepackage{graphicx}
\usepackage{comment}
\usepackage{needspace}

\usepackage{amsthm}
\usepackage{ragged2e}
\usepackage{etoolbox}

\AtBeginEnvironment{thebibliography}{\justifying}
\makeatletter
\renewenvironment{proof}[1][Proof]{%
  \par\pushQED{\qed}%
  \normalfont\topsep6\p@\@plus6\p@\relax
  \trivlist
  \item[\hskip\labelsep\bfseries #1.]%
}{%
  \popQED\endtrivlist\@endpefalse
}
\makeatother

\usepackage{titlesec} 
\usepackage[colorlinks,citecolor=blue,urlcolor=blue,hypertexnames=true]{hyperref}
\usepackage[caption=false]{subfig}

\usepackage[usenames,dvipsnames,svgnames,table]{xcolor} 
 
\begin{document}

\title{Quantum algorithm for solving differential equations using SLAC derivatives}

\author{Rakshit M. Gharat}
\email{rakshit.gharat@students.mq.edu.au}
\affiliation{School of Mathematical and Physical Sciences, Macquarie University, 2109 NSW, Australia}
\author{Gopikrishnan Muraleedharan}
\affiliation{BTQ Technologies, 16-104 555 Burrard Street, Vancouver BC, V7X 1M8 Canada}
\author{Dominic W. Berry}
\affiliation{School of Mathematical and Physical Sciences, Macquarie University, 2109 NSW, Australia}
\author{Gavin K. Brennen}
\affiliation{School of Mathematical and Physical Sciences, Macquarie University, 2109 NSW, Australia}
\affiliation{BTQ Technologies, 16-104 555 Burrard Street, Vancouver BC, V7X 1M8 Canada}


\begin{abstract}
In numerical approaches to solving differential equations on a lattice, a representation of the derivative operator that correctly matches the continuum behaviour of functions of momentum up to the band limit must be non-local. We present the construction of efficient linear-combination-of-unitaries ($\mathrm{LCU}$)-based block-encodings for the first-order derivative and Laplacian operators in the non-local \(N=2^n\)-dimensional SLAC representation. We use state-preparation techniques designed for smoothly decaying functions to prepare the dense $\mathrm{LCU}$ amplitudes with high success probability and low gate cost. Furthermore, we demonstrate how Shannon wavelet transforms can be applied to these block-encodings to obtain multiscale representations of the SLAC derivative operators. We then show how to apply a diagonal preconditioner that reduces the condition number of these matrices in the multiscale wavelet basis to a small constant. This enables the solution of partial differential equations (PDEs) with SLAC-discretised derivative operators on a finite lattice using the quantum linear solving algorithm ($\mathrm{QLSA}$). For a $d$-dimensional PDE, after projection away from the nullspace, the resulting quantum linear-system algorithm has overall gate complexity $\mathcal{O}(dn^3\alpha^{(k)}\log(1/\varepsilon))$, where $\alpha^{(k)}$ is the subnormalisation factor of the order-$k$ SLAC block-encoding and $\varepsilon$ denotes the algorithmic approximation error.
\end{abstract}
\maketitle



\section{Introduction}\label{sec:introduction}
Lattice discretisation is a fundamental ingredient in the simulation of many-body quantum systems and continuum field theories. A good discretisation should preserve the relevant continuum structure as faithfully as possible, since violations of these structures can introduce unphysical lattice artefacts. In non-relativistic quantum mechanics, the spatial momentum operator is given by $k \equiv -\,i\,d/dx$. For $k$ to define a physical observable, it must be Hermitian, which requires the derivative operator $d/dx$ to be anti-Hermitian. The standard one-sided finite-difference derivative does not satisfy this property, as its momentum-space symbol carries an unwanted phase. The symmetric, or central, finite difference restores anti-Hermiticity and hence unitarity under evolution, but its Fourier symbol vanishes at the edge of the first Brillouin zone. In fermionic theories, this gives rise to spurious zero modes, or doublers. More broadly, the Nielsen--Ninomiya theorem shows there is no translationally invariant lattice Dirac operator that simulataneously satisfies the following: spatial locality, i.e. superpolynomial or faster falloff of matrix elements with lattice site separation; the correct continuum limit of the momentum operator; absence of fermionic doubling; and chiral invariance \cite{NIELSEN198120}. Indeed a local lattice derivative operator implies the Fourier transform is periodic in momentum $k$ with period $2\pi/a$ where $a$ is the lattice spacing. This effects not only simulations of fermions on a lattice but also of bosons where local lattice Laplacians give an incorrect dispersion relation which impacts calculations of quantities like the specific heat even at intermediate temperatures \cite{finiteswt}. This motivates the study of a non-local discretisation that preserves the correct continuum physics for momenta within the first Brillouin zone: $k\in(-\pi/a,\pi/a]$.

To address this issue, Drell, Weinstein, and Yankielowicz proposed the SLAC derivative~\cite{slac1}, whose discrete Fourier symbol is chosen to be exactly the continuum momentum within the first Brillouin zone. The resulting momentum operator is Hermitian and avoids mirror modes. The price paid for this exact momentum-space fidelity is that the derivative becomes highly non-local in position space, leading to dense matrices and potentially large computational costs on classical computers~\cite{Quinn:1986mzs,slac2}.

The SLAC derivative is defined in momentum space and then transformed back to position space using the inverse Fourier transform. This makes sampling theory a natural language for describing it. In particular, the Shannon wavelet transform is built from the sinc function, which is the standard interpolation kernel for band-limited functions~\cite{shannon1,finiteswt,shannon4}. Since the Shannon transform separates momentum modes exactly~\cite{finiteswt}, it can split the lattice modes into infrared (IR) and ultraviolet (UV) sectors. This makes it a useful tool for constructing multi-scale representations of SLAC derivative operators.

Multi-scaled operator representations are powerful tools for studying continuum field theories~\cite{qft1,qft2,qft3,qft4,qft5}. They introduce an explicit scale degree of freedom, allowing coarse-grained features to be separated from fine-scale structure, and they can also support compressed representations of operators. For derivative operators, wavelet-based multi-scale representations have an additional advantage: they can admit diagonal wavelet preconditioners that substantially reduce the condition number of the discretised operator. Such preconditioners are important for solving differential equations in a wavelet basis. Recently, the authors in Ref.~\cite{fspde} showed that elliptic partial differential equations can be solved efficiently, and exponentially faster than known classical methods, by combining the HHL algorithm~\cite{hh1,hhl2,hhl3} with a diagonal wavelet preconditioner. That work used discrete and compactly supported wavelet bases which yield sparse and local derivative operators. Improved QLSAs (see \cite{Childshhl} and \cite{Ambainis:2010wfy}) achieve nearly linear dependence on the condition number and the sparsity. In the compact-wavelet setting of~\cite{fspde}, an appropriate wavelet basis and diagonal preconditioner ensure that both the sparsity and condition number remain effectively constant. However, as discussed above, simulations using that basis would experience artifacts like fermion doubling as observed in  \cite{alves2024}, as would simulations based on continuous but local wavelet bases like those studied in \cite{PhysRevD.109.016018}.

In contrast, SLAC  derivatives are dense, non-local operators. Matrix density is a basis dependent quantity and indeed on a periodic lattice the SLAC derivatives are diagonal in the Fourier basis. Yet, in the presence of terms which break translational invariance (see Sec.\ref{brokentransinv}) there are compelling reasons to work in the position basis as we do for this work. We show that the SLAC derivative operators can be block-encoded efficiently using the Linear Combination of Unitaries, ($\mathrm{LCU}$) \cite{Childs:2012gwh} framework, provided the required dense ancillary states can be prepared efficiently. To achieve this, we adapt nested-box inequality-test state-preparation techniques originally developed in~\cite{nestedboxes}. These routines allow the required $\mathrm{LCU}$ coefficient states to be prepared with controlled error and only constant success-probability overhead.

Building on these block-encodings, we then show how recursive applications of the Quantum Shannon Wavelet Transform yield efficient multi-scale realisations of SLAC derivative operators. We further show how diagonal wavelet preconditioning can be incorporated to reduce the condition number of the resulting multi-scale matrices while preserving efficient quantum implementability. Together, these ingredients provide a framework for implementing dense, non-local derivative operators and their multi-scale preconditioned forms in quantum circuits, enabling their use in quantum algorithms for the simulation and solution of differential equations.

We expect the constructions developed in this work to extend naturally to physical systems where preserving continuum properties after discretisation is important. More broadly, the $\mathrm{LCU}$ state-preparation protocols used here apply to dense coefficient functions that vary smoothly and analytically with the computational basis index $j$, provided they can be implemented with low $T$-cost arithmetic circuits. The corresponding success probabilities, however, depend on the particular functional form being prepared.

One of the main goals of this work is to construct efficient block-encodings of \(N=2^n\)-dimensional, or equivalently \(n\)-qubit, SLAC derivative operators. Throughout the paper, we use the notation \((\alpha,a,\epsilon)\)-block-encoding, where \(\alpha\) is the subnormalisation factor, \(a\) is the number of ancilla qubits used by the block-encoding, and \(\epsilon\) is the allowed implementation error. The main contributions of this work are summarised as follows:
\begin{itemize}
 \item Proposition~\ref{corl1}
shows that the $n$-qubit SLAC Laplacian admits an optimal block-encoding.
More precisely, we construct an $\varepsilon$-approximate $((\pi^2+24)/3,\,\mathcal{O}(n),\,\varepsilon)$-- block-encoding of the SLAC Laplacian using the $\mathrm{LCU}$ framework, with overall gate complexity $\mathcal{O}(n^2)$ and ancilla cost $\mathcal{O}(n)$. 

\item Proposition~\ref{corl2} provides
an efficient block-encoding of the $n$-qubit first-order SLAC derivative. We show that this operator admits an $\varepsilon$-approximate $\bigl(2(n-1),\,\mathcal{O}(n),\,\varepsilon\bigr)$ block-encoding within the $\mathrm{LCU}$ framework, with overall gate complexity $\mathcal{O}(n^2)$
and using $\mathcal{O}(n)$ ancilla qubits.

\item Result~\ref{re:msslac} shows that the multiscale representation of an \(N\)-dimensional SLAC derivative operator, for derivative orders \(k=1,2\), can be constructed from a single application of its block-encoding together with \(2(n-1)\) applications of the quantum Shannon wavelet transform (\(\mathrm{QSWT}\)). Since both the SLAC block-encoding and each \(\mathrm{QSWT}\) application have gate complexity \(\mathcal{O}(n^2)\), the full multiscale construction requires \(\mathcal{O}(n)\) \(\mathrm{QSWT}\) calls, one SLAC block-encoding query, and has overall gate complexity \(\mathcal{O}(n^3)\).

\item Result~\ref{re:SLACPDE} addresses the solution of $d$-dimensional PDEs discretised using the SLAC formalism. For the corresponding linear system $Au=b$, given access to the multiscale preconditioned block-encoding of $A$ together with a preparation routine $\mathcal{P}_b$ for the state $\ket{b}$, one can generate an $\varepsilon$-approximate solution state $\ket{\tilde u}\propto A^{-1}\ket{b}$. The protocol uses $\mathcal{O}(1)$ calls to $\mathcal{P}_b$ and $\mathcal{O}(\alpha^{(k)}\log(1/\varepsilon))$ queries to the multiscale preconditioned block-encoding. Since each such query contains one call to the SLAC block-encoding and $\mathcal{O}(n)$ applications of the $d$-dimensional $\mathrm{QSWT}$ protocol, the overall gate complexity is $\mathcal{O}(dn^3\alpha^{(k)}\log(1/\varepsilon))$. Here, $\alpha^{(k)}$ denotes the subnormalisation factor associated with the block-encoding of the order-$k$ SLAC operator, with $k=1,2$.
\end{itemize}

The remainder of this work is organised as follows. 
Section~\ref{prelim} reviews the background material needed for the subsequent sections. Section~\ref{sec:swt} reviews the wavelet-based tools and analysis used to expose the multi-scale structure of the SLAC operator. Section~\ref{sec: slacder} provides a brief overview of the SLAC formalism for discretising derivative operators.
Building on this, Section~\ref{sec:qswt} presents an efficient circuit implementation of the quantum Shannon wavelet transform using elementary quantum gates.  Section~\ref{sec: slaclap} then describes the block-encoding of the SLAC Laplacian based on our state-preparation techniques. In Section~\ref{sec: slac1be}, these techniques are extended to block-encode the first-order SLAC derivative, and we analyse the associated truncation error and gate complexity. Section~\ref{sec:multisc} combines the quantum Shannon wavelet transform with these block-encodings to obtain an efficient implementation of the multi-scale SLAC operator. 
Section~\ref{sec: slaclc} briefly explains how existing $\mathrm{LCU}$ methods allow us to form a block-encoding of a linear combination of SLAC derivatives from the individual block-encodings. In Section~\ref{sec:precon}, we show how to apply a diagonal wavelet preconditioner efficiently by first block-encoding the preconditioner itself. Finally, Section~\ref{sec: HHL} demonstrates how the block-encoding of the preconditioned SLAC derivative operator can be used within a QLSA to solve differential equations discretised via the SLAC formalism.
\section{Background}\label{prelim}
In this section we briefly review the preliminaries pertinent to subsequent sections. First, we describe the theory of shannon wavelet transform and its application for achieving perfect momentum separation. Then we provide a brief background on the analytical form of SLAC derivative operators discretised on a finite lattice of size $N$ having periodic boundary conditions.
\subsection{Shannon wavelet transform and perfect momentum separation}\label{sec:swt}
A wavelet transformation typically involves splitting the input data into two identical copies, which are then filtered using a high-pass (UltraViolet) and a low-pass (InfraRed) filter. After the input data passes through the high-pass (UV) and low-pass (IR) filters, it is divided into two separate signals: one containing high-frequency components (details) and the other containing low-frequency components (approximations). The resulting signals are subsequently processed by downsampling, effectively reducing the degrees of freedom by half per branch. This procedure is repeated multiple times, ensuring that the transformation remains unitary throughout.

There exists a wide range of wavelet transformations, each offering a different balance between filter quality and locality. For instance, consider Daubechies wavelets of order $N$ whose filter functions operate on $N$ neighboring lattice sites. The quality of these wavelets improves with increasing $N$ but this comes at the cost of reduced locality. Such local filters, by definition, are constrained to a small region in the time domain, which limits their ability to precisely separate frequencies. This results in some overlap or ``smearing'' between low and high frequencies in the filtered output. These local filters are often used in signal processing and data compression. The partial derivative operator when discretised in Daubechies wavelet basis admits a sparse representation~\cite{finiteswt}. These sparse (and local) multi-scale derivative operators were employed to simulate a free bosonic quantum field theory in~\cite{qft1,qft3}. In contrast, (non-local) multi-scale SLAC derivative operators were applied to renormalize field theory in~\cite{finiteswt} using the Shannon wavelet transform. In this section, we review the key protocols developed in that work, which will serve as the foundation for designing a quantum protocol that implements the Shannon wavelet transform via efficient quantum circuits.

\subsubsection{Shannon Sampling Theorem}
In sampling theory, an analog signal is typically reconstructed by first taking samples at regular intervals, and then interpolating these samples using suitable kernels. More formally, one first samples the signal $f$ at intervals of length $s$, and then filters it by convolution with an interpolation kernel $\phi_s$ to recover the original function.

The Shannon--Whittaker sampling theorem states that if the Fourier transform of a continuous analog signal $f$, denoted by $\hat{f}(\omega)$, is supported within the frequency band $\omega \in [-\pi/s, \pi/s]$, then the signal can be exactly reconstructed using the sinc interpolation kernel defined as:
\begin{equation*}
\phi_s(t) = \frac{\sin(\pi\,t/s)}{\pi\,t/s}.
\end{equation*}

Specifically, the signal reconstruction is given by:
\begin{equation}
f(t) = \sum_{n=-\infty}^{\infty} f(n s)\,\phi_s(t - n s).
\end{equation}

This sampling process can equivalently be interpreted in the frequency domain as making the Fourier transform $\hat{f}$ periodic with period $2\pi/s$. Explicitly, the sampled signal's Fourier transform becomes:
\begin{equation}
\hat{f}_d(\omega) = \frac{1}{s}\sum_{k=-\infty}^{\infty}\hat{f}\left(\omega - \frac{2\pi k}{s}\right).
\end{equation}
Thus, the sampled signal's spectrum consists of the original spectrum together with infinitely many translated copies shifted by multiples of $2\pi/s$. If the signal is band-limited to frequencies below the Nyquist frequency $\pi/s$, these shifted copies remain disjoint and do not overlap. As a result, applying the ideal low-pass sinc filter cleanly extracts the central copy of the spectrum, thereby reconstructing the original signal exactly.

For a more formal and intricate discussion on the sampling theorem, the reader is referred to Chapter 3 of~\cite{shannon1}.

\subsubsection{Multiresolution analysis and Shannon approximations}
A central concept in wavelet theory is the \emph{Multiresolution Analysis} (MRA), which organizes the Hilbert space, $L^2(\mathbb{R})$ into a hierarchy of subspaces of increasing resolution. Each subspace captures information at a particular scale, and together they provide a complete and systematic framework for analyzing functions. Formally, the wavelet basis defines a Multiresolution Analysis (MRA) of the Hilbert space \( L^2(\mathbb{R}) \), characterized by a nested sequence of closed subspaces:
\begin{equation}
    \dots \subset V_2 \subset V_1 \subset V_0 \subset V_{-1} \subset V_{-2} \subset \dots
\end{equation}
satisfying
\begin{equation}
    \bigcap_{r \in \mathbb{Z}} V_r = \{0\}, \quad \bigcup_{r \in \mathbb{Z}} V_r = L^2(\mathbb{R}).
\end{equation}
By introducing the subspace \( W_r \) as the orthogonal complement of \( V_r \) in \( V_{r-1} \):
\begin{equation}
    V_{r-1} = V_r \oplus W_r,
\end{equation}
one obtains a multi-scale decomposition of the space \( L^2(\mathbb{R}) \):
\begin{equation}
    L^2(\mathbb{R}) = \bigoplus_{r \in \mathbb{Z}} W_r.
\end{equation}
The MRA provides a framework to extract the scale-dependent features of a function as it is resolved at progressively finer scales. In particular, the band-limited functions described by Shannon’s sampling theorem naturally give rise to multiresolution approximations. The closed subspaces $V_j$ can be defined as the set of band-limited functions whose Fourier support lies within $[-2^{-j}\pi,\,2^{-j}\pi]$. According to the sampling theorem, the sinc kernel forms an orthonormal basis for the subspace $V_0$. This implies that at resolution $2^{-j}$, the best approximation of $f \in L^2(\mathbb{R})$ is given by its orthogonal projection $P_{V_j} f$ onto the band-limited subspace $V_j$. In the Fourier domain, this projection takes the form,
\begin{equation}
    \widehat{P_{V_j} f}(\omega) \;=\; \widehat f(\omega)\,\mathbf{1}_{[-2^{-j}\pi,\;2^{-j}\pi]}(\omega),
\end{equation}
which corresponds to an ideal low-pass filter with cutoff frequency $\pm 2^{-j}\pi$. However, because this filter introduces a sharp cutoff in the frequency domain, the projection exhibits an algebraic decay of order $1/|t|$ in the time domain, even when the original function is compactly supported. This follows from the Fourier uncertainty principle: sharply restricting frequency localization inevitably produces long-range tails in time.

For more details on MRA and Shannon approximation, readers are referred to Chapter 7 of~\cite{shannon1}.

\subsubsection{Shannon wavelets}
Shannon wavelets provide the simplest example of a wavelet system, defined directly in the Fourier domain. They are constructed from the scaling function $\phi$ and associated filter functions, which relate the scaling and wavelet spaces in the MRA framework. In signal processing, important classes of wavelets are studied via Fourier transforms derived from the general formula:
\begin{align}
\widehat\psi(\omega)
&=
\frac{1}{\sqrt{2}}\;\widehat g\left(\frac{\omega}{2}\right)\widehat\phi\left(\frac{\omega}{2}\right)
\nonumber \\
&= \frac{1}{\sqrt{2}}\;e^{-i\omega/2}\,\widehat h^{*}\left(\frac{\omega}{2} + \pi\right)\widehat\phi\left(\frac{\omega}{2}\right).
\end{align}
Here, \(\psi\) is the mother wavelet, \(\phi\) is the scaling (father) function, and \(h, g\) represent corresponding filter functions.

In the Shannon MRA, functions are approximated as band-limited signals. Consequently, the scaling and filter functions are explicitly defined in the Fourier domain by
\[
\widehat{\phi}(\omega)=\mathbf{1}_{[-\pi,\pi]}(\omega),\quad
\widehat{h}(\omega)=\sqrt{2}\,\mathbf{1}_{[-\pi/2,\pi/2]}(\omega).
\]

With this choice, the corresponding mother wavelet takes the form
\begin{equation}
\widehat\psi(\omega)=
\begin{cases}
e^{-i\omega/2}, & \omega\in[-2\pi,-\pi]\cup[\pi,2\pi],\\[6pt]
0, & \text{otherwise},
\end{cases}
\end{equation}
leading in the time domain to
\[
\psi(t)=
\frac{\sin\bigl(2\pi t\bigr)-\cos\bigl(\pi t)}{\pi/2 - \pi t}.
\]

Since the Fourier transform of the Shannon wavelet has compact frequency support, its inverse Fourier transform—i.e., the wavelet in the time domain—is infinitely differentiable (\(C^\infty\)) but not compactly supported, exhibiting slow asymptotic decay. Moreover, because \(\widehat\psi(\omega)\) vanishes in a neighborhood around \(\omega=0\), all its frequency-domain derivatives vanish at the origin. Consequently, the Shannon wavelet possesses infinitely many vanishing moments. Finally, the discontinuities of \(\widehat\psi(\omega)\) at the points \(\omega=\pm\pi, \pm 2\pi\) imply that \(|\psi(t)|\) decays only algebraically, specifically as \(1/|t|\), at infinity. These properties highlight the trade-off inherent in the Shannon wavelet: it achieves perfect frequency localization and infinitely many vanishing moments, but only at the cost of non-locality in time. This balance will be important in the subsequent discussion.

\subsubsection{Wavelet-based renormalisation in the Shannon limit}
As discussed earlier, the basis functions of the Shannon wavelet transform are band-limited, meaning that their momentum modes are perfectly separated. This property was exploited in~\cite{finiteswt} to construct a renormalisation scheme for quantum field theory. Below, we briefly summarise the key features of this wavelet-based renormalisation scheme in the Shannon limit of perfect momentum separation. This discussion serves to motivate the efficient construction of the Quantum Shannon Wavelet Transform in the following section.\\
In the Shannon limit, the wavelet functions possess infinite support, and the corresponding filters achieve exact separation between low- and high-momentum modes. However, this comes at the cost of complete non-locality in the filter functions. These non-local filters are nevertheless optimal in achieving perfect momentum separation. When the Shannon wavelet transform acts on field modes, the momentum modes $\{k\}$ are split cleanly into two equal halves:
\[
|k| \leq \frac{\pi}{2} \quad (\text{IR}) \quad \text{and} \quad \frac{\pi}{2} < |k| \leq \pi \quad (\text{UV}).
\]

In wavelet-based renormalisation schemes for field theories, it is often desirable to decouple the low-momentum (IR) and high-momentum (UV) modes. For wavelets with finite support (local filters), the filter functions overlap, which induces residual couplings between the IR and UV sectors. This overlap complicates the clean separation of scales, thereby making renormalisation more difficult. In contrast, in the Shannon limit the IR and UV sectors decouple completely. Such decoupling is essential for effective renormalisation, as it enables a sharp separation and independent treatment of different scales. One might attempt to achieve this using finite-range wavelets together with additional downstream disentanglers to cancel the residual IR--UV couplings. However, the combined action of a wavelet and its disentanglers always reproduces the Shannon limit, regardless of the initial wavelet choice. In other words, any attempt to enforce perfect scale separation using local filters ultimately converges to the non-local Shannon limit, where ideal frequency separation is naturally realised.

Applying perfect momentum separation to free bosonic or fermionic models effectively splits the dispersion relation of the excitation modes into an IR sector and a UV sector. By iteratively applying this procedure to the IR sector, one obtains an RG scheme that progressively ``zooms in'' on the low-$k$ region of the dispersion relation. This repeated zooming defines the Renormalisation Group (RG) fixed point. For example, if the initial dispersion relation is parabolic near $k=0$, successive applications of perfect momentum separation make the IR sector increasingly resemble an exact parabola, driving the flow toward an RG fixed point with a perfectly quadratic dispersion. At each stage of the RG flow, the IR and UV sectors decouple and evolve independently, while the couplings within each sector become increasingly long-ranged. This behaviour naturally links the RG fixed point--characterised by a perfectly quadratic dispersion---to the modified Laplacian obtained via SLAC discretisation, as will be discussed in the following sections. Since these derivative operators are diagonal in the Fourier basis, they can equivalently be represented in the scale-field basis provided by the Shannon wavelet transform. For a more detailed account of bosonic field renormalisation in the Shannon limit of perfect momentum separation, we refer the reader to~\cite{finiteswt}.

The Shannon limit of perfect momentum separation can be achieved by employing a straightforward sequence of Fourier transformations and permutations: 
\begin{enumerate}
    \item Start with a dataset containing \( N \) discrete values.
    \item Transform this dataset into \( k \)-space using a discrete Fourier transform (DFT) of size \( N \), which acts on permuted indices of the lattice.
    \item Manually divide the output into two equal halves: IR and UV. This can be implemented using unitary transformations.
    \item Transform each half back to real space using an inverse DFT of size \( N/2 \).
\end{enumerate}
This process effectively implements a perfect Shannon filter, achieving perfect momentum separation without the need for explicit wavelets or downsamplers.

\subsection{SLAC derivative operators}\label{sec: slacder}
Discretising derivatives on a lattice is a fundamental step in formulating field theories and quantum simulations. However, naive discretisations often introduce unphysical artefacts, such as fermion doubling or the loss of Hermiticity, which obscure the connection to the continuum theory. To address these issues, Drell, Weinstein, and Yankielowicz introduced the ideal SLAC derivative operator for a one-dimensional infinite lattice in the context of lattice gauge theory~\cite{slac1}. 

In contrast to local finite-difference schemes, the SLAC discretisation connects each lattice site to every other site along the line, thereby avoiding both fermion doubling and chiral-symmetry breaking. Below, we briefly review the motivation, following~\cite{finiteswt}, for employing the SLAC formalism to discretise derivatives.

A quantum field $\phi(x)$ in $d$ dimensions can always be expanded in plane waves, which form the natural eigenbasis of the Laplacian:
\begin{align}
\phi(x) &= \int \frac{d^d k}{(2\pi)^d}\,\tilde\phi(k)\,e^{i k\cdot x}, \\[6pt]
\tilde\phi(k) &= \int d^d x\,\phi(x)\,e^{-i k\cdot x}.
\end{align}
In the continuum, applying the Laplacian to a plane wave $\exp(ik\cdot x)$ simply yields $-k^2$ as the eigenvalue:
\begin{align}
\Delta &\;\equiv\; \sum_{j=1}^{d} \frac{\partial^2}{\partial x_j^2}, \\[6pt]
\Delta\,e^{i k\cdot x}
&= \Bigl(\sum_{j=1}^d \partial_{x_j}^2\Bigr)e^{i k\cdot x}
= \Bigl(\sum_{j=1}^d (i k_j)^2\Bigr)e^{i k\cdot x}
= -|k|^2\,e^{i k\cdot x}.
\end{align}
This corresponds to the continuum dispersion \(\omega\propto k^2\) and provides the benchmark against which lattice discretisations of the Laplacian should be judged. 

In practice, however, one often approximates \(\Delta\) on the lattice by the nearest-neighbour finite difference,
\begin{equation}
(\Delta_{\mathrm{FD}}\phi)_n 
= \frac{1}{a^2}\bigl[\phi_{n+1}-2\phi_n+\phi_{n-1}\bigr],
\end{equation}
which produces the dispersion
\[
\omega = 1 - \cos(a k),
\qquad
k \in \Bigl(-\tfrac{\pi}{a},\,\tfrac{\pi}{a}\Bigr],
\]
and only matches the ideal \(k^2\) form for \(\lvert k\rvert\ll \tfrac\pi a\). At larger momenta near the Brillouin-zone boundary, the dispersion deviates strongly from the continuum form, leading to unphysical artefacts such as fermion doubling. This motivates the search for alternative discretisations, such as the SLAC derivative, that reproduce the exact continuum spectrum within the Brillouin zone.

To restore the exact continuum behaviour at all lattice momenta, we demand a discretisation that remains \emph{translationally invariant} yet reproduces \(-k^2\) on the Brillouin zone. Any such linear, shift-invariant operator must act by convolution:
\[
(\Delta^{(2)}\psi)_n
=\sum_{n'}\Delta^{(2)}_{\,n-n'}\,\psi_{n'}
=\sum_{r}\Delta^{(2)}_r\,\psi_{n-r}
=(\Delta^{(2)}*\psi)_n.
\]
The kernel \(\Delta^{(2)}_r\) is then fixed by the inverse Fourier transform of the desired dispersion $(k^2)$:
\begin{align} \label{slaclco}
\Delta^{(2)}_r
&= -\frac{a}{2\pi}
   \int_{-\frac{\pi}{a}}^{\frac{\pi}{a}}
     k^2\,e^{i a k r}\,\mathrm{d}k \notag \\[6pt]
&=
\begin{cases}
-\dfrac{\pi^2}{3\,a^2}, & r = 0, \\[8pt]
-\dfrac{2\,(-1)^r}{r^2\,a^2}, & r \neq 0.
\end{cases}
\end{align}
These derivative operators are commonly studied within the framework of pseudo-differential operators (PDOs), and a general procedure for constructing block-encodings of such PDOs is presented in~\cite{PDOpaper}. In contrast, here we develop a new, explicit, and efficient block-encoding scheme tailored specifically to SLAC derivatives. A detailed complexity comparison between our construction and the PDO-based method is given in Appendix~\ref{app:PDO}.

The spectral accuracy of the SLAC representation can be assessed by comparing its Fourier symbols with those of the standard central finite-difference discretization. 
As shown in Fig.~\ref{fig:slac_symbol_comparison}, the upper-left panel compares the symbols for the first-order derivative, while the lower-left panel compares the corresponding symbols for the Laplacian.

Because SLAC derivative operators are exactly band-limited within the first Brillouin zone, they are naturally suited to a scale-field representation induced by the Shannon wavelet transform; see Section~\ref{sec:swt}. In the following sections, we construct an explicit and efficient implementation of the Quantum Shannon Wavelet Transform ($\mathrm{QSWT}$), together with $\mathrm{LCU}$ block-encodings of the SLAC operators based on nested-box inequality-test state preparation~\cite{nestedboxes}; see Sections~\ref{sec: slaclap} and~\ref{sec: slac1be}. Furthermore, in Section~\ref{sec:multisc}, we show how multi-scale SLAC derivative operators can be obtained by recursively applying the $\mathrm{QSWT}$ protocol developed in Section~\ref{sec:qswt}.

\section{Explicit construction of Quantum Shannon wavelet transform}\label{sec:qswt}
In this section, we present the construction of the $N=2^n$ dimensional Quantum Shannon wavelet transform (QSWT) using a quantum circuit, following the approach of~\cite{finiteswt}. As noted earlier, it is not necessary to explicitly invoke wavelet terminology here, since the Shannon wavelet transform can be realised through a combination of Fourier transforms and permutation matrices. Consider a discrete Fourier transform (DFT) acting on a square-normalisable finite complex vector $f=\{f_0,\ldots,f_{N-1}\}\in\mathbb{C}^N$ with $N \geq 4$ components, assuming $N=2^n$. This can be expressed as:
\begin{align*}
\tilde{f}_q &:= \frac{1}{\sqrt{N}}
  \sum_{n=0}^{N-1} e^{\frac{2\pi i n q}{N}} f_n,
  &
q &\in \left\{ -\frac{N}{2}, \dots, \frac{N}{2}-1 \right\}.
\end{align*}
Notice that the Fourier transform acts on shifted indices $q$. This operation can be realised in a quantum circuit by combining the Quantum Fourier Transform (QFT) with a permutation matrix $U_{\mathrm{shift}}$, whose leading $1$ lies in column $N/2$. We denote this combined operation as $\mathrm{SHIFTQFT}=U_{\mathrm{shift}}\cdot \mathrm{QFT}$. Since the permutation matrix can be implemented using adder circuits~\cite{addergidney}, the overall cost of implementing an $n-$ qubit $\mathrm{SHIFTQFT}$ is $\mathcal{O}\bigl(n^2\bigr)$ gates. \\

Before splitting the momentum modes into two equal halves (IR/UV), as argued in~\cite{finiteswt}, the modes at the lattice edges $\pm N/4$ must first be transformed by assigning their symmetric combination to the IR sector and their antisymmetric combination to the UV sector:
\begin{align}
\tilde{f}^{\text{IR}}_{-N/4}
  &= \frac{1}{\sqrt{2}} \left( \tilde{f}_{N/4} + \tilde{f}_{-N/4} \right),
  \notag\\
\tilde{f}^{\text{UV}}_{-N/4}
  &= \frac{1}{\sqrt{2}}\, i \left( \tilde{f}_{N/4} - \tilde{f}_{-N/4} \right).
\end{align}
\begin{figure*}[htb!]
	\centering
    \includegraphics[width=\textwidth]{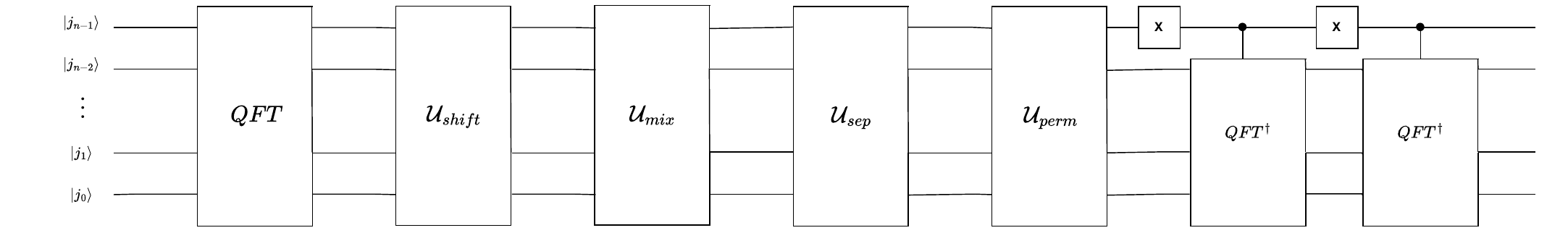}
	\caption{Implementation of QSWT in a quantum circuit on $n$ qubits using QFT, modular adders and 2-sparse unitaries}
	\label{fig:qswtqc}
\end{figure*}
The mixing of the edge modes can be realised in a quantum circuit through a $2$-sparse unitary matrix:
\begin{align}
\mathcal U_{mix} &=
\begin{bmatrix}
1 & 0 & 0 & \cdots & 0 & 0 & \cdots & 0 \\
0 & 1 & 0 & \cdots & 0 & 0 & \cdots & 0 \\
\vdots & \vdots & \ddots &  & \vdots & \vdots &  & \vdots \\
0 & 0 & \cdots & \tfrac{1}{\sqrt{2}}_{i,i} & \cdots & \tfrac{1}{\sqrt{2}}_{i,j} & \cdots & 0 \\
\vdots & \vdots &  & \vdots & \vdots & \vdots & & \vdots \\
0 & 0 & \cdots & \tfrac{i}{\sqrt{2}}_{j,i} & \cdots & -\tfrac{i}{\sqrt{2}}_{j,j} & \cdots & 0 \\
\vdots & \vdots &  & \vdots & \vdots & \vdots & \ddots & \vdots \\
0 & 0 & \cdots & 0 & 0 & \cdots & 1 & 0 \\
0 & 0 & \cdots & 0 & 0 & \cdots & 0 & 1 \\
\end{bmatrix}_{N\times N},
\end{align}
where $i = N/4$ and $j = 3N/4$. When this unitary matrix acts on the shifted momentum modes (obtained via $\mathrm{SHIFTQFT}$), it produces new momentum modes in which the edge modes of the IR/UV sectors are ``mixed''. We denote this unitary operation as $\mathcal U_{\mathrm{mix}}$.

Since the Shannon wavelet transform (SWT) implements perfect momentum separation, the next step is to divide the remaining momentum modes into two equal halves, with the modes indexed as:
\begin{align}
&\tilde{f}^{\mathrm{IR}}_q = \tilde{f}_q,\qquad\qquad\qquad
\tilde{f}^{\mathrm{UV}}_q = \tilde{f}_{q - \operatorname{sgn}(q)\,N/2}, \nonumber \\
q &\in \left\{ -\frac{N}{4} + 1, \ldots, \frac{N}{4} - 1 \right\}.
\end{align}
The unitary that divides the momentum modes into two equal halves can be realised in the quantum circuit through a cyclic permutation matrix. This matrix places its leading $1$ at column index $N/4$, with each subsequent row obtained by shifting one position to the right. We denote this operation by $\mathcal U_{\mathrm{split}}$.

The final step is to map the two halves back to real space using an inverse Fourier transform. Prior to this, the momentum modes must be reindexed by applying a block-diagonal unitary, $\mathcal U_{\mathrm{block}}$. With respect to the IR/UV partition, $U_{\mathrm{perm}}$ has two diagonal blocks of size $N/2 \times N/2$: the top block is a permutation matrix with its leading $1$ at column index $N/4$, while the bottom block contains the same permutation matrix, except that the entry at column index $3N/4$ is assigned a phase of $-1$. The implementation cost of $\mathcal U_{\mathrm{block}}$ is dominated by that of the underlying adder circuits (from \cite{addergidney}). 

Finally, controlled on the most-significant qubit (on both $\ket{1}$ and $\ket{0}$), we apply the $(N/2)$-dimensional inverse $\mathrm{QFT}$ to the remaining qubits, thereby transforming the IR and UV blocks separately back into the real-space basis. 

Therefore, implementing $\mathrm{QSWT}$ can be achieved by decomposing the operation into unitary transformations:
\begin{align}
\mathrm{QSWT}
  = \mathrm{QFT}^{\dagger}\, \mathcal U_{\mathrm{block}}\, \mathcal U_{\mathrm{split}}\, \mathcal U_{\mathrm{mix}}\, \mathcal U_{\mathrm{shift}}\, \mathrm{QFT}. \nonumber
\end{align}
If we apply this operation to an equal superposition of $N$ elements ($N$ lattice points), the basis states are transformed as:
\begin{align}
\mathrm{QSWT}\sum_{j=0}^{N-1} \ket{j}
  &= \sum_{j=0}^{N-1} S^{\mathrm{IR/UV}}_{2i-j} \ket{j},
\end{align}
where $S^{\mathrm{IR/UV}}_{2i-j}$ denotes the element in row $i$ and column $j$ of a rectangular block matrix, with the first $N/2$ rows corresponding to the IR transformation and the remaining rows corresponding to the UV transformation. The coefficients $S^{\mathrm{IR/UV}}$ are given in Eq.~(19) of~\cite{finiteswt}. 

Since $\mathrm{QSWT}$ is built from $\mathrm{QFT}$s together with at most $2$-sparse unitary transformations and modular adder circuits, the overall cost of implementing $\mathrm{QSWT}$ is dominated by the cost of $\mathrm{QFT}$, scaling as $\mathcal{O}\bigl(n^2\bigr)$. The circuit diagram for implementing $\mathrm{QSWT}$ on $n$ qubits, $\ket{j_0 j_1 \ldots j_{n-1}}$, is shown in Fig.~\ref{fig:qswtqc}.
\section{Constructing Efficient block-encoding of SLAC Laplacian}\label{sec: slaclap}
The SLAC derivative operator is a non-unitary, circulant, and dense matrix, which we implement in a quantum circuit using the block-encoding framework. Its circulant structure allows it to be written as a linear combination of right-shift permutation matrices, enabling a block-encoding through the $\mathrm{LCU}$ method, as widely studied for structured matrices~\cite{be5,be6}. The right-shift permutations serve as the unitaries in this decomposition and can be implemented efficiently with low $T$-gate cost using modular adder (or subtractor) circuits~\cite{addergidney}. The primary challenge in block-encoding SLAC operators lies in preparing the ancillary state with $\mathcal{O}(2^n)$ distinct amplitudes for an $n$-qubit system. 

To overcome this, we employ a hierarchical construction based on the nested-boxes inequality-test framework~\cite{nestedboxes,qft3}, which avoids the exponentially decreasing success probability of naive amplitude preparation schemes. In this approach, the ancillary register is partitioned coherently into dyadic intervals in $\log N$ steps, and comparator-based inequality tests are applied within each box to imprint the required amplitude scalings. This yields $\log N$-depth state preparation even in the presence of $N$ distinct amplitudes, providing an efficient alternative to Grover--Rudolph state preparation~\cite{domyuv}. The resulting construction achieves polylogarithmic circuit depth, and we can simultaneously bound the dominant truncation error $\varepsilon_{\mathrm{trunc}}$ that arises from approximating the infinite-lattice coefficients on a finite lattice. A detailed analysis of these $\mathrm{LCU}$-based block-encodings is given in the following subsections.

A natural alternative is to implement the SLAC operator via momentum-space diagonalisation and an inverse QFT. Although QFT is highly efficient ($O(n^2)$ gate complexity), we prioritize a direct position-basis block-encoding for two reasons. First, the SLAC operator admits a natural Linear Combination of Unitaries (LCU) decomposition using cyclic shifts, which are implementable with minimal gate complexity in the position basis. Second, practical applications—such as solving differential equations with local potentials—often break translational invariance. Keeping all operators native to the position basis, including the diagonal wavelet preconditioners, avoids the significant overhead of interleaving controlled-basis transformations within larger quantum algorithms.

We first consider the block-encoding of the SLAC Laplacian and subsequently show that the same methods extend to the first-order SLAC derivative. 
Henceforth, we set the lattice spacing to $a=1$.
Following~\cite{costella2002newproposalfermiondoubling}, the SLAC Laplacian in Eq.~\eqref{slaclco} can be optimally truncated by enforcing periodic boundary conditions on a finite lattice.
\begin{equation}\label{slacc}
\Delta_{j}^{(2)} = 
\begin{cases}
-\dfrac{\pi^2}{3} - \dfrac{2\pi^2}{3N^2}, & j = 0, \\[10pt]
\dfrac{2\pi^2\,(-1)^{1+j}}{N^2\,\sin^2\!\bigl(\tfrac{\pi j}{N}\bigr)}, & j \neq 0.
\end{cases}
\end{equation}
The coefficient at row $r$, column $c$ is given by 
$\Delta^{(2)}_{(r-c)\bmod N}$, depending only on their difference and thus 
manifesting translational invariance. Because the matrix is circulant, 
it can be expressed as an $\mathrm{LCU}$ over cyclic-shift permutation matrices. 
However, preparing amplitudes proportional to 
$\sqrt{\Delta^{(2)}_{(j)}}$ is hindered by the $\sin^{-2}(\pi j/N)$ factor.
To avoid this, we work in the infinite-lattice limit of Eq.~\eqref{slacc} 
and encode the Laplacian with coefficients:
\begin{equation}\label{slaccf}
\widetilde{\Delta}_{j}^{(2)} =
\begin{cases}
-\dfrac{\pi^2}{3}, & j = 0, \\[10pt]
\dfrac{2\,(-1)^{1+j}}{j^2}, & 1 \le j \le N/2 - 1, \\[10pt]
\dfrac{2\,(-1)^{1+j}}{(N-j)^2}, & N/2+1 \le j \le N-1.
\end{cases}
\end{equation}
The lower-right panel of Fig.~\ref{fig:slac_symbol_comparison} shows the truncation error for the SLAC Laplacian, obtained by comparing the Fourier symbol of the exact operator in Eq.~\eqref{slacc} with that of the truncated representation in Eq.~\eqref{slaccf}.

In the large-$N$ limit, the coefficient at the midpoint of the lattice, i.e.\ $\widetilde{\Delta}_{j=N/2}^{(2)}$ vanishes, and the coefficients of the first half 
of the sequence $(j=1,\dots,N/2-1)$ are mirrored in the second half $(j=N/2+1,\dots,N-1)$. We exploit this symmetry by preparing amplitudes only on the first half and then using a Hadamard to reproduce the second half. We denote the truncation error by $\varepsilon_{\text{trunc}}$, which we will show scales as $\mathcal{O}(1/N)$.
The SLAC Laplacian matrix for a lattice of $N$ sites can then be expressed as an $\mathrm{LCU}$:
\begin{equation}\label{lcuslacla}
\begin{aligned}
    & \widetilde\Delta^{(2)}_{\rm slac} = \sum_{j=0}^{N-1}\widetilde{\Delta}_{j}^{(2)}P^j = \sum_{j=0}^{N/2-1}\widetilde{\Delta}_{j}^{(2)}\Big(P^j+P^{N-j}\Big).
\end{aligned}
\end{equation}
Here, $P^j$ denotes the $j^{\text{th}}$ power of the right-shift cyclic (permutation) matrix of dimension $N$, with $P^0$ being the identity. These permutation matrices can be efficiently implemented using modular 
adder circuits, following the construction in~\cite{addergidney} and 
detailed in Appendix~\ref{app: sel}. However, the bottleneck of this $\mathrm{LCU}$
lies in preparing the fully dense set of $N$ amplitudes 
proportional to $\sqrt{\widetilde{\Delta}_{j}^{(2)}}$ on computational 
basis state $\ket{j}$.

In the next subsection, we outline how to prepare states with amplitudes proportional to $\sqrt{\widetilde{\Delta}_{j}^{(2)}}$ using nested-boxes inequality tests, as introduced in~\cite{nestedboxes}. We then construct 
the $\mathrm{LCU}$ from permutation matrices together with a sign oracle to assign the correct phases, analyse the complexity of the scheme, and derive error bounds for the full construction.

\subsection{State preparation of SLAC Laplacian using inequality tests}\label{subsec:preplap}
In this subsection, our goal is to prepare the unnormalised state
\begin{equation}
\sum_{j=0}^{N-1}\sqrt{\alpha_j}\ket{j},
\end{equation}
where $\alpha_j$ are determined by the unsigned coefficients in Eq.~\eqref{slaccf}.

To realise these amplitudes via inequality tests, we begin by initialising the quantum registers in the following state:
\begin{equation}\label{qreg}
\ket{[0]_{(n-1)}}_\mu\;\ket{[0]_{(n-1)}}_j\;\ket{0}_d\ket{[0]_{n_{ref}}}\ket{0}_{f}.
\end{equation}
Here, the notation $\ket{[0]_{n}}$ denotes an $n$-qubit all-zero string. Throughout this work, we use little-endian ordering for all binary encodings. The register $j$, together with the single-qubit register $d$, forms the $n$-qubit data register on which we prepare the $\mathrm{LCU}$ coefficients. The qubit $d$ is taken to be the most significant bit (MSB) of the integer $j$, so that $\mathrm{MSB}=0,1$ cleanly separates the two halves of the lattice. The register $\mu$, is an $(n-1)$-qubit unary register which will be used to prepare a conditional superposition on the data register. The ancilla $a$ distinguishes the $j=0$ branch from the $j\neq 0$ branch, while the flag qubit $f$ is used to test the inequality with respect to the $n_{ref}-$~qubit reference register.

The overall strategy is as follows. The register $\mu$, which is unary with equal weighting, acts as a placeholder that encodes $2^{\mu}$ values of the register $j$. By controlling on $\mu$, we superpose the values $j=2^{\mu},\dots,2^{\mu+1}-1$ through a specific set of nested conditions implemented using a ladder of $\mathrm{CX}$ gates. For brevity, we sometimes say that each box $B_\mu$ encodes a disjoint set of $2^{\mu}$ values of $j$. In this way, the registers $(\mu,j)$ jointly index the amplitudes in a hierarchical (logarithmic)  manner, allowing the preparation of $N$ distinct amplitudes in $\log N$ steps. These amplitudes are then mapped to the corresponding $(\mu,j)$ pairs via inequality tests. This inequality-test framework for state preparation avoids the resource-intensive Grover-based method (which uses controlled rotations), as discussed in~\cite{domyuv}. Thus, the nested-boxes structure for testing inequalities achieves low Toffoli cost and is well-suited to a fault-tolerant setting.
 
The step-by-step procedure to prepare the desired amplitudes in register $\ket{j}$ can be summarised as follows:

\begin{itemize}
\item An $R_y$ rotation is applied to the ancilla $a$ to split the $j=0$ and $j>0$ branches while assigning the appropriate relative amplitudes required by the $\mathrm{LCU}$ coefficients.

\item Using $\mathcal{O}(n)$ rotations, we then prepare a superposition over the register $\mu$ in unary encoding, where the basis state $\ket{[\mu]_{n-1}}$ contains exactly $\mu$ active ones. This efficiently assigns each state, $\ket\mu$ an amplitude scaling as $\sqrt{2^{-\mu}}$, establishing the required target probability distribution.

\item For each $\mu$, condition on the active qubits of the unary register and apply controlled Hadamard gates to the first $\mu$ qubits in register $j$, thereby creating an equal superposition over the range $j=0,\dots,2^\mu-1$.

\item Subsequently, the unary-encoded register $\mu$ is converted into a one-hot encoding by effectively incrementing its basis states by one. This is implemented using a cascade of $\mathrm{CX}$ gates.

\item Consequently, when a sequence of $\mathrm{CX}$ gates is applied from the $|\mu\rangle$ register to the index register $|j\rangle$, the resulting superposition is mapped directly onto the dyadic interval,
\begin{equation}\label{boxed}
B_{\mu} = \{\, j \mid 2^{\mu} \le j < 2^{\mu+1} \,\}, \quad \mu \ge 0.
\end{equation}

\item Finally, we impose the desired amplitudes via an inequality test, efficiently implemented with a comparator~\cite{domyuv}. 
\end{itemize}

\begin{figure*}[htb!]
	\centering
	\includegraphics[width=\textwidth]{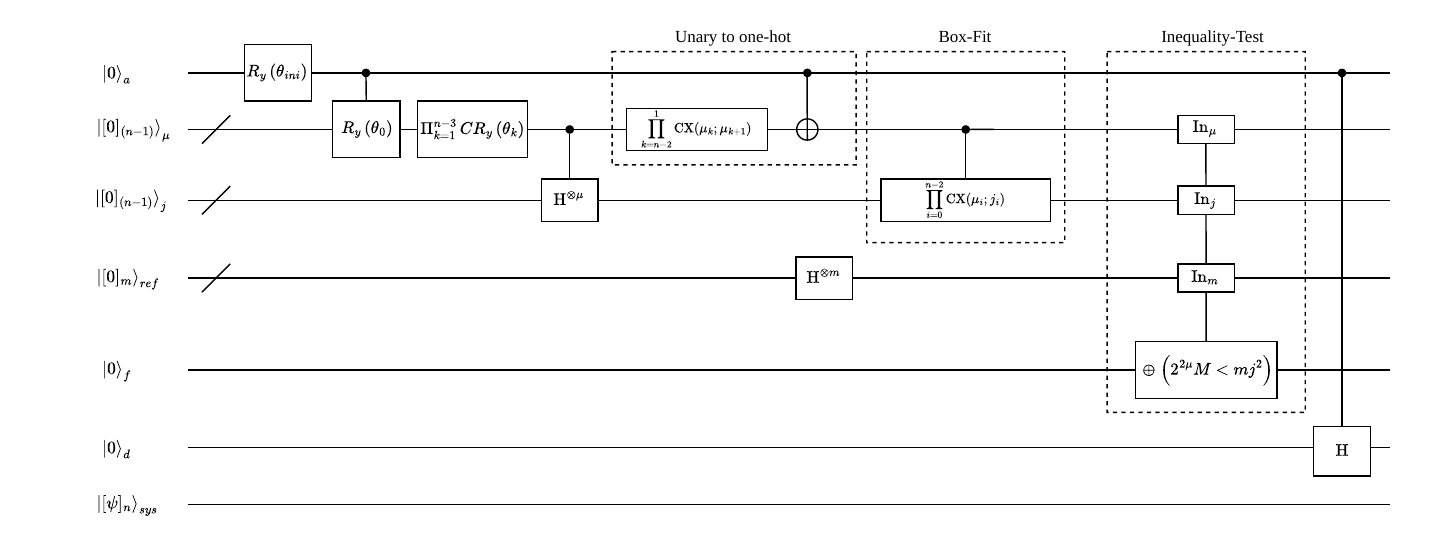}
    \caption{High-level circuit implementing the state-preparation protocol for the $\mathrm{LCU}$ of the SLAC Laplacian in~Eq.~\eqref{lcuslacla} using $O_{\mathrm{prep}}^{(2)}$. The procedure begins by applying an $R_y(\theta_{\mathrm{ini}})$ gate to the ancilla qubit $a$, where $\theta_{\mathrm{ini}}$ is chosen to ensure the amplitude ratio between the $j=0$ and $j>0$ branches matches the required $\mathrm{LCU}$ coefficients. Next, controlled rotations ($\mathrm{CROT}$) on the unary-encoded register $\mu$ are followed by controlled Hadamards to generate a conditional superposition on register $j$. The first dotted region converts $\mu$ into a one-hot encoding via a downward cascade of $\mathrm{CX}$ gates (controlled on qubit $\mu_k$, targeting qubit $\mu_{k+1}$ for $k = n-2,\dots,1$), terminating with a $\mathrm{CX}$ gate controlled on $a$ targeting the least-significant qubit of $\mu$. The second dotted region (\emph{Box-Fit}) constructs the boxes $B_\mu$ using a ladder of $\mathrm{CX}$ gates controlled on the $i^{\mathrm{th}}$ qubit of $\mu$, targeting the corresponding qubit of register $j$. Finally, the inequality test in~Eq.~\eqref{ineq1} is performed against an equal superposition in register $ref$ using a flag qubit $f$, and an $a=1$-controlled Hadamard is applied to the most significant qubit of $j$ (labelled $d$), yielding the prepared state in~Eq.~\eqref{fis}. The first-order SLAC derivative coefficients detailed in Eq.~\eqref{slac1cf} are prepared using an identical circuit framework, modified by omitting the initial stage, adjusting the rotation angles, and implementing the updated inequality criteria from Eq.~\eqref{ineqder1}.
}
\label{fig:Opreplap}
\end{figure*}

We now provide a complete description of the state-preparation protocol for the $\mathrm{LCU}$ block-encoding of the SLAC Laplacian. To partition the zero ($j=0$) and non-zero ($j>0$) branches of the index register, we initialise the ancilla qubit $a$ using an $R_y(\theta_{\mathrm{ini}})$ rotation. The angle $\theta_{\mathrm{ini}}$ is determined by the required ratio of the $\mathrm{LCU}$ coefficients defined in Eq.~\eqref{slaccf}. Since the non-zero branch ultimately splits into two equiprobable halves, its final target probability must be exactly twice that of the zero branch.
Taking into account the inequality-test step, the final success probabilities of these branches are given by $P_{j>0} = \sin^2(\theta_{\mathrm{ini}})\, p_{\mathrm{succ}}$ and $P_{j=0} = \cos^2(\theta_{\mathrm{ini}})$, where $p_{\mathrm{succ}} = \pi^2/12$ is the success probability of the inequality test (see Fig.~\ref{fig:crs}). Enforcing the required ratio $P_{j>0} / P_{j=0} = 2$ fixes the initial angle via $\tan(\theta_{\mathrm{ini}}) = \sqrt{24/\pi^2}$. Consequently, this initial rotation prepares the state:
\begin{equation}
    \left( \sqrt{\frac{\pi^2}{\pi^2+24}} |0\rangle_a + \sqrt{\frac{24}{\pi^2+24}} |1\rangle_a \right) |\bar{0}\rangle_{\mu,j}\,.
\end{equation}

We then prepare a geometric superposition over the $(n-1)$ qubits of the unary register $\mu$. To ensure that the $\ket{0}_a$ branch remains strictly in the vacuum state $\ket{\mu=0}$, the initial $R_y(\theta_0)$ rotation is conditioned explicitly on the ancilla state $\ket{1}_a$. Subsequent $R_y(\theta_k)$ rotations are applied to qubit $\mu_{k+1}$, each controlled on the preceding qubit $\mu_k = 1$. By applying this sequential cascade of $n-2$ controlled rotations, one can efficiently generate the following superposition:
\begin{align}
&\sqrt{\frac{12}{(\pi^2+24)(1-2^{-(n-1)})}} 
\sum_{\mu=0}^{n-2} \sqrt{2^{-\mu}}
\ket{[\mu]_{n-1}}\ket{[0]_{n-1}}\ket{1}_{a} \nonumber\\&+ \sqrt\frac{\pi^2}{\pi^2+24}\ket{\mu=0}\ket{j=0}\ket{0}_{a}.
\end{align}
To correctly distribute the probability amplitudes according to the finite geometric series $2^{-\mu}$, the rotation angles are chosen as:
\begin{equation}
\theta_k = 2\arccos\left( \frac{1}{\sqrt{2(1 - 2^{-(n-1-k)})}} \right)\,.
\end{equation}
Note that the registers $\mu$ and $j$ each consist of $n-1$ qubits; henceforth, we omit this subscript in the analytical expressions. Furthermore, we disregard the $\ket{0}_a$ branch, as this component of the state remains unaffected by the subsequent operations. 

Next, we generate an equal superposition over the basis states of the register $j$. This is achieved by applying controlled Hadamard gates to the $n-1$ qubits of $j$, the controls being determined by the value stored in the register $\mu$. As a result, the superposition created on $j$ is conditioned on the corresponding value of $\mu$, which yields the following state:
\begin{align}
&\sqrt{\frac{12}{(\pi^2+24)\;(1-2^{-(n-1)})}} 
   \sum_{\mu=0}^{n-2}\sqrt{2^{-2\mu}}\ket{\mu}\sum_{j=0}^{2^{\mu}-1}\ket{j}\ket1_a.
\end{align}
Next, we apply a cascade of $\mathrm{CX}$ gates, $\prod_{k=n-3}^{0} \mathrm{CX}(\mu_k;\mu_{k+1})$, to the qubits of the register $\mu$, where each gate is controlled on the $k^{\mathrm{th}}$ qubit and targets the $(k+1)^{\mathrm{th}}$ qubit. The cascade is applied in descending order, starting from $k=n-3$.
After this sequence, an additional $\mathrm{CX}$ gate is applied to the least significant qubit (LSB) of the register $\ket{\mu}$, controlled on the
ancilla qubit $a$. This ensures that the $a=0$ branch continues to correspond to $\mu=0$ in the binary representation. Following these operations, the resulting superposition over the bitstrings $[\mu]_{n-1}$ contains exactly one active qubit located at position $\mu$, thereby realising a one-hot encoding of the register $\mu$. Note that, despite this change in representation, the overall quantum state remains unchanged from the previous step.                                      

Then a sequence of $\mathrm{CX}$ gates, $\prod_{i=0}^{n-2} \mathrm{CX}(\mu_i; j_i)$ is applied to the register $j$
where each gate is controlled on the $i^{\mathrm{th}}$ qubit of the register $\mu$ and targets the corresponding qubit of the register $j$. These gates are applied sequentially for $i = 0, \dots, n-2$. Following this sequence, on the $a=1$ branch the conditional superposition in
register $j$ occupies the interval $2^{\mu} \le j < 2^{\mu+1}$. Consequently, the components with $j \neq 0$ are organised into the nested boxes $B_{\mu}$ defined in Eq.~ \eqref{boxed}. The overall quantum state can therefore be written as:
\begin{align}\label{boxfitlap}
&\sqrt{\frac{12}{(\pi^2+24)\;(1-2^{-(n-1)})}} 
   \sum_{\mu=0}^{n-2}\sqrt{2^{-2\mu}}\ket{\mu}\sum_{j \in B_{\mu}}\ket{j}\ket{1}_a \nonumber \\&  + \sqrt\frac{\pi^2}{\pi^2+24}\ket{\mu=0}\ket{j=0}\ket{0}_{a}\,.
\end{align}

For the components labelled by the nested boxes $B_\mu$, with $j \in B_\mu$, an inequality test is then applied to the $(n-1)$-qubit registers $j$ and $\mu$. This procedure introduces a fresh flag qubit $f$, allowing postselection on the joint success branch defined by $f = 0$. The inequality is evaluated with respect to an equal superposition prepared on a separate $n_{ref}$-qubit ancilla reference register,
\[
\sum_{m=0}^{M-1} \ket{[m]_{n_{ref}}}.
\]
For each value of $\mu$, an overall weighting factor proportional to $1/j$ is applied to the components of the register with $j \neq 0$ by implementing the following inequality test, conditioned on the joint state of the registers $\mu$ and $j$:
\begin{equation}\label{ineq1}
2^{2  \mu}M > {m}\;j^2.
\end{equation}
The inequality is chosen so that the ratio $m/{M}$ is strictly less than $1$. This guarantees that the prepared amplitudes are exactly proportional to the desired target coefficients, and ensures that the success probability of the target state approaches a bounded asymptotic value as $n$ becomes large.

To test this inequality, we first compute and then compare the left-hand side ($\mathrm{LHS}$) and the right-hand side ($\mathrm{RHS}$). Because the integer $M$ is chosen as a power of $2$, the $\mathrm{LHS}$ value of $2^{2\mu}M$ can be computed in place with no quantum gates by simply relabeling the qubits. As explained in Section II.C of \cite{Su:2021lut}, we can interleave zeroed qubits between the qubits of the one-hot representation of $\mu$ to yield $2\mu$. This unary representation of $2\mu$ intrinsically acts as a binary representation of $2^{2\mu}$. For example, if we have a five-qubit one-hot string for $\mu=2$, given by $00100$, interleaving zeros to the left of each bit yields $\underline{0}0\underline{0}0\underline{0}1\underline{0}0\underline{0}0$. This new string represents $2\mu=4$ in unary, but evaluates exactly to $2^{2\mu}=16$ in standard binary. Lastly, to multiply by the factor $M$, we pad the least significant bits with $n_{ref}=\log_2 M$ zeros.  In practice, we do not need to actually include these zeroed ancillae and we can simply relabel the qubits of $\mu$. Consequently, the non-Clifford cost of the inequality test depends entirely on the RHS computation.

The $\mathrm{RHS}$ of the inequality, $mj^2$, is obtained by first computing $j^2$ and then multiplying the result by $m$. Since the index register $j$ contains $n-1$ qubits, Lemma 7 of \cite{Su:2021lut} gives a Toffoli cost of $(n-1)^2-(n-1)$ for evaluating $j^2$. Moreover, Lemma 10 of \cite{Su:2021lut} shows that the bitwise multiplication $m \times j^2$ can be implemented with at most $4(n-1)n_{\mathrm{ref}}-n_{\mathrm{ref}}$ Toffoli gates. Therefore, the total Toffoli cost of evaluating the right-hand side is $(n-1)^2-(n-1)+4(n-1)n_{\mathrm{ref}}-n_{\mathrm{ref}}$. For the qubit cost, the squaring step uses $2(n-1)$ scratchpad qubits to store the value of $j^2$. These same $2(n-1)$ qubits are then reused in the multiplication step, so $j^2$ does not need to be uncomputed beforehand. The multiplication requires a further $n_{\mathrm{ref}}$ ancilla qubits. Hence, the total number of ancilla qubits required for computing the $\mathrm{RHS}$ is $2(n-1)+n_{\mathrm{ref}}$.

Finally, after testing the inequality if we apply a Hadamard gate to the most significant qubit $d$, controlled on ancilla $a=1$, the state from Eq.\eqref{boxfitlap} becomes: 
\begin{align}\label{fis}
&\sqrt\frac{\pi^2}{\pi^2+24}\ket{\mu=0}\ket{j=0}\ket{0}_{a;d}\ket{0}_f\ket{0}_{ref} \notag \\
&+ \sqrt{\frac{12}{(\pi^2+24)\;(1-2^{-(n-1)})}}\Biggl\{\sum_{\mu=0}^{n-2} \sum_{j  \in B_\mu} \sum_{m=0}^{Q-1} 
\frac{2^{-\mu}}{\sqrt{{M}}}\ket{\mu}\ket{j}\notag\\&\otimes\ket{m}_{ref}\ket{1}_a\ket{0}_f\ket{+}_d \Biggr\} + |\Psi^\perp\rangle.
\end{align}                                                     
In this expression, $\ket{\Psi^\perp}$ captures the orthogonal failure state of the inequality test. The success branch is governed by the $Q$ specific values of $m$ that satisfy the inequality:
\begin{equation}\label{Q}
\begin{aligned}
 Q = \Bigg\lceil {\frac{2^{2\mu}\; M\;}{j^2}}\Bigg\rceil\;\;.
\end{aligned}
\end{equation}                                                                                     

In the large-${M}$ limit, the ceiling function can be safely neglected. However, working with a finite ${M}$ introduces a state preparation error, denoted $\varepsilon_{\mathrm{prep}}$, which we analyse in the following subsection. The resulting state in large-$M$ limit is shown in Eq.\eqref{app: O_prep}. Finally, the total probability of successfully post-selecting the flag state $f = 0$ is evaluated by summing the squared amplitudes of the target state, yielding:
\begin{align}\label{psucclap}
p^{(2)}
&=  \frac{\pi^2}{\pi^2+24}+ {\frac{12}{(\pi^2+24)\;(1-2^{-(n-1)})}}\notag\\
&\times\Biggl\{ \sum_{\mu=0}^{n-2}
     \sum_{j \in B_\mu}\frac{2^{-2\mu}}{M}\Bigg\lceil {\frac{2^{2\mu}\; M\;}{j^2}}\Bigg\rceil\Biggr\} \notag\\
p^{(2)}_{\;n\gg1} &= \frac{\pi^2}{\pi^2+24} + \frac{12}{\pi^2+24}
     \sum_{j=1}^{2^{n-1}-1} \frac{1}{j^2}\notag\\
&\approx \frac{3\pi^2}{\pi^2+24}= 0.874.
\end{align}
In the large-$n$ limit, we approximate the summation over the index $j$ using the Basel series, $\sum_{j=1}^{\infty} 1/j^2 = \pi^2/6$. Note that the entire setup of preparing nested boxes and subsequently testing the inequality to realise the desired $1/j$ amplitude on the state $\ket{j}$ can be viewed within the framework of quantum rejection sampling \cite{Lemieux:2024pmt}, as described in Fig.~\ref{fig:crs}. The initial rotation step assigns a probability of $\frac{24}{\pi^2+24}$ to the branch $j\ne0$, while the acceptance probability from Fig.~\ref{fig:crs} is $\frac{\pi^2}{12}$. Hence, as a final check, to find the overall success probability of preparing the state $\ket{j}$, we take the product of these two probabilities, yielding $\frac{2\pi^2}{\pi^2+24}$, which is exactly the same as that in Eq.~\eqref{psucclap}. 

\begin{figure}[htb!]
           \includegraphics[width=\columnwidth]{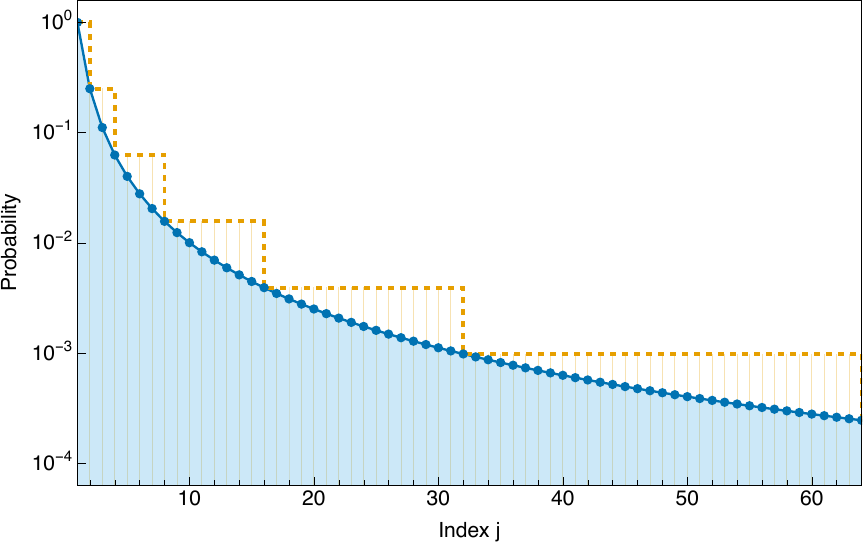}
		\caption{Probability per state $j$ within box $B_\mu$ during the quantum rejection sampling protocol. We first prepare a proposal state with uniform probabilities, $2^{-2\mu}$ within each box $B_\mu$, forming the orange dashed step-function envelope. The blue points indicate the target $1/j^2$ probability distribution. The inequality test in Eq.~\eqref{ineq1}, effectively rejects the probability mass corresponding to the unshaded region between the orange envelope and the blue curve. The overall success probability is the ratio of the target (blue) area ($\sum_{j=1}^{\infty} 1/j^2 = \pi^2/6$) to the area under the proposal envelope ($\sum_{\mu=0}^{n-2} 2^{-\mu} = 2$ for $n \gg 1$), yielding an asymptotic acceptance probability of $\pi^2/12$.}
		\label{fig:crs}
\end{figure}
We now analyse the complexity of each step in the state-preparation procedure required for the $\mathrm{LCU}$ block-encoding of the SLAC Laplacian.
\begin{itemize}
    \item The step, which prepares the superposition over the register \(\mu\) with amplitudes proportional to \(\sqrt{2^{-(\mu-2)}}\), can be implemented using $(n-2)-$ controlled  $R_y-$ rotation gates.
    \item The next step generates an equal superposition over the register $j$, conditioned on the unary-encoded register $\mu$. This is achieved by applying controlled-Hadamard gates, triggered by $\mu$, to the $n-1$ qubits of $j$. As detailed in~\cite{Su:2021lut}, utilizing a catalytic state allows each controlled-Hadamard to be synthesized with exactly one Toffoli gate. Consequently, the total Toffoli cost for this state-preparation step is strictly $n-1$.
    \item The unary-encoded register \(\mu\) is converted into a one-hot encoding using a sequential cascade of \(\mathcal{O}(n)\) \(\mathrm{CX}\) gates.
    \item The crucial step of preparing the superposition in register \(j\), conditioned in a box \(B_\mu = \{2^{\mu} \le j < 2^{\mu+1}\}\), can be achieved using  a series of \(\mathcal{O}(n)\) \(\mathrm{CX}\) gates triggered by the one-hot \(\mu\) register.
    \item The subsequent step, preparing an equal superposition over the reference register \(\ket{m}\), involves \(\mathcal{O}(\log M) \equiv \mathcal{O}(n_{ref})\) Hadamard gates, where \(n_{ref}\) is the number of qubits in the register.
    \item As established, the inequality can be tested using a $(2(n-1)+n_{\mathrm{ref}})$-bit comparator, which has cost $2(n-1)+n_{\mathrm{ref}}$. Furthermore, the Toffoli cost of computing the RHS $mj^2$ is $(n-1)^2-(n-1)+4(n-1)n_{\mathrm{ref}}-n_{\mathrm{ref}}$. Therefore, the total cost of squaring, multiplying, and then testing the inequality is $(n-1
    )^2+(n-1)+4(n-1)n_{\mathrm{ref}}$.
\end{itemize}
This concludes the state-preparation stage of the $\mathrm{LCU}$ construction for block-encoding the SLAC Laplacian. The high-level circuit diagram detailing this procedure is depicted in Fig.~\ref{fig:Opreplap}.
\subsection{Block-encoding using the LCU}\label{subsec:analyticallapbe}
We first recall the standard $\mathrm{LCU}$ construction for a matrix $M$: 
\newtheorem{df1}{Definition}
\begin{df1}\label{lm1}
Let $M = \sum_{j=0}^{N-1} \alpha_j W_j$ be a linear combination of unitaries $W_j$ with $\alpha_j \geq 0$ for all $j$ and $\sum_{j=0}^{N-1} \alpha_j = 1$.
Let $O_\alpha$ be any operator that satisfies 
$O_\alpha \ket{0^n} = \sum_{j=0}^{N-1} \sqrt{\alpha_j} \ket{j}$, 
where $n$ is the number of qubits used to represent $\ket{j}$, and define 
$\mathrm{select}(W) = \sum_{j=0}^{N-1} \ket{j} \bra{j} \otimes W_j$.
Then,
\begin{equation}
\begin{aligned}
    & \bigl(O_\alpha^\dagger \otimes I\bigr) \, \mathrm{select}(W) \, \bigl(O_\alpha \otimes I\bigr) \ket{0^n} \ket{\psi}\\
    & = \ket{0^n} M \ket{\psi} + \ket{\Psi^\perp}\quad\mathrm{where}\;\;\;\;\bigl(\ket{0^n} \bra{0^n} \otimes I\bigr) \ket{\Psi^\perp} = 0.
\end{aligned}
\end{equation}
\end{df1}
In our setting, the matrix $M$ is a non-unitary circulant matrix, and we therefore require additional ancilla registers to obtain a block-encoding of the SLAC Laplacian using the oracles $O_{\mathrm{prep}}$, $O_{\mathrm{sgn}}$, and $O_{\mathrm{select}}$. In particular, relative to Definition~\ref{lm1} we include the registers defined in~Eq.~\eqref{qreg}; projecting these ancilla registers onto the all-zero state yields the desired action of the SLAC Laplacian.
Accordingly, we define the unitary implementing the SLAC Laplacian construction as:
\begin{equation}
    U_\mathrm{slac}^{(2)} = O_\mathrm{prep}^{(2)\dagger} O_\mathrm{cpy}^{(2)}O_\mathrm{select}^{(2)}O_\mathrm{sgn}^{(2)}O_\mathrm{prep}^{(2)}.
\end{equation}

\begin{figure*}[htb!]
\centering

\begin{minipage}[t]{0.48\textwidth}
    \centering
    \vspace{0pt} 
    \subfloat[]{
        \includegraphics[width=\linewidth, height=3.0cm, keepaspectratio]{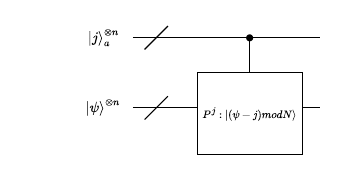}
    }
\end{minipage}
\hfill
\begin{minipage}[t]{0.48\textwidth}
    \centering
    \vspace{14.85pt} 
    \subfloat[]{
        \includegraphics[width=8.2 cm, height=3.5cm, keepaspectratio]{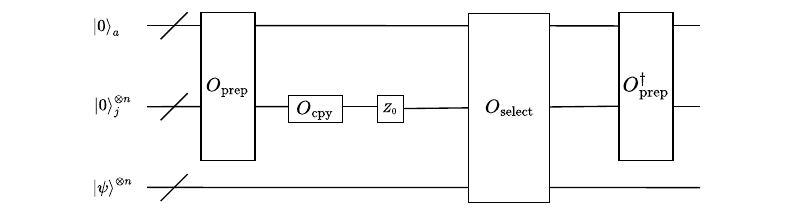}
    }
\end{minipage}

\caption{Complete $\mathrm{LCU}$ architecture showing:
(a) The SELECT operator $O_{\mathrm{select}}$, which implements the permutation $P^{j}$ via a modular subtractor, mapping the system state $\ket{j}_a\ket{\psi}$ to $\ket{j}_a P^{j}\ket{\psi}$; ancillary registers used for inequality tests are omitted; see Appendix~\ref{app: sel}.
(b) The full $\mathrm{LCU}$ block structure composed of $O_{\mathrm{prep}}$, $O_{\mathrm{sgn}}$, $O_{\mathrm{select}}$, $O_{\mathrm{cpy}}$, and $O_{\mathrm{prep}}^{\dagger}$, which implements the block-encoding of the SLAC Laplacian in Eq.~\eqref{beslac2} upon postselection of all ancillae in the state $\ket{\bar{0}}$.
For the SLAC Laplacian, the operator $O_{\mathrm{sgn}}$ is realised as a $Z$ gate acting on the least significant qubit of the register $j$, whereas for the first-order SLAC derivative, see Eq.~\eqref{beslac1}, the corresponding phase operator $O_{\mathrm{phase}}$ is used; see Appendices~\ref{app: O_sgn} and~\ref{app: phase}. The operator $O_{\mathrm{cpy}}$, implemented with $\mathrm{CX}$ gates, marks the success and failure branches; see Appendix~\ref{app: phase}.}
\label{fig:lcu-architecture}
\end{figure*}

The preparation step of~Eq.~\eqref{fis} prepares the computational basis state $j$ with the desired $1/j$ amplitudes, and its action is described in Appendix~\ref{app: O_prep}. Since the prepared state also includes undesirable branches, the copy oracle marks these by copying the flag values to a marker qubit, as shown in Appendix~\ref{app:cpy}. Because the matrix has coefficients with alternating signs, we also provide a sketch of the implementation of the $\mathrm{SIGN}$ oracle in Appendix~\ref{app: O_sgn}. The $\mathrm{SELECT}$ circuit, described in Appendix~\ref{app: sel} and shown in Figure~\ref{fig:lcu-architecture}, implements the cyclic shift matrices for each computational basis state, thereby capturing the circulant structure of the Laplacian. Finally, applying $U^{(2)}_{\mathrm{slac}}$ to the ancilla and main registers and postselecting the ancilla registers in the state $\ket{\bar{0}}$ yields:

\begin{widetext}
\begin{equation}\label{slac2be}
\begin{aligned}
&\Bigl(\bra{[0]_{(n-1)}}\bra{[0]_{n}}\bra{0}_f\bra{0}_a\bra{[0]_{m}}\bra{0}_C\otimes I_n \Bigr) 
U_\mathrm{\mathrm{slac}} ^{(2)}
\Bigl(\ket{[0]_{(n-1)}}\ket{[0]_{n}}\ket{[0]_{m}}\ket{0}_a\ket{0}_f\ket{0}_C\otimes I_n \Bigr)\ket{\psi}_{\text{sys}} \\
&= \Bigl(\bra{[0]_{(n-1)}}\bra{[0]_{n}}\bra{\bar0}_f\bra{0}_a\bra{[0]_{m}}\bra{0}_C
\otimes I_n \Bigr) 
O_{\text{prep}}^\dagger\Biggl\{-\sqrt{\frac{\pi^2}{\pi^2+24}}\ket{\mu=0}\ket{j = 0}\ket{0}_a\otimes P^0_{sys}+\sqrt\frac{12}{\pi^2+24}
      \sum_{\mu=0}^{n-2}\ket{\mu}\\
      & \sum_{j \in B_{\mu}}\frac{(-1)^{1+j}}{j}\ket{j}\otimes\biggl(\ket{0}_d\otimes P^{j}_{sys}+\ket{1}_d\otimes P^{N-j}_{sys}\biggr)\Biggr\}\ket0_a\ket{0}_{f}\ket{0}_C\ket{\psi}_{sys}+ \sqrt{1-p^{(2)}}\;\ket\omega_{fail}\ket{1}_{f}\ket{1}_C\ket{\psi}_{sys}\\
&=  \frac{3}{\pi^2+24}\Biggl\{\frac{-\,\pi^{2}}{3}\otimes P^0_{sys}
+ \sum_{j=1}^{N/2-1}\frac{(-1)^{1+j}\;2}{j^2}\;\otimes P^j_{sys}
+ \sum_{j=N/2+1}^{N-1}\frac{(-1)^{1+j}\;2}{j^2}\;\otimes P^{N-j}_{sys}
\Biggr\}\ket{\psi}_{\text{sys}} \\
&= \frac{3}{\pi^2+24}\;
\tilde{\Delta}^{(2)}_{\mathrm{slac}} \;\;\ket{\psi}_{\text{sys}}.
\end{aligned}
\end{equation}
\end{widetext}

Consequently, the block-encoding of the second-order SLAC derivative (the Laplacian) is obtained with a constant rescaling factor:
\begin{equation}\label{scf}
\alpha^{(2)} = \frac{\pi^2+24}{3} \approx 11.29. 
\end{equation} 

\subsubsection{Error Analysis of the block-encoding}
In this subsection, we bound the overall error of the block-encoding, which in turn allows us to bound the number of qubits required in the reference register for the inequality test, i.e.\ $n_{ref}$. There are two main sources of error in the state‐preparation protocol that propagate through the $\mathrm{LCU}$. Let $\varepsilon_{\mathrm{trunc}}$ denote the error due to encoding the large-$N$ limit of the finite–lattice coefficients--that is, preparing a state that yields the coefficients in
Eq.~\eqref{slaccf} rather than those in Eq.~\eqref{slacc}. This error can be bounded by the
operator norm of the difference matrix $\Delta^{(2)}_\mathrm{slac} \;-\;\tilde\Delta^{(2)}_\mathrm{slac},$
between these two operators. Since the matrix is circulant (and normal), its operator norm
equals its maximum eigenvalue. We use the Fourier Transform, given by:
\begin{equation*}
  F \;=\;\frac{1}{\sqrt{N}}\,\omega^{-jk},
  \quad
  \omega \equiv e^{2\pi i/N},
\end{equation*}
where $\omega$ is the $N^{\mathrm{th}}$ root of unity. It is straightforward to verify that, upon
diagonalising the Laplacian with the QFT, the diagonal entries (eigenvalues) are given by the Fourier
transform of the first row, and the maximum eigenvalue occurs at mode $k = N/2$.
We use this to bound the operator‐norm error of the SLAC Laplacian as shown below:
\begin{widetext}
    \begin{align} \label{slaclaperror}
&\Bigl\|
 \alpha \bigl(\bra0^a\otimes I_s\bigr)\,U^{(2)}_\mathrm{slac}\,\bigl(\ket0^a\otimes I_s\bigr)
 -{\Delta}^{(2)}_{\mathrm{slac}}
\Bigr\| \;=   
\Bigl\|
 \tilde{\Delta}^{(2)}_{\mathrm{slac}}
 -{\Delta}^{(2)}_{\mathrm{slac}}
\Bigr\| = \Bigl\|F^{\dagger}\tilde{D}F-F^{\dagger}{D}F\Bigl\| \;                                                 \notag
\end{align}
$\textrm{where}\;\; D=    \sum_{j=0}^{N-1}{\Delta}^{(2)}(j)\;e^{\frac{-i2\pi jk}{N}} \ket k \bra k \;\;\textrm{as}\;\;     \Delta^{(2)}_{\mathrm{slac}}=\sum_j \Delta^{(2)}(j)P^j $

\begin{align}
\therefore \;\;\Bigl\|
 \tilde{\Delta}^{(2)}_{\mathrm{slac}}
 -{\Delta}^{(2)}_{\mathrm{slac}}
\Bigr\|&=   max_k\left\{\sum_j\Bigl(\tilde\Delta^{(2)}(j)-\Delta^{(2)}(j)\Bigr)e^{\frac{-i2\pi jk}{N}}\right\} \notag \\
&=  \left\{\sum_j\Bigl(\tilde\Delta^{(2)}(j)-\Delta^{(2)}(j)\Bigr)e^{\frac{-i2\pi jk}{N}}\right\}_{k=N/2}\notag \\
\tilde\Delta^{(2)}(j)-\Delta^{(2)}(j) & = (-1)^{1+j}\left\{\frac{2}{j^2}-\dfrac{2\pi^2\,}{N^2\,\sin^2\!\Bigl(\tfrac{\pi j}{N}\Bigr)} + \frac{2\pi^2}{6n^2}\right\}  \notag \\
&\le (-1)^{1+j}\Biggl(\frac{2}{j^2} + \frac{2\pi^2}{3N^2} + \frac{2\pi^4j^2}{15N^4}- \frac{2}{j^2} - \frac{2\pi^2}{3N^2} \Biggr)\notag\\
&\le (-1)^{1+j}\Biggl(\frac{2\pi^4j^2}{15N^4} \Biggr)\notag\\
 \implies \;\;\Bigl\|
 \tilde{\Delta}^{(2)}_{\mathrm{slac}}
 -{\Delta}^{(2)}_{\mathrm{slac}}
\Bigr\|&\le \left|\sum_{j=1}^{N-1} (-1)^{1+j}\frac{2\pi^4j^2}{15N^4}  \exp\Biggl({\frac{-i2\pi j (N/2)}{N}}\Biggr)\right|= \left|-\frac{4\pi^4}{15N^4}\int_{j=1}^{N/2-1} j^2dj \right|\notag \\
& \therefore \;\;\Bigl\|\tilde{\Delta}^{(2)}_{\mathrm{slac}} -{\Delta}^{(2)}_{\mathrm{slac}}\Bigr\| = \mathcal{O}\left(\frac{1}{N}\right)\,.
\end{align}
\end{widetext}
Hence, the finite-lattice truncation error can be bounded as
\[
\varepsilon_{\mathrm{trunc}}
\;\le\; 
\bigl\|\Delta^{(2)}-\widetilde{\Delta}^{(2)}\bigr\|
\;=\; 
\mathcal{O}\!\left(\frac{1}{N}\right).
\]
The $\mathcal{O}(1/N)$ scaling of the truncation error is corroborated numerically: plotting the operator-norm difference between the two matrices for varying system size $N$ (see Fig.~\ref{fig:slac1infi}) confirms this expected behaviour.
This contribution is intrinsic to the finite lattice and is independent of the quantum algorithm. 
We denote the target total precision by $\varepsilon$ and require the total precision to satisfy
\[
\varepsilon_{\mathrm{trunc}}
+
\varepsilon_{\mathrm{prep}}
+\cdots
\le
\varepsilon .
\]
Here, $\varepsilon_{\mathrm{prep}}$ arises from using a large but finite value of $M$ and from neglecting the ceiling functions in the inequality test in Eq.~\eqref{Q}. The ellipsis represents the remaining algorithmic errors, such as those resulting from the block-encoding approximation and the QLSA.

Moreover, the lattice size is chosen so that
\[
\varepsilon_{\mathrm{trunc}}=\mathcal{O}(\varepsilon),
\]
which requires
\[
N=\Omega\!\left(\frac{1}{\varepsilon}\right).
\]

The error, $\varepsilon_\mathrm{prep}$ is upper bounded by the squared sum of the relevant amplitudes in Eq.~\eqref{fis}, i.e.,
\begin{align}
\varepsilon_{\rm prep}
&< 
\frac{12}{M\,(\pi^2+24)\,(1-2^{-(n-1)})}
\sum_{\mu=0}^{n-2}\sum_{j\in B_\mu}2^{-2\mu}
\notag\\
&=
\mathcal{O}\!\left(\frac{1}{M}\right).
\end{align}
Requiring this contribution to be of order $\varepsilon$ gives
\[
M=\Omega\!\left(\frac{1}{\varepsilon}\right).
\]
Since the lattice size is chosen as $N=\Omega(1/\varepsilon)$, we may take $M=\Theta(N)$. 
The reference register therefore requires
\[
n_{\mathrm{ref}}
=
\Theta(\log M)
=
\Theta\!\left(\log\frac{1}{\varepsilon}\right)
=
\Theta(\log N)
=
\Theta(n)
\]
qubits.

This establishes that we can define an efficient block-encoding for the SLAC Laplacian. Before doing so, we first recall the formal definition of a block-encoding of a non-unitary matrix $A$, as given below.
\begin{df1}\label{lm2}
Formally, the $(n + a)$-qubit unitary $U$ is an $(\alpha, a, \varepsilon)$-block-encoding of $A$ if:
\begin{equation}
    \Bigl\| A  - \alpha \bigl( \bra{0^a} \otimes I_n \bigr)\;U\; \bigl( \ket{0^a} \otimes I_n \bigr) \Bigr\| \leq \varepsilon.
\end{equation}
\end{df1}
\newtheorem{result}{Result}
Furthermore, a block-encoding is said to be efficient if it meets the criteria in the following definition.
\begin{df1}\label{df1}\cite{llyod}. Given a matrix $A \in \mathbb{C}^{N \times N}$, an $(\alpha, a, \varepsilon)$-block-encoding $U$ of $A$ is good if and only if the ratio 
$\alpha / \|A\| = \mathcal{O}\bigl(\mathrm{polylog}\left(N\right)\bigr)$, and additionally, $a = \mathcal{O}\bigl(\mathrm{polylog}\left(\tfrac{N}{\varepsilon}\right)\bigr)$ and the circuit depth of $U$ is $\mathcal{O}\bigl(\mathrm{polylog}\left(\tfrac{N}{\varepsilon}\right)\bigr)$. The block-encoding is called optimal if $\alpha = \Theta\bigl(\|A\|\bigr)$.
\end{df1}
With this definitions in place, we now state Proposition~\ref{corl1}, which gives an optimal $\mathrm{LCU}$-based construction of an $\varepsilon$-approximate block-encoding of the $n$-qubit SLAC Laplacian.
\newtheorem{prop}{Proposition}
\begin{prop}\label{corl1}
An $\varepsilon$-approximate block-encoding of the $n$-qubit SLAC Laplacian can be constructed within the $\mathrm{LCU}$ framework by combining the inequality-test--based state-preparation procedure in~Eq.~\eqref{eq:prep2} with the modular-subtractor–based $\mathrm{SELECT}$ step in~Eq.~\eqref{sel2}. This construction has gate complexity of $\mathcal{O}(n^2)$ and requires $\mathcal{O}(n)$ ancilla qubits. Using Definition~\ref{lm2}, the resulting implementation defines a $((\pi^2+24)/3,\,\mathcal{O}(n),\,\varepsilon)$ optimal block-encoding of the SLAC Laplacian, given by
\begin{widetext}
  \begin{equation}\label{beslac2}
\Biggl\| \tilde{\Delta}^{(2)} - \alpha^{(2)}
\Bigl(\bra{[\mu]_{(n-1)}}\bra{[0]_{n}}\bra{[0]_{m}}\bra{0}_C\bra{\bar0}_f\otimes I_n \Bigr) 
U_{\mathrm{slac}}^{(2)} 
\Bigl(  \ket{[\mu]_{(n-1)}}\ket{[0]_{n}}\ket{[0]_{m}}\ket{0}_C\ket{\bar0}_f \otimes I_n \Bigr) 
\Biggr\| \le \varepsilon.
\end{equation}  
\end{widetext}
\end{prop}
\begin{proof}
\normalfont
The block-encoding of the $n$-qubit SLAC Laplacian follows the LCU construction described in Definition~\ref{lm1}. As established in Result~\ref{lem:prep}, the state-preparation routine has a gate complexity of $\mathcal{O}(n^2)$, while Result~\ref{lem:sel} proves that the multi-controlled subtractor-based SELECT oracle (see Appendix \ref{app: sel}) similarly incurs a cost of $\mathcal{O}(n^2)$. Auxiliary COPY and SIGN oracles (Appendices~\ref{app:cpy} and~\ref{app: O_sgn}) introduce only negligible overhead to the total depth. 

The arithmetic logic for the intermediate square and inequality test adds a Toffoli cost of $(n-1)^2+(n-1)+4(n-1)n_{ref} = \mathcal{O}(n^2)$. Regarding spatial resources, the scratchpad registers (for the inequality test) require $2(n-1) + n_{ref}$ qubits, resulting in a total ancilla requirement of $4n + 2n_{ref} \equiv \mathcal{O}(n)$. 

Finally, the efficiency of the construction is verified by the ratio of the rescaling factor to the operator norm; with $\alpha^{(2)} \approx 11.29$ and $\|\tilde{\Delta}^{(2)}\| \leq \pi^2$, the ratio $\alpha^{(2)}/\|\tilde{\Delta}^{(2)}\| $ is a constant independent of $n$. According to Definition~\ref{lm2}, this yields an optimal block-encoding of the SLAC Laplacian.
\end{proof}

\section{Block-encoding of the first order SLAC derivative }\label{sec: slac1be}
As shown in~\cite{slac2}, the first-order SLAC derivative truncated on a periodic lattice of $N$ sites is translationally invariant, with matrix elements $\Delta^{(1)}_{j-k}$ at row $j$ and column $k$, repeating periodically and given by the following expression:
\begin{equation}\label{slacfo}
\Delta_{j}^{(1)}=\left\{
\begin{array}{ll}
\frac{-i\pi}{N}, & \text{if } j = 0 \\[10pt]
\frac{\pi (-1)^{1+j}}{N}\Big(\frac{\cos({\pi j / N})}{\sin({\pi j / N})} + i\Big), & \text{if } j \neq 0\,.
\end{array}
\right.
\end{equation}
Notice that at sites $0$ and $N/2$ this coefficient is the same, and for points beyond the mid-lattice point the coefficients (in any row of the circulant matrix) begin to repeat in reverse order due to translational invariance. Such a circulant matrix can be block-encoded using techniques analogous to those employed for the SLAC Laplacian operator. Since, as discussed previously, the bottleneck of this block encoding lies in the state preparation, we employ the infinite-lattice limit $N \rightarrow \infty$ of the above coefficients:
\begin{equation}\label{slac1cf}
\tilde{\Delta}_{j}^{(1)}=\left\{
\begin{array}{ll}
0, & \text{if } j = 0,N/2 \\[10pt]
\frac{(-1)^{1+j}}{(j)}e^{i\pi j/N}, & \mathrm{if\;} 0<j<N/2\\[10pt]
\frac{(-1)^{1+j}}{(N-j)}e^{i\pi j/N} & \mathrm{otherwise}\,.
\end{array}
\right.
\end{equation}

The truncation error for the Fourier symbol of the exact SLAC derivative in Eq.~\eqref{slacfo} with that of the truncated representation in Eq.~\eqref{slac1cf} is shown in the upper-right panel of Fig.~\ref{fig:slac_symbol_comparison}.

In the $\mathrm{LCU}$ construction, we aim to prepare an unnormalised state whose amplitudes are given by the square roots of the coefficients defined above. Since the same amplitudes repeat from $j = N/2$, we can employ an inequality test, analogous to that used in the state preparation for the SLAC Laplacian, to generate these amplitudes on $n-1$ qubits. A Hadamard gate is then applied to the $n^{\text{th}}$ qubit, thereby duplicating these $N/2$ amplitudes across the second half of the $N$ states.
\begin{figure}[t!]
            \includegraphics[width=\columnwidth]{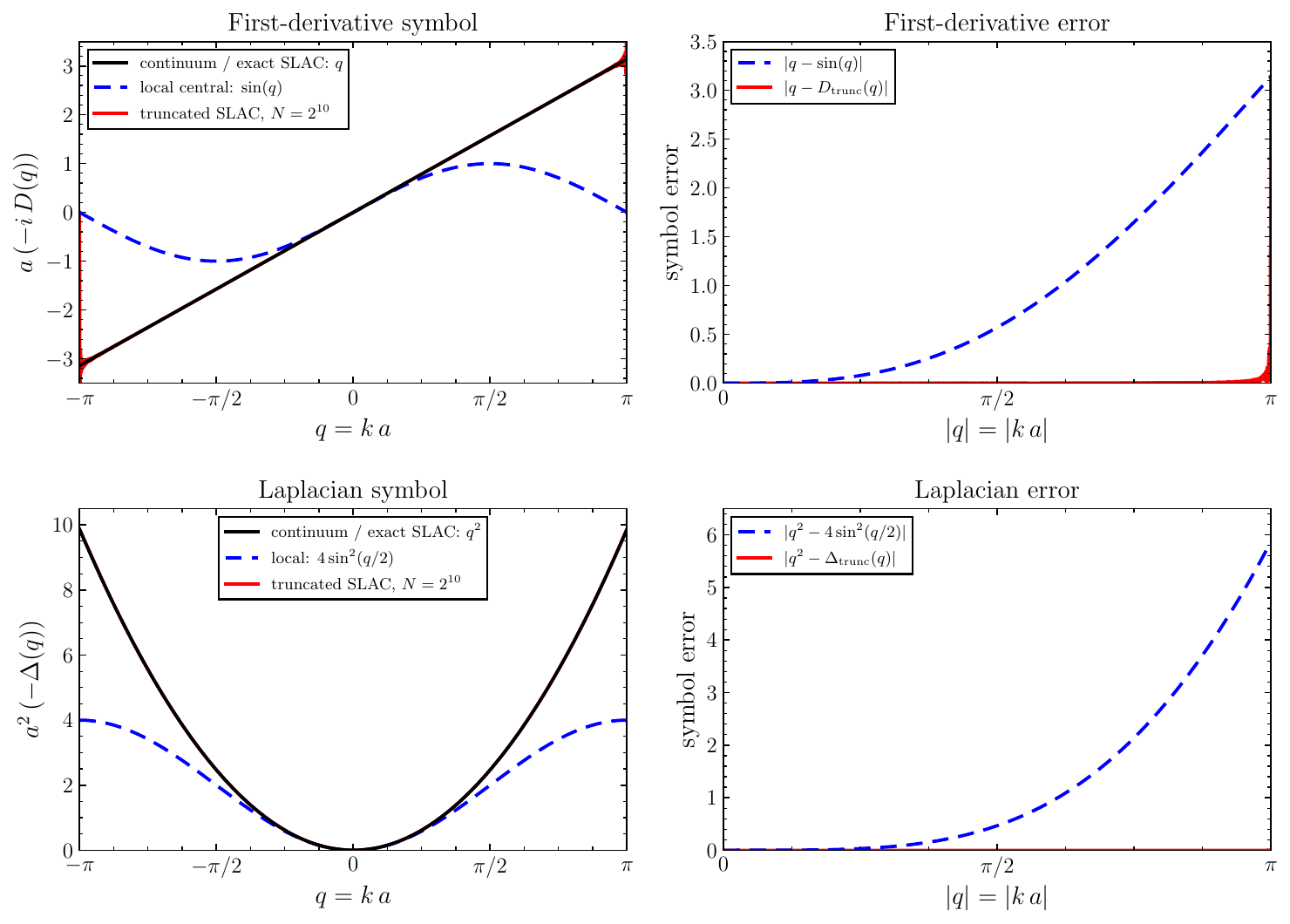}
\caption{
Comparison of local finite-difference and truncated SLAC symbols on the first Brillouin zone, with dimensionless momentum $q=ka$. 
    The exact SLAC first derivative and Laplacian have symbols $q$ and $q^2$, respectively, while the local finite-difference approximations have symbols $\sin(q)$ and $4\sin^2(q/2)$. 
    The truncated SLAC symbols are obtained from the finite Fourier sums
    $D_{\mathrm{trunc}}(q)=2\sum_{j=1}^{N-1}(-1)^{j+1}\sin(jq)/j$ and
    $\Delta_{\mathrm{trunc}}(q)=\pi^2/3+4\sum_{j=1}^{N-1}(-1)^j\cos(jq)/j^2$.
    The plots show that truncating the SLAC representation produces a small finite-truncation error, whereas local finite differences exhibit a dispersion error away from $q=0$.
}

		\label{fig:slac_symbol_comparison}
\end{figure}

\subsection{State preparation for first order SLAC derivative}\label{subsec:prepfo}
We proceed by implementing the same steps as described in Subsection~\ref{subsec:preplap} on the $n-1$ qubits of the ancilla register $j$, i.e., we ``place'' each value of $j$ into boxes, where each box is controlled by another ancilla register $\mu$. An inequality test is then performed so that the success probability of the state preparation remains nominal and the block-encoding scales with a nominal scaling factor. A set of flag qubits is employed to separate the relevant states from the unwanted ones in the overall superposition. Since the sum of the coefficients, $1/j$ scales as the harmonic number \(H_N\), the contribution from the \(1/j\) terms grows only logarithmically with \(N\). Thus, for \(N=2^n\), one expects the corresponding block-encoding to have a subnormalisation constant scaling as \(\mathcal{O}(n)\).

\begin{figure}[t!]
            \includegraphics[width=\columnwidth]{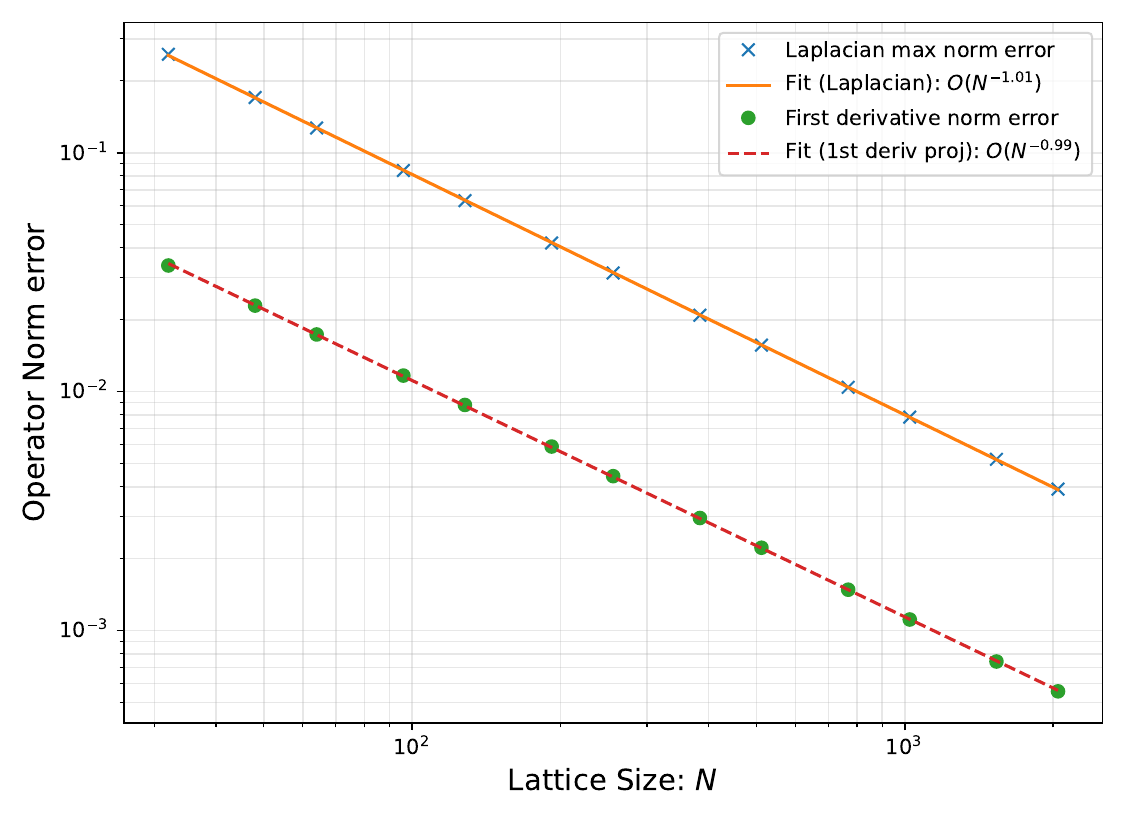}
\caption{
Log--log plot showing the scaling of the operator-norm error between the ideal SLAC operators and their block-encoded counterparts as a function of the lattice size~$N$. The second-order SLAC Laplacian of~Eq.~\eqref{slacc} is compared with its truncated form~Eq.~\eqref{slaccf}, and the first-order SLAC derivative of~Eq.~\eqref{slacfo} is compared with its truncated form~Eq.~\eqref{slac1cf}. Both operators exhibit an error scaling of $\mathcal{O}(1/N)$, in agreement with the analytical estimates of~Eq.~\eqref{slaclaperror} and~Eq.~\eqref{errorfo}. For the first-order derivative, the high-momentum modes have been projected out to emphasize the physically relevant subspace.
}

		\label{fig:slac1infi}
\end{figure}
Note that in the $\mathrm{LCU}$ construction for the first-order SLAC derivative, the coefficient associated with the $j=0$ state is absent. Therefore, the initial rotation used in the Laplacian state-preparation protocol is not required. Instead, the protocol for the first-order derivative begins by preparing the following superposition over the unary register $\mu$:
\begin{align}
&\sqrt{\frac{1}{(n-1)}} 
\sum_{\mu=0}^{n-2}\ket{[\mu]_{n-1}} \ket{[0]_{n-1}}_j\ket{0}_d.
\end{align}
This superposition is efficiently prepared using a sequential cascade of parameterized $R_y$ rotations. We begin by applying an unconditional $R_y(\theta_0)$ gate to the first qubit, $\mu_0$. Next, we apply a sequence of controlled-$R_y(\theta_k)$ rotations to each subsequent qubit $\mu_k$, controlled on the state of the preceding qubit $\mu_{k-1}$, for $k = 1, \dots, n-3$. To ensure that the probability amplitude is partitioned equally among all $n-1$ unary branches, the classical rotation angles are precomputed as $\theta_k = 2\arccos\!\bigl(1/\sqrt{n - 1 - k}\bigr)$. 

Following this unary state preparation, an equal superposition is generated in the register $j$ by applying a series of Hadamard gates, each controlled by the corresponding active qubits of the unary register $\mu$:
\begin{align}
&\sqrt{\frac{1}{(n-1)}} 
\sum_{\mu=0}^{n-2}  
\sqrt{2^{-\mu}}\;\ket{[\mu]_{n-1}}\sum_{j=0}^{2^\mu-1}\ket{[j]_{n-1}}_j\ket{0}_d.
\end{align}

Next, we convert the unary-encoded register $\mu$ into a one-hot representation by applying a cascade of $\mathrm{CX}$ gates, $\mathrm{CX}_k^{k+1}$, sequentially starting from the qubit indexed $(n-3)$, followed by an $\mathrm{X}$ gate applied to the least significant qubit (LSB) of $\mu$. This step closely follows the procedure used in the Laplacian state-preparation protocol. 

We then prepare the conditional superposition over the register $j$ corresponding to each value of $\mu$, such that $\mu \ge 0 \Rightarrow j = 2^{\mu}, \dots, 2^{\mu+1}-1$. In particular, the sequence of gates $\prod_{i=0}^{n-2} \mathrm{CX}_{\mu_i}^{j_i}$ is applied, with each $\mathrm{CX}$ gate controlled on the $i^{\mathrm{th}}$ qubit of the register $\mu$ and targeted on the corresponding qubit of the register $j$. This yields a superposition in which the values $j = 0, \dots, N/2 - 1$ are grouped into boxes labelled $B_\mu$,
\begin{align}\label{boxfito1}
&\sqrt{\frac{1}{({n}-1)}}
   \sum_{\mu=0}^{n-2}\sqrt{2^{-\mu}}\ket{\mu}\sum_{j \in B_{\mu}} \ket{[j]_{n-1}}.
\end{align}
To perform the inequality test, we first prepare the reference register $\mathit{ref}$ in an equal superposition over the $M = 2^{n_{ref}}$ values, and then check that the condition $m/M < 1$ holds by testing the following inequality:
\begin{equation}\label{ineqder1}
\begin{aligned}
M\,\times{2^{\mu}} > {m}\times j, \qquad \mathrm{for\,} j \neq 0.
\end{aligned}
\end{equation}

The inequality is evaluated such that, for the chosen parameters, the state preparation achieves a nominal success probability. Consequently, the amplitude associated with each state $\ket{j}$ is multiplied by the square root of,
\begin{equation}
    \begin{aligned}
 &Q = \Bigg\lceil {\frac{2^{\mu}\; M\;}{j}}\Bigg\rceil.
\end{aligned}
\end{equation}

After applying the inequality test on the $(n-1)$ qubits of the register $j$, the state obtained has its success branch identified by $\ket{0}_{f}$. For $M \gg 0$, the ceiling functions appearing in the above expressions can be neglected, introducing a preparation error $\varepsilon_{\mathrm{prep}}$ that can be bounded analogously to the Laplacian state-preparation case. Therefore, after performing the inequality test and applying a Hadamard gate to the most significant qubit $d$, the resulting state becomes:
\begin{equation}\label{fis1}
\begin{aligned}
&\sqrt{\frac{1}{M\,({n}-1)}}\sum_{\mu=0}^{n-2}\ket{\mu}\sum_{j \in B_{\mu}}\sqrt{2^{-\mu}}\ket{[j]_{n-1}}\sum_{m=0}^{Q-1}\ket{m}
\ket{+}_d\ket{0}_{f}\\&+\sqrt{1-p^{(1)}}\ket\omega_f\,.
\end{aligned}
\end{equation}
Here, $\ket{+}_d$ denotes the equal superposition state prepared in the most significant qubit of the register $j$. 

The parameter $p^{(1)}$ represents the success probability associated with the implementation of the inequality test and is defined as the sum of the squared amplitudes within the success branch:
\begin{align}\label{normalisingfo}
&p^{(1)}
=  \frac{1}{\,(n-1)}\Biggl\{\sum_{\mu=0}^{n-2}
     \sum_{j \in B_\mu}\frac{2^{-\mu}}{M}.Q\Biggr\}\notag\\
&p^{(1)}_{\;n\gg1} = \frac{1}{(n-1)}\sum_{j=1}^{N/2-1}\frac{1}{j}\notag\\
&\qquad\;\;\approx\frac{1}{(n-1)}H_{N/2-1}\approx 0.693 + \frac\gamma n\,.
\end{align}
In the asymptotic limit of large $n$, the success probability converges to a strict constant, $\lim_{n \to \infty} p^{(1)} = \ln(2) \approx 0.693$. Here, $H_N$ denotes the $N^{\text{th}}$ harmonic number, which can be asymptotically expressed as $H_N \approx \ln N + \gamma$, where $\gamma$ is the Euler-Mascheroni constant. For finite $n$, the success probability strictly bounds the asymptotic limit from above. Because the initial state is prepared with a significantly high probability of success ($\geq 69.3\%$), this state-preparation protocol requires zero rounds of amplitude amplification.

All preceding state-preparation steps share the same complexity as in the Laplacian case; however, the arithmetic cost of the inequality test is significantly reduced. The $\mathrm{LHS}$ of the inequality can be evaluated in a target register of size $(n-1)+n_{\mathrm{ref}}$. Unlike the Laplacian case, we do not need to interleave with extra zeros. Because $\mu$ is encoded in one-hot unary, its active wire intrinsically represents the binary value $2^\mu$. To multiply by $M=2^{n_{\mathrm{ref}}}$, we simply pad the register with $n_{\mathrm{ref}}$ zeroed ancillae. Consequently, the Toffoli cost of computing the $\mathrm{LHS}$ is strictly zero. 

Computing the right-hand side ($\mathrm{RHS}$), $m \times j$, requires a standard quantum multiplier writing to a target register of size $(n-1)+n_{\mathrm{ref}}$, which incurs a cost of $2(n-1)n_{\mathrm{ref}}-n_{\mathrm{ref}}$ Toffoli gates. The inequality can then be tested using an $((n-1)+n_{\mathrm{ref}})$-bit comparator with a Toffoli cost of $(n-1)+n_{\mathrm{ref}}$. Summing these contributions, the total cost of computing the $\mathrm{LHS}$, $\mathrm{RHS}$, and executing the comparator is exactly $2(n-1)n_{\mathrm{ref}}+(n-1)$. As will be shown later, the required reference register size scales as $n_{\mathrm{ref}}=\mathcal{O}(n)$. Therefore, the overall Toffoli complexity of implementing this first-order derivative inequality test is bounded by $\mathcal{O}(n^2)$.

\subsection{LCU for the first order SLAC derivative}
\label{SLAC1LCU}
To construct the block-encoding using the $PREP$ and $\mathrm{SELECT}$ circuits, one can follow steps similar to those in the SLAC Laplacian's block-encoding. We denote all the operations involved in the $\mathrm{LCU}$ of the first-order SLAC derivative as 
\begin{equation}
    U_\mathrm{slac}^{(1)} = O_\mathrm{prep}^{(1)\dagger} O_\mathrm{select}^{(1)}O_\mathrm{phase}^{(1)}O_\mathrm{cpy}^{(1)}O_\mathrm{prep}^{(1)},
\end{equation}
and list the intermediate states obtained at each stage of the $\mathrm{LCU}$ in the Appendix, starting from Appendix~\ref{app: O_prep}. The copy oracle first marks the desired branch of the superposition by flipping a qubit $C$, such that the desired branch of the prepared state is located at $C=0$ as shown in Eq.~\eqref{cpy1}. The phase oracle is then applied to the state obtained in Eq.~\eqref{fis1}, producing alternating signed weights in the $\mathrm{LCU}$ as in Eq.~\eqref{signfo}.  The $\mathrm{SELECT}$ oracle acts with controls on the register $j$ and $d$: a $0$-control on $d$ applies the adder circuit implementing $P^{j}$, while a $1$-control applies $P^{N-j}$ on the system register, as expressed in Eq.~\eqref{selfo}, thereby realising the circulant structure in the block-encoding. Finally, all operations in $PREP$ are uncomputed, which discards the irrelevant parts of the superposition. In this way, an efficient block-encoding of the first-order SLAC derivative is implemented, as shown in the following derivation:
\begin{widetext}
\begin{equation}
\begin{aligned}\label{slac1be}
&\bigl(
  \bra{[0]_{n-1}}_\mu\,\bra{[0]_n}\bra{[0]_{n_{ref}}}_{\mathrm{ref}}\bra{0}_{C}\bra{0}_{f} \otimes I_n
\bigr)\,
U_{\mathrm{slac}}^{(1)}\,
\bigl(
  \ket{[0]_{n-1}}_\mu\,\ket{[0]_n}\ket{[0]_{n_{ref}}}_{\mathrm{ref}}\ket{0}_{C}\ket{0}_{f} \otimes I_n
\bigr)\ket{\psi}_{sys}
\\
&\;=\;
\bigl(
  \bra{[0]_{n-2}}_\mu\,\bra{[0]_{n_{ref}}}_{\mathrm{ref}}\,\bra{[0]_n}\,\bra{0_{f}}\bra{0}_{C} \otimes I_n
\bigr)\,
O_{\mathrm{prep}}^\dagger
\sqrt{\frac{1}{2\,(n-1)}}\Biggr\{\sum_{\mu=0}^{n-2}\ket{\mu}
   \sum_{j \in B_{\mu}} \exp\Big(\frac{i\pi j}{N}\Big)(-1)^{1+j}\\&\otimes\sqrt{\frac1j}\ket{[j]_{n-1}}\Bigg(\ket{0}_d\otimes P^{j}_{sys}+\ket{1}_d\otimes P^{N-j}_{sys}\Bigg)\Biggl\}\ket{0}_{f}\ket{0}_C+ \ket{\omega}_{fail}\ket{1}_C\ket{\psi}_{sys}
\\
&\;\approx\;
\frac{1}{2\,(n-1)} 
\Biggl\{ 
   \sum_{j=1}^{N/2-1} \exp\!\Bigl(\frac{i\pi j}{N}\Bigr)(-1)^{1+j}
   \Biggl(
      \frac1j\otimes P^j
      + 
      \frac1j\otimes P^{N-j}
   \Biggr)\Biggl\}\ket{\psi}_{sys}
\\
&=
\frac{1}{2\,(n-1)}  \,\tilde\Delta^{(1)}\ket{\psi}_{sys}\,.%
\end{aligned}
\end{equation}
\end{widetext}
Hence, we obtain the block-encoding of the first-order SLAC derivative operator $\tilde{\Delta}^{(1)} / \alpha^{(1)}$, where the rescaling factor is $\alpha^{(1)} = 2(n - 1)$.

In direct analogy with the Laplacian case, we bound the error in the block-encoding of the first-order SLAC derivative by estimating the operator norm of the difference between the ideal and truncated matrices. Both matrices are antisymmetric and circulant, and hence can be diagonalised by applying the quantum Fourier transform to their first row. For the full first-order derivative operator, the largest eigenvalue occurs near the highest momentum mode, at $k=N/2+1$. However, in the present setting we restrict attention to the low-momentum sector, since the high-$k$ modes are not physically probed. Applying the same truncation-error analysis as in the Laplacian case then yields the following bound on the norm error of the first-order derivative:
\begin{widetext}
    \begin{align}\label{errorfo}
&\bigl\|
 \alpha \bigl(\bigl\langle 0^a\bigr|\otimes I_s\bigr)\,U\,\bigl(\lvert 0^a\rangle\otimes I_s\bigr)
 -{\Delta}^{(1)}_{\mathrm{slac}}
\bigr\| \;=   
\bigl\|
 \tilde{\Delta}^{(1)}_{\mathrm{slac}}
 -{\Delta}^{(1)}_{\mathrm{slac}}
\bigr\| = \bigl\|F^{\dagger}\tilde{D}F-F^{\dagger}{D}F\bigl\| \;                                                 \notag\\
&\mathrm{where}\;\; D=    \sum_{j=1}^{N/2-1}{\Delta}^{(1)}(j)\;e^{\frac{-i2\pi jk}{N}} \ket k \bra k \;\;\mathrm{as}\;\;     \Delta^{(1)}_{\mathrm{slac}}=\sum_j \Delta^{(1)}P^j  \notag\\
&\therefore \;\;\bigl\|
 \tilde{\Delta}^{(1)}_{\mathrm{slac}}
 -{\Delta}^{(1)}_{\mathrm{slac}}
\bigr\|=   \max\left\{\sum_j(\tilde\Delta^{(1)}(j)-\Delta^{(1)}(j))e^{\frac{-i2\pi jk}{N}}  \right\} = \left\{\sum_j[\tilde\Delta^{(1)}(j)-\Delta^{(1)}(j)]e^{\frac{-i2\pi jk}{N}}\right\}_{k=N/2+1}        \notag\\
& \tilde\Delta^{(1)}(j)-\Delta^{(1)}(j) \le (-1)^{1+j}\left(\frac{1}{j} - \frac{1}{j} + \frac{\pi^2j}{3N^2} \right)e^{i\pi j/N} = (-1)^{1+j}\left(\frac{\pi^2j}{3N^2} \right)e^{i\pi j/N} \notag\\
& \mathrm{Define}\hspace{5pt} \Pi_k = \sum_{|k|<=k_{max}}\ket{k}\bra{k} \notag\\
&
\Pi_k (\tilde\Delta^{(1)}(j) - \Delta^{(1)}(j))\Pi_k^\dagger \;\le\;
\left| \sum_{j=1}^{N/2-1} 
\left( \frac{\pi^2 j}{3N^2} \right) 
\exp\!\left( i j \left(\pi - \tfrac{2\pi k}{N}\right)\right) \right|_{k_{max}} = \left| \sum_{j=1}^{N/2-1} 
\left( \frac{\pi^2 j}{3N^2} \right) 
\big(\exp\!\left( i\theta_k\right)\big)^j \right|    \notag\\&
\bigl\|
 \tilde{\Delta}^{(1)}_{\mathrm{slac}} 
 -{\Delta}^{(1)}_{\mathrm{slac}}
\bigr\|= \left| \sum_{j=1}^{N/2-1} 
\left( \frac{\pi^2 j}{3N^2} \right) 
\big(r_k\big)^j \right|_{k_{max}=N/2.5}\equiv \mathcal{O}\left(\frac{1}{N}\right) \,.
\end{align}
\end{widetext}
Here $r_k=\exp{(i\theta_k)},\quad s.t.\,\,\theta_k=(\pi-2\pi k/N)$ and $\Pi_k$ denotes a projector in momentum space that restricts the momentum modes to values less than a chosen cutoff $k_{max} = N/2.5$, thereby excluding the high-momentum modes that could otherwise worsen the error scaling. This error bound on the projected-out subspace of the eigenvalues is justified, as one seldom probes physics in this high-energy (projected-out) regime of any theory.

Within the projected subspace, the finite-lattice truncation error for the first-order SLAC derivative behaves analogously to the Laplacian case, namely
\[
\varepsilon_{\mathrm{trunc}}
=
\mathcal{O}\!\left(\frac{1}{N}\right).
\]
If we set the target precision to be $\varepsilon$, then as $\varepsilon_{trunc}\leq \varepsilon$, the lattice size can be chosen as $N=\Omega(1/\varepsilon)$.

The finite reference register used in the inequality test of Eq.~\eqref{ineqder1} introduces a preparation error. 
Neglecting ceiling effects, this error is bounded by the squared amplitudes immediately before the inequality test in Eq.~\eqref{fis1}, giving
\[
\varepsilon_{\mathrm{prep}}
\le
\frac{1}{4M(n-1)}
\sum_{\mu=2}^{n-1}\sum_{j\in B_\mu}2^{-(\mu-1)}
=
\mathcal{O}\!\left(\frac{1}{M}\right).
\]
Requiring this contribution to be of order $\varepsilon$ gives
\[
M=\Omega\!\left(\frac{1}{\varepsilon}\right).
\]
Since the lattice size is chosen such that $N=\Omega(1/\varepsilon)$, we may take
\[
M=\Theta(N).
\]
Therefore, the reference register requires
\[
n_{\mathrm{ref}}
=
\Theta(\log M)
=
\Theta\!\left(\log\frac{1}{\varepsilon}\right)
=
\Theta(\log N)
=
\Theta(n)
\]
qubits.

We formalise the results of this section in the following proposition, which gives an efficient block-encoding of the first-order SLAC derivative operator.
\begin{prop}\label{corl2}
An $\varepsilon$-approximate block-encoding of the $n$-qubit first-order SLAC derivative can be obtained within the $\mathrm{LCU}$ framework by combining the inequality-test–based state-preparation routine in~Eq.~\eqref{PREP1} with the modular-subtractor–based $\mathrm{SELECT}$ step in~Eq.~\eqref{selfo}. This construction incurs a gate complexity of $\mathcal{O}(n^2)$ and requires $\mathcal{O}(n)$ ancilla qubits. Applying Definition~\ref{lm2}, this construction yields a $(2(n-1),\,\mathcal{O}(n),\,\varepsilon)$ efficient block-encoding of the first-order SLAC derivative, explicitly satisfying
\begin{widetext}
\begin{equation}\label{beslac1}
\left\| {\tilde{\Delta}^{(1)}} - \alpha^{(1)}
\left(  \bra{[\mu]_{(n-1)}}\bra{[0]_{n}}\bra{[0]_{m}}\bra{0}_C\bra{0}_f\otimes I_n \right) 
U_{\mathrm{slac}}^{(1)} 
\left(  \ket{[\mu]_{(n-1)}}\ket{[0]_{n}}\ket{[0]_{m}}\ket{0}_C\ket{0}_f \otimes I_n \right) 
\right\| \leq \varepsilon.
\end{equation}  
\end{widetext}
\end{prop}
\begin{proof}
\normalfont
The construction of the $n$-qubit first-order SLAC derivative block-encoding follows the LCU framework as formalized in Definition~\ref{lm1}. The total gate complexity is dominated by two primary components: state preparation and the SELECT oracle. As established in Result~\ref{cor:prep1}, the state-preparation routine requires $\mathcal{O}(n^2)$ gates, while Result~\ref{lem:sel} shows that the modular-subtractor-based SELECT oracle similarly incurs a cost of $\mathcal{O}(n^2)$. Auxiliary components, including the COPY and PHASE oracles (Appendices~\ref{app:cpy} and~\ref{app: phase}), introduce only negligible overhead to the total depth. 

The inequality test requires $n+n_{\mathrm{ref}}$ scratchpad qubits. This gives a total ancilla count of $3n+2n_{\mathrm{ref}}$, which is $\mathcal{O}(n)$.

Finally, as the rescaling ratio $\alpha^{(1)}/\|\tilde{\Delta}^{(1)}\|$ scales linearly with $n$ (where $\|\tilde{\Delta}^{(1)}\| = \pi)$, the construction successfully yields an  $(\alpha^{(1)}, \mathcal{O}(n), \varepsilon)$ efficient block-encoding with an overall gate complexity of $\mathcal{O}(n^2)$, satisfying the criteria of Definition~\ref{lm2}.
\end{proof}

\section{Breaking the translational invariance of the SLAC operators}
\label{brokentransinv}
The SLAC derivative operators exhibit a circulant structure that arises from their inherent translational invariance, permitting their block-encodings to be efficiently constructed through diagonalization in the Fourier basis. In this section, however, we elaborate on why such a diagonalisation-based approach was not adopted in our study. In many physically relevant situations---for instance, when external potentials are introduced or when symmetry is explicitly broken at the discretisation stage---the translational symmetry of the operator no longer holds. Under these conditions, the $\mathrm{LCU}$-based construction developed in this work, which integrates the inequality-test based state-preparation circuit with an adder-based $\mathrm{SELECT}$ operation, provides a more natural and adaptable framework than the diagonalisation method, particularly for handling operators with broken or spatially varying symmetries.

To illustrate the procedure, we first consider a case study in which an inhomogeneous mass term is incorporated into the SLAC derivative operator in the form of a potential, denoted by $V_d$. A related scenario was examined in Section~4.7.1 of~\cite{qft3}, where the generation of the ground state of an inhomogeneous-mass quantum field theory was investigated by introducing a potential comprising a {single non-zero diagonal element}---chosen to exceed the free-field mass. This localized perturbation, commonly referred to as a {point-defect potential}, can be expressed as 
\[
V_{\mathrm{pd}}(x) = \lambda\,\delta(x - x_0),
\]
which, in operator form, corresponds to an $N = 2^n$-dimensional diagonal matrix whose entries are all zero except at $x = x_0$, where the diagonal element takes the value $\lambda$. The Hamiltonian of the resulting quantum system is therefore 
\[
H = \Delta_{\mathrm{SLAC}}^{(1;2)} + V_{\mathrm{pd}}.
\]
This Hamiltonian can be efficiently block-encoded by combining the block-encodings of the SLAC operators---constructed in Sections~\ref{sec: slac1be} and~\ref{sec: slaclap}---with that of the potential term, which is trivially block-encoded owing to its {1-sparse} structure. By invoking the {sum-of-block-encodings} technique described in Section~4.3 of~\cite{qsvt}, one can thereby construct a block-encoding of the complete Hamiltonian, which in this case is {not translationally invariant}.

In general, if an efficient block-encoding of a generic potential operator $V(x)$ is accessible, the corresponding Hamiltonian with broken translational symmetry can be readily constructed using the aforementioned case study. Notably, this statement also holds when the SLAC operator itself is block-encoded through the diagonalization approach, wherein one block-encodes the diagonal matrix (see~\cite{be2}) of SLAC eigenvalues, denoted by $D_{\mathrm{SLAC}}$, and implements the following similarity transformation via a quantum circuit:
\begin{equation*}
    \Delta_{\mathrm{SLAC}} = QFT^{\dagger} D_{\mathrm{SLAC}} QFT.
\end{equation*}
Thus, in both formulations---whether employing the diagonalisation-based or the $\mathrm{LCU}$-based block-encoding---the construction of a non-translationally invariant Hamiltonian, $H = \Delta_{\mathrm{SLAC}} + V(x),$ relies on the same key requirement: access to an efficient block-encoding of the potential operator $V(x)$, which typically constitutes the principal computational bottleneck in such constructions.

Next, we consider a more unconventional setting in which the derivative operator is still motivated by the SLAC formalism, but no longer retains exact translational invariance at the level of its matrix entries. Specifically, suppose that each row \(j=0,\ldots,2^n-1\) is assigned a row-dependent cutoff \(c_j\), beyond which the ideal SLAC coefficients, proportional to \(1/j\) or \(1/j^2\), are set to zero. The resulting SLAC-type operator is therefore no longer circulant, and is represented by a non-circulant matrix as,
\begin{equation}\label{slacmod}
    \Delta_{\mathrm{mod}} = \sum_{j=1}^{\lfloor N/2 \rfloor} 
    \widetilde{\Delta}^{(1;2)}_{j}\, 
    M_j \left(P^{j} - P^{N-j}\right),
\end{equation}
where $\widetilde{\Delta}^{(1;2)}_{j}$ denote the coefficients of the ideal SLAC operators as defined in Eq.~\eqref{slaccf} and Eq.~\eqref{slac1cf}, and $M_j$ is a diagonal masking matrix given by
\begin{equation}\label{mask}
    (M_j)_{kk} = \chi_k(j) =
    \begin{cases}
        1, & j \leq c_k, \\[4pt]
        0, & j > c_k.
    \end{cases}
\end{equation}

We now outline the construction of the block-encoding of $\Delta_{\mathrm{mod}}$. It is evident that the preparation circuits introduced in Subsections~\ref{subsec:preplap} and~\ref{subsec:prepfo} can still be employed for this purpose. At a high level, the initial step of the circuit proceeds as follows (up to normalisation):
\begin{align}\label{slacmodprep}
&\ket{[0]_{n}}\ket{0}_f\ket{0}_{\mathrm{sel}}\ket{0}_m\ket{\psi}_{\mathrm{sys}} 
\xrightarrow{\mathrm{PREP}}\nonumber \\
&\sum_{j,k=0}^{N-1} \sqrt{\widetilde{\Delta}_j}\, \psi_k 
\ket{[j]_{n}}\ket{0}_f\ket{0}_{\mathrm{sel}}\ket{0}_m\ket{k}_{\mathrm{sys}},
\end{align}
where the state of the system register is expanded as $\ket{\psi}_{\mathrm{sys}} = \sum_{k} \psi_k \ket{k}$. 

Next, to implement the action of the diagonal mask $M_j$, we introduce a predicate operation, denoted by ${KEEP}$, which first loads the row-dependent cut-off $c_k$ corresponding to each row $k$, and then flips the flag qubit $f$ whenever the current index satisfies $j \le c_k$, as determined by a quantum comparator circuit~\cite{compgid}. The states for which the flag is flipped correspond to matrix elements with $M_j = 1$, i.e., the coefficients of the SLAC operator that are retained in the success branch. The action of this predicate oracle, denoted as $KEEP(j,k)$, is defined as
\begin{equation}\label{predicate}
    \ket{j}\ket{k}\ket{f=0} 
    \xrightarrow{\mathrm{KEEP}} 
    \ket{j}\ket{k}\ket{f \oplus [j \le c_k]}.
\end{equation}

When the predicate oracle \text{KEEP} is applied to the prepared state of~Eq.~\eqref{slacmodprep}, the resulting state can be written as:
\begin{equation*}
    \sum_{k=0}^{N-1}\psi_k\left(\sum_{j\le c_k}\sqrt{\Delta_j}\ket{[j]_{n}}\ket{1}_f+\sum_{j> c_k}\sqrt{\Delta_j}\ket{[j]_{n}}\ket{0}_f\right)\ket{k}.
\end{equation*}
Before implementing the \text{SELECT} operation of the $\mathrm{LCU}$ construction, we introduce an additional step in which the most significant bit (MSB) of the register $j$ is coherently copied (using controlled Toffoli) onto a fresh flag qubit, denoted by $sl$, conditioned on the flag $f$. After this controlled operation, the overall state transforms to
\begin{align}\label{afterkeep}
    &\sum_{k=0}^{N-1} \Bigg\{\bigg(
        \sum_{\substack{j \le c_k \\ j < N/2}} 
            \sqrt{\Delta_j} \ket{[j]_n} \ket{0}_{sl}  +\sum_{\substack{j \le c_k \\ j \ge N/2}} 
            \sqrt{\Delta_j} \ket{[j]_n} \ket{1}_{sl}\bigg)\nonumber\\&
     \otimes\ket{1}_f +\;\sum_{j > c_k} 
            \sqrt{\Delta_j} \ket{[j]_n} \ket{0}_{sl} \ket{0}_f
    \Bigg\}\psi_k  \ket{k}_{\mathrm{sys}}.
\end{align}
The \text{SELECT} circuit required to implement the modified derivative operator applies modular-adder operations corresponding to the permutation operators $P^{j}$ or $P^{N-j}$ on the success branch, in a manner analogous to that described in~Eq.~\eqref{selfo}. In addition, the modified \text{SELECT} circuit must flip the marker flag qubit $m$ so that the failure branch is coherently separated and subsequently discarded during the uncomputation step. The complete action of this modified \text{SELECT} operator can therefore be expressed as:
\begin{align}\label{selmod}
\sum_{j}\ket{j}\bra{j}\otimes\Biggr\{\Big(&\ket{f{=}1}\bra{f{=}1}
    \otimes \ket{0}\bra{0}_{sl}
    \otimes \mathbb{1}_m
    \otimes P^{j}_{sys} \Big) \nonumber \\[6pt]
+~\Big(&\ket{f{=}1}\bra{f{=}1}
    \otimes \ket{1}\bra{1}_{sl}
    \otimes \mathbb{1}_m
    \otimes P^{N-j}_{sys} \Big)\nonumber \\[6pt]
+~\Big(&\ket{f{=}0}\bra{f{=}0}
    \otimes \ket{0}\bra{0}_{sl}
    \otimes {X}_m
    \otimes \mathbb{1}_{sys} \Big)\Biggr\}.
\end{align}
After applying the modified \text{SELECT} oracle defined in~Eq.~\eqref{selmod} to the state prepared in~Eq.~\eqref{afterkeep}, the system evolves to the following state:
\begin{align}
    \sum_{k=0}^{N-1} \Bigg(
        &\sum_{\substack{j \le c_k \\ j < N/2}} 
            \widetilde{\Delta}_j \ket{[j]_n} \ket{0}_{sl} \ket{1}_f \ket{0}_m \otimes P^j_{\mathrm{sys}} \nonumber \\[4pt]
      +\;&\sum_{\substack{j \le c_k \\ j \ge N/2}} 
            \widetilde{\Delta}_j \ket{[j]_n} \ket{1}_{sl} \ket{1}_f \ket{0}_m \otimes P^{N-j}_{\mathrm{sys}} \nonumber \\[4pt]
      +\;&\sum_{j > c_k} 
            \widetilde{\Delta}_j \ket{[j]_n} \ket{0}_{sl} \ket{0}_f \ket{1}_m \otimes \mathbb{I}_{\mathrm{sys}}
    \Bigg) \psi_k\ket{k}_{\mathrm{sys}}.
\end{align}

The respective phase oracles described in~\ref{app: O_sgn} and~\ref{app: phase} can likewise be implemented under the control of the flag qubit $f$. It then follows that, after flipping the flag and then uncomputing the \text{PREP} step, the matrix elements of the SLAC operator corresponding to the failure branch---identified by the qubits $(m{=}1, f{=}1)$---are coherently suppressed, while those in the success branch $(m{=}0, f{=}0)$ are retained in the modified SLAC matrix. Consequently, the combined action of the predicate oracle \text{KEEP} defined in~Eq.~\eqref{predicate} and the modified \text{SELECT} oracle of~Eq.~\eqref{selmod} effectively realises the diagonal mask $M_j$ introduced in~Eq.~\eqref{mask}. After the \text{UNPREP} step, postselecting on all ancilla qubits being in the $\ket{0}$ state yields an operator acting on the system register that corresponds precisely to the modified SLAC-type operator of~Eq.~\eqref{slacmod}. It should be emphasised that the discussion presented here provides only a high-level description and does not account for the explicit normalisation and rescaling factors that ultimately determine the efficiency of the resulting block-encoding.

We emphasize that any modified SLAC-type derivative operator exhibiting broken translational symmetry can, in principle, be generated by implementing a projector $M_j$ through a dedicated predicate-type oracle (corresponding to the \text{KEEP} step in the above example) together with the modified version of the \text{SELECT} operator. If the operation $M_j$ can be realised efficiently, then the block-encoding of a non-circulant, modified SLAC operator can likewise be implemented efficiently using the inequality-test-based preparation circuits developed in this work.

Such modified derivative operators, where translational symmetry is broken at the component of the derivative operator, cannot in general be conveniently block-encoded via the diagonalisation approach, since the resulting matrices are no longer circulant. Implementing them by decomposing into circulant and non-circulant components requires repeated basis transformations between the position and Fourier representations, which accumulate significant depth overhead when invoked multiple times in amplitude amplification or QSVT routines. In contrast, the preparation and \text{SELECT} circuits introduced here, together with the predicate-type oracles that underpin the non-circulant construction, operate entirely within the position basis. 

If $\alpha$ denotes the rescaling factor associated with the block-encoding of the diagonal matrix $D_{\mathrm{SLAC}}$, the diagonalisation-based implementation incurs a cost of $O(\alpha n^2)$ per application due to the two $n$-qubit QFTs required for each transformation. The inequality-test-based construction achieves a comparable asymptotic scaling but avoids repeated basis changes, thereby providing a more streamlined and scalable framework for constructing block-encodings of operators with broken translational symmetry. This also ensures that the preconditioner described in~\ref{sec:precon} can be naturally and efficiently applied to the block-encodings obtained through our framework, as the preconditioner is likewise represented in the position basis.

\section{Implementation of Multi-scaled SLAC derivative operator}\label{sec:multisc}
According to Section 5 of~\cite{finiteswt}, the Shannon wavelet transform (SWT) in the Shannon limit provides a natural framework for recursively renormalising the Hamiltonian of a bosonic quantum field theory. In this regime, perfect momentum separation ensures that each application of the SWT splits the Hamiltonian into two equal halves,  $H \;=\; H_{IR}\;\oplus\;H_{UV}$.
The renormalised Hamiltonian is then identified with the $H_{IR}$ block. At the next scale, one applies the SWT of reduced size---specifically, half the size of the previous step---to $H_{IR}$, which once again separates it into IR and UV sectors. Iterating this procedure over successive scales $r$ generates a sequence of Hamiltonians whose dimensions decrease by a factor of two at each step.

Motivated by this recursive renormalisation framework, and noting that SLAC derivative operators naturally lie in the scale-field basis of the SWT, in this section we present a construction of their multi-scale representation tailored for qubit-based quantum computation. For this purpose, we combine the block-encodings of SLAC operators developed in Sections~\ref{sec: slaclap},\ref{sec: slac1be} with the QSWT implementation introduced in Section~\ref{sec:qswt}.
\begin{figure}[ht]
            \includegraphics[width=\columnwidth]{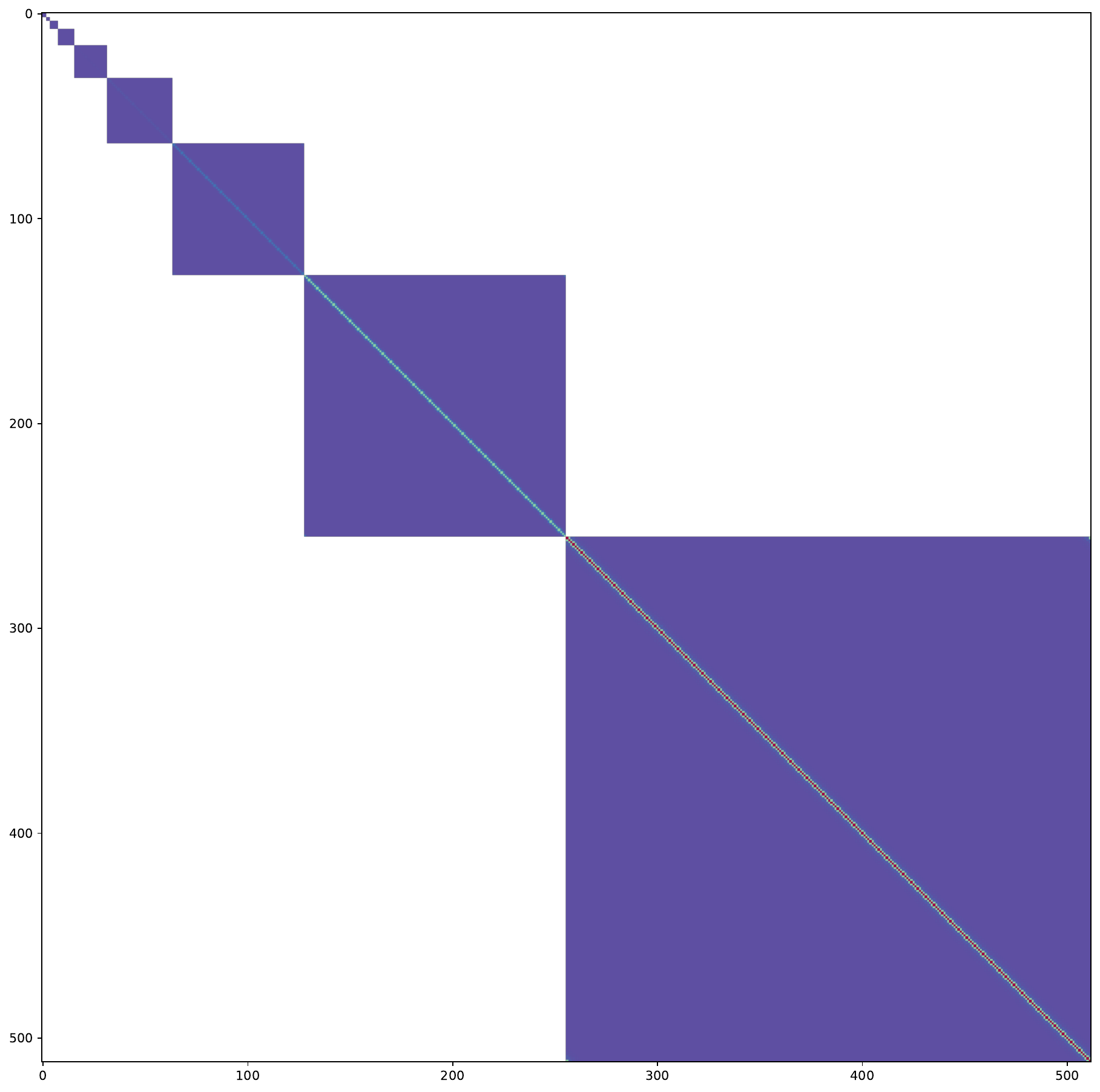}
		\caption{Multi-scaled $512$ sized SLAC Laplacian operator obtained using recursive applications of QSWT is a block-diagonal matrix with IR block at scale $8$ and UV blocks at lower scales with sizes getting halved at each scale.}
		\label{fig:multi-scaledlap}
\end{figure}

Let the \(N\)-dimensional block-encoding of the SLAC derivative be identified with the derivative operator at scale \(r=0\), and denote the quantum Shannon wavelet transform by \(\mathrm{QSWT}=S\). Applying the \(N\)-dimensional transform \(S_{(r=0)}\) as a similarity transformation to this operator gives:
\begin{equation}
S \Delta S^\dagger = 
\begin{bmatrix}
S^{\text{IR}} \Delta S^{\text{IR} \dagger} & 0 \\
0 & S^{\text{UV}} \Delta S^{\text{UV} \dagger}
\end{bmatrix}=
\begin{bmatrix}
\Delta^{IR}_{(r=1)}  & 0 \\
0 & \Delta^{UV}_{(r=1)}
\end{bmatrix}.
\end{equation}
Since perfect momentum separation eliminates the cross terms, the $N$-dimensional SLAC derivative decomposes into two independent blocks of size $N/2$. 

Applying the wavelet transform at scale 1 to the IR block of this matrix yields a further decomposition of the SLAC derivative operator:
\begin{equation}
S_{(r=1)}S \Delta S^\dagger S^\dagger_{(r=1)} = 
\begin{bmatrix}
\Delta^{IR}_{(r=2)}  & 0  &0 \\
0  & \Delta^{UV}_{(r=2)}  &0 \\
0 & 0 &\Delta^{UV}_{(r=1)}
\end{bmatrix} .
\end{equation}

In the above expression, recursively applying the QSWT $S$ to the IR sector produces the derivative operator at scale $r=2$, which has size $N/4$. Continuing this process up to scale $r=k$ yields a multi-scale representation of the SLAC derivative operator. This SLAC operator (at scale $k$) in the multi-scaled basis can be represented by a block-diagonal matrix, where the smallest block corresponds to $\Delta^{IR}_{(r=k)}$, and the remaining blocks are UV components whose sizes grow by successive powers of two as the scale decreases:
\begin{equation}\label{mslac}
\begin{aligned}
&\Delta_{(r=k)}=S_{(r=k)} \cdots S_{(r=0)} \, \Delta_{(r=0)} \, S^\dagger_{(r=0)} \cdots S^\dagger_{(r=k)}\\& \quad\quad\quad=
\begin{bmatrix}
\Delta^{\text{IR}}_{(r=k)} & 0 & 0 & \cdots & 0 \\
0 & \Delta^{\text{UV}}_{(r=k)} & 0 & \cdots & 0 \\
0 & 0 & \Delta^{\text{UV}}_{(r=k-1)} & \cdots & 0 \\
\vdots & \vdots & \vdots & \ddots & \vdots \\
0 & 0 & 0 & \cdots & \Delta^{\text{UV}}_{(r=1)}
\end{bmatrix}.
\end{aligned}
\end{equation}

In Fig.~\ref{fig:multi-scaledlap} we present the matrix plot of a 512-dimensional SLAC Laplacian obtained through recursive applications of the Shannon wavelet transform. The resulting multi-scale representation consists of UV blocks of sizes $256,128,64,\ldots,2$ together with an IR block of size $2$. Since the block dimension is halved at each scale, the Laplacian admits a multi-scale decomposition up to scale $r=8$. More generally, a SLAC derivative operator of size $N$ can be represented as a multi-scale matrix up to scale $r=\log_2 N - 1 = n-1$, where at each scale $r$ one obtains a block of size $2^{\,n-r}\times 2^{\,n-r}$.
\begin{figure*}[htb!]
	\centering
            \includegraphics[width=\textwidth]{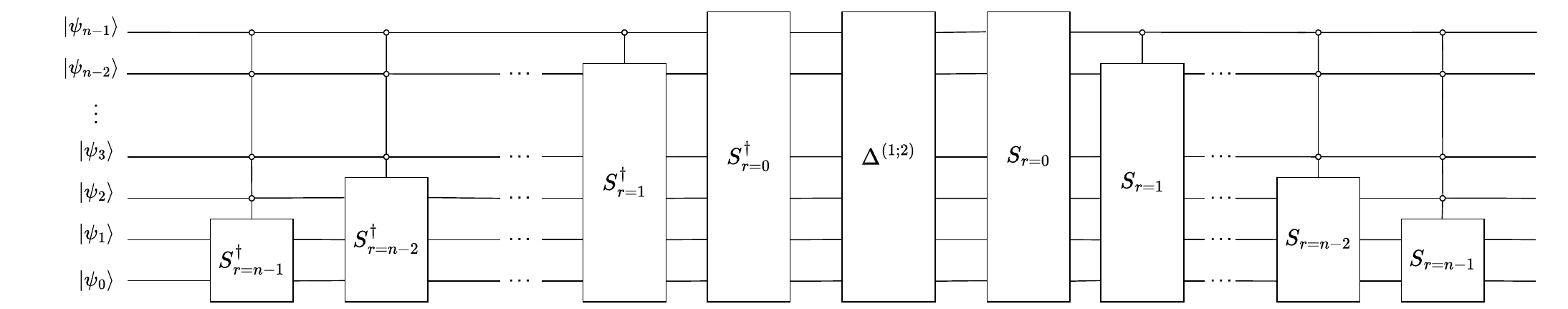}
		\caption{Quantum circuit for recursively applying the wavelet transform at each scale to the SLAC derivative operator (first-order/Laplacian), using its block-encoding together with the construction of $S \equiv \mathrm{QSWT}$. Ancilla qubits associated with the $\mathrm{LCU}$ block-encoding of the derivative operator are omitted for clarity. The circuit realises the multi-scale SLAC derivative operator (at sacle $n-1$) acting on the input state $\ket{\psi}$.}
		\label{fig:msslackt}
\end{figure*}

To implement the multi-scale SLAC operator on a quantum computer, we use two ingredients as subroutines: (i) a block-encoding of the SLAC derivative at a given scale (Section~\ref{subsec:preplap} and \ref{subsec:analyticallapbe}), and (ii) the QSWT that effects perfect IR/UV separation at that scale (Section~\ref{sec:qswt}). At each stage $r$, the QSWT $S_{(r)}$ is applied only on the leading $2^{\,n-r}$-dimensional active subspace (leaving the remaining qubits untouched), thereby block-diagonalising the current SLAC block into its IR/UV components. 
The multi-scale representation can therefore be implemented coherently as:
\begin{align}\label{multiqswt}
\Delta_{(r=k)} \ket{\psi} 
&= (S_{(r=k)} \oplus I_{N - 2^{n - k}})
\cdots (S_{(r=1)} \oplus I_{N/2})\nonumber \\
&\quad S_{(r=0)} \Delta_{(r=0)} 
S^\dagger_{(r=0)} \nonumber \\
&\quad (S^\dagger_{(r=1)} \oplus I_{N/2}) 
\cdots (S^\dagger_{(r=k)} \oplus I_{N - 2^{n - k}})
\ket{\psi}\,.
\end{align}
Here, $I_k$ denotes the identity matrix of dimension $k$. The sequence of $\mathrm{QSWT} \equiv S$ blocks and identity matrices can be realised using controlled-on-$0$ versions of the QSWT. 

At scale $r=0$, the $\mathrm{QSWT}$ and its inverse appear as a similarity transform applied to the block-encoding of the SLAC derivative across all qubits. At scale $r=1$, the wavelet transform acts only on the first $n-1$ qubits, controlled on the $n^{\text{th}}$ qubit being in state $\ket{0}$, so that the half-sized $\mathrm{QSWT}$ is restricted to the IR subspace of scale $r=0$. This recursive procedure continues: at each successive scale $r$, a halved $\mathrm{QSWT}$ is applied to one fewer qubit, with the control-on-$0$ condition shifting down by one qubit each time, until reaching scale $r=n-1$. The complete quantum circuit implementing the multi-scale representation of the SLAC derivative operators is shown in Fig.~\ref{fig:msslackt}.

Based on the preceding discussion, we summarise the complexity of implementing the multi-scale SLAC operator from its block-encoding in the following result.
\begin{result}\label{re:msslac}
The multiscale representation, up to scale \(n-1\), of an \(N\)-dimensional SLAC derivative operator, for derivative orders \(k=1,2\), can be realised by combining a single application of its block-encoding with \(2(n-1)\) applications of the quantum Shannon wavelet transform (\(\mathrm{QSWT}\)). The block-encoding can be implemented using \(\mathcal{O}(n^2)\) gates, and each application of \(\mathrm{QSWT}\) likewise has gate complexity \(\mathcal{O}(n^2)\). Thus, the query complexity of the construction is \(\mathcal{O}(n)\) calls to the \(\mathrm{QSWT}\) routine together with one call to the block-encoding of the SLAC operator, while each primitive call is implementable with \(\mathcal{O}(n^2)\) gates, yielding an overall gate complexity of \(\mathcal{O}(n^3)\).
\end{result}

\section{Implementing linear combination of SLAC derivatives} \label{sec: slaclc}
When solving partial differential equations, one often needs to block-encode linear combinations of derivative operators of different orders. If these derivatives are discretised within the SLAC framework, this naturally leads to the problem of block-encoding linear combinations of SLAC derivative operators. In particular, we can use the block-encodings constructed in the preceding sections: \(U_{\mathrm{slac}}^{(1)}\) for the first-order SLAC derivative and \(U_{\mathrm{slac}}^{(2)}\) for the SLAC Laplacian. In this section, we first present a simple construction for implementing the sum of these two operators, and then extend the method to general linear combinations of SLAC derivative operators in Appendix~\ref{app:sumlcu}.

\subsection{Implementing sum of SLAC first order derivative and Laplacian}
We first consider the special case where each block-encoding of the first-order SLAC derivative and the Laplacian is used once. A block-encoding of sum of SLAC derivative operators, such as $\alpha \tilde{\Delta}^{(1)} + \beta \tilde{\Delta}^{(2)}$, can be constructed using a straightforward method outlined in~\cite{PDOpaper}. This involves invoking the explicit block-encodings derived in Eq.~\eqref{slac1be} and Eq.~\eqref{slac2be}. In these expressions, the first- and second-order SLAC derivative operators are rescaled by the factors $p=\alpha^{(1)}=\mathcal{O}(n)$ and $q=\alpha^{(2)}=\mathcal{O}(1)$, respectively. The procedure begins by preparing the normalised state
\begin{equation}\label{phi}
    \ket{\phi} = \frac{1}{\sqrt{\alpha p+\beta q}} \left(\sqrt{\alpha p}\;\ket{0} + \sqrt{\beta q}\;\ket{1}\right)\,.
\end{equation}

With a single invocation of each block-encoding of the SLAC derivatives, one can build a controlled unitary of the form,
\begin{equation}
    U = \ket{0}\bra{0}\otimes U_\mathrm{slac}^{(1)} + \ket{1}\bra{1} \otimes U_\mathrm{slac}^{(2)}\,.
\end{equation}

Within the $\mathrm{LCU}$ framework, one first prepares the state $\ket{\phi}$, applies the controlled unitary defined above, and then uncomputes the state. Under this procedure, it can be verified that:
\begin{align}
&\left( \bra{\phi} \otimes \mathbf{1}_n \right) 
U \left( \ket{\phi} \otimes \mathbf{1}_n \right)\;\; \ket{\psi}_{sys}\nonumber \\
&= \left( \bra{\phi} \otimes \mathbf{1}_n \right) \frac1{\sqrt{\alpha p+\beta q}}
\left( {\sqrt{\alpha p}}\,\ket{0} U^{1} 
+ {\sqrt{\beta q}}\,\ket{1} U^{2} \right)\ket\psi \nonumber \\
&= \frac{1}{\alpha p + \beta q} \left(\alpha \tilde{\Delta}^{(1)} + \beta \tilde{\Delta}^{(2)} \right) \ket\psi\,.
\end{align}
Note that in the above equation we have omitted the ancillas required for state preparation in each block-encoding. By introducing one additional ancilla (prepared using a $CROT$ operation of Eq.~\eqref{phi}) and making a single call to the block-encodings of the Laplacian and the first-order SLAC derivative operator, it is possible to implement their linear combination with an overall scaling factor of ${\alpha p + \beta q}$ and with both ancilla overhead and gate complexity of $\mathcal{O}(\mathrm{poly}(n))$. Since $p \equiv \mathcal{O}(n)$, the block-encoding of any arbitrary linear combination of SLAC derivative operators satisfies the requirements for a ``good'' block-encoding (see Definition~\ref{df1}). As discussed earlier, the raw error in the block-encoded first-order derivative relative to its ideal counterpart scales slightly worse. Therefore, when implementing the sum of these two operators, one can project out the high-energy modes as in Eq.~\eqref{errorfo}, in which case the dominant error in realising this block-encoded sum is $O(1/N)$.

To realise an arbitrary linear combination of SLAC derivative operators of the form $A = \sum_{j=0}^{m-1} y_j A_j^{(k)}$, with $y \in \mathbf{C}^m$ such that $||y||_1 \leq \beta$ and where $A_j^{(k)}$ denotes the $k^{\text{th}}$-order SLAC derivative operator, one can adapt the $QSVT$ framework introduced in~\cite{qsvt} to implement a linear combination of block-encodings of SLAC derivative operators. A brief discussion of this construction is provided in Appendix~\ref{app:sumlcu}.

 \section{Preconditioning the multi-scaled SLAC derivative operator}\label{sec:precon}
One of the potential applications of the efficient block-encodings constructed in this work is the quantum solution of partial differential equations using the HHL algorithm or QLSA~\cite{HHL}. As established in prior works, the complexity of inverting a matrix using the HHL algorithm scales linearly with the condition number $\kappa$ of the matrix. Therefore, in order to preserve the exponential speedup promised by HHL, it is essential to reduce the condition number to at least $O(\log N)$. If the condition number is large, the matrix is referred to as ill-conditioned, and such matrices must be preconditioned to reduce $\kappa$ prior to applying HHL.

Since the eigenvalues of the SLAC operators exactly correspond to their dispersion relations, it follows that the condition number of the block-encoded SLAC operators scales proportionally with the matrix dimension, i.e., as $\mathcal{O}(N)$. As discussed earlier, however, the SLAC operators also admit a multi-scale representation that can be implemented efficiently within the $\mathrm{QSWT}$ framework presented in Section~\ref{sec:multisc}. This representation provides a key advantage, as the wavelet basis naturally supports the use of wavelet-based preconditioning strategies. In particular, the multi-scale SLAC derivative operator introduced in Section~\ref{sec:multisc} can be effectively preconditioned using a diagonal preconditioner—following the approach outlined in~\cite{fspde}—with block-diagonal elements of the form $1/2^r$, where $r$ labels the scale index in the multi-scale decomposition.

\subsection{Implementing the diagonal wavelet preconditoner}
We now briefly revisit the construction of the diagonal wavelet preconditioner and analyze its effect on the multi-scale SLAC operator. For further details on the application of diagonal preconditioners to differential operators represented in a wavelet basis, the reader is referred to~\cite{fspde}. The diagonal preconditioner matrix for a system of size $N = 2^n$ is given by:  
\begin{equation}\label{precondM}
    P = \begin{bmatrix}
\mathbb{1}_1 & & & & \\
& \frac{1}{2^1} \mathbb{1}_1 & & & \\
& & \frac{1}{2^2} \mathbb{1}_2 & & \\
& & & \frac{1}{2^3} \mathbb{1}_3 & \\
& & & & \ddots \\
& & & & & \frac{1}{2^{n-1}} \mathbb{1}_{n-1}
\end{bmatrix}\,.
\end{equation}
where $\mathbb{1}_1$ is $2-$ dimensional identity matrix. The preconditioner matrix $P$ can be expressed as a linear combination of two unitary matrices, $U^{\pm}$, such that $P = (U^{+} + U^{-})/2$, where
\begin{equation}\label{preU}
U^{\pm} := P \pm i \sqrt{\mathbb{1} - P^2} = e^{\pm i \arccos P}\,.
\end{equation}
Here, $U^{\pm}$ are diagonal matrices with a structure analogous to that of the preconditioner itself, whose diagonal entries take the form $e^{\pm i \theta}$ for some angle $\theta$. 

As noted in~\cite{fspde}, such unitaries correspond to a sequence of controlled rotations and can be represented as follows:
\begin{equation}\label{preconU}
U^{\pm} = \prod_{r=1}^{n-1} \Lambda_{0}^{r-1} \left( R_z(\pm \theta_{n-r}) \right) \otimes \mathbb{1}_{n-r},
\end{equation}
where $\Lambda_{0}^{r} \left( R_z(\theta)\right)$ denotes $\ket{0^r}-$ controlled-$R_z(\theta)=exp(i\theta Z)$ rotation gates, with $\theta_r=\arccos(1/2^r)$. 

As, stated in {Lemma 2} of~\cite{fspde} these unitaries can be implemented using $O(n)$ Toffoli, one and two-qubit gates with $n$ additional ancilla qubits. Using these unitary gates, one can implement a $(1,1,0)-$ block-encoding (using the $\mathrm{LCU}$ technique) for the preconditioner as:
\begin{equation}\label{preconbe}
U_P := (H \otimes \mathbb{1}_n) \Lambda_0(U^+) \Lambda_1(U^-) (H \otimes \mathbb{1}_n)
\end{equation}
with a gate cost of $O(n)$.

The Shannon wavelet transform block-diagonalises the exact SLAC Laplacian into dyadic momentum bands, which form the multiscale SLAC operator described in Section~\ref{sec:multisc}. 
We denote the resulting multiscale SLAC Laplacian by $\Delta_{\mathrm{MS}}$. 
After projecting out the zero mode, the $\ell$-th scale contains Fourier modes whose integer labels satisfy,
\[
2^{\ell-1}\le |j| < 2^\ell ,
\]
with the upper endpoint included in the finest Nyquist band. 
The corresponding physical momenta (after projecting our the nullspace) are,
\[
k_j=\frac{2\pi j}{N} , \quad j\in\left\{-\frac{N}{2}+1,\ldots,\frac{N}{2}\right\}\setminus\{0\}. 
\]
Since the exact SLAC Laplacian has Fourier symbol proportional to $k_j^2$, its eigenvalue magnitudes are proportional to
\[
\lambda_j=\left(\frac{2\pi j}{N}\right)^2.
\]
Equivalently, we define the rescaled eigenvalues
\[
\bar{\lambda}_j
:=
\left(\frac{N}{2\pi}\right)^2 |\lambda_j|
=
j^2 .
\]
Thus, in the $\ell$-th dyadic block, the rescaled eigenvalue magnitudes satisfy
\[
2^{2(\ell-1)}
\le
\bar{\lambda}_j
\le
2^{2\ell}.
\]

The diagonal preconditioner in Eq.~\eqref{precondM} assigns weight $w_\ell=2^{-\ell}$ to this scale. 
Thus, under the symmetric preconditioning sandwich $P\Delta_{\mathrm{MS}}P$, the rescaled eigenvalues in the $\ell$-th block are transformed as
\[
w_\ell^2\bar{\lambda}_j
=
2^{-2\ell}\bar{\lambda}_j .
\]
Consequently,
\[
2^{-2\ell}2^{2(\ell-1)}
\le
w_\ell^2\bar{\lambda}_j
\le
2^{-2\ell}2^{2\ell},
\]
or equivalently
\[
\frac{1}{4}
\le
w_\ell^2\bar{\lambda}_j
\le
1 .
\]

The bound is saturated on the finest dyadic band. 
For $N=2^n$, the finest scale $\ell=n-1$ contains the edge modes $|j|=2^{n-2}$ and $|j|=2^{n-1}$. 
After applying the preconditioner weight $w_{n-1}=2^{-(n-1)}$, these modes give
\[
2^{-2(n-1)}2^{2(n-2)}=\frac{1}{4},
\qquad
2^{-2(n-1)}2^{2(n-1)}=1.
\]
Therefore, the rescaled nonzero spectrum of the preconditioned exact SLAC Laplacian lies in the interval $[1/4,1]$ and attains both endpoints. 
The unscaled eigenvalues differ only by the common factor $(2\pi/N)^2$, which cancels in the condition number. 
Hence, after projecting out the nullspace, the condition number of the preconditioned SLAC Laplacian is $\kappa_p=4.$

For the truncated SLAC Laplacian's block-encoding, the same argument applies up to the truncation error.

For the first-order derivative, the Fourier symbol scales as $|k|$ rather than $k^2$. 
Thus, on the $\ell$-th dyadic band, the eigenvalues scale as $2^\ell$, and the corresponding symmetric preconditioner must use weights $w_\ell=2^{-\ell/2}$, giving a bounded condition number with ratio at most $2$. Hence, the preconditioner in Eq.~\eqref{precondM} must be modified by taking the square root of its diagonal entries. This modified preconditioner can be implemented with the same asymptotic gate cost by adjusting the angles \(\theta_r\) in Eq.~\eqref{preconU}. Under this square-root preconditioning, the first-order SLAC derivative attains an effective condition number of approximately \(2\).

Thus, under the action of the preconditioner in Eq.\eqref{precondM} (and its modified version) the condition number of the truncated SLAC Laplacian (and first order derivative) remains $\mathcal{O}(1)$.

\subsection{Block-encoding for the preconditoned multi-scaled SLAC derivative operator}
Building on the construction of the multi-scale SLAC derivative operator presented in Section~\ref{sec:multisc}, the preconditioned SLAC derivative operator can be realized by making a single call to the block-encoding of the derivative operator, $2 \times (n-1)$ applications of the controlled QSWT operations, and two calls to the block-encoding of the preconditioner matrix, $U_p$.

Let us denote the recursive sequence of $\mathrm{QSWT}$ applications in Eq.~\eqref{multiqswt} by $W$. The $(\alpha,a,\varepsilon)$ block-encoding of the preconditioned SLAC operator, $U_{A_p}$, can then be formally expressed as
\begin{equation}\label{preconAbe}
U_{A_p} := (\mathbb{1}_{a-1} \otimes U_P)(\mathbb{1}_a \otimes W) U_{\mathrm{slac}} (\mathbb{1}_a \otimes W^\dagger)(\mathbb{1}_{a-1} \otimes U_P),
\end{equation}
where the parameters $\alpha$ and $a$ depend on the order of the derivative operator. The gate cost of the constituent primitive calls is determined by $U_{\mathrm{slac}}$ and the elementary $\mathrm{QSWT}$ operations, each of which can be implemented using $\mathcal{O}(n^2)$ gates.

It is important to note that the smallest block at the coarsest scale $r=n-1$ of the multi-scale SLAC operator originates from the IR contribution and contains an eigenvalue associated with the nullspace of the operator. Consequently, even after applying the preconditioner, the matrix remains ill-conditioned unless the eigenvector corresponding to this null mode, $\ket{\psi_{\mathrm{null}}}$, is projected out. For an $n$-qubit SLAC derivative operator, the eigenvector of this nullspace mode is given by $\ket{\psi_{\mathrm{null}}} = H_0 \ket{0}^{\otimes n}$.

Since the nullspace of the SLAC derivative operator is known analytically, the right-hand side vector of the linear system, \( \vec{b} \), can be projected entirely outside this nullspace; see Appendix~\ref{app:nullspaceIR}. Consequently, the effective condition number of the projected and preconditioned SLAC Laplacian becomes independent of the system size, with \( \kappa' = 4 \).

As further benchmarks, we consider the following four second-order elliptic operators:
\begin{equation}
\begin{aligned}
\mathcal L_1 &= \partial_x^2-\partial_x+1,\\
\mathcal L_2 &= -\partial_x\!\left(c_a\,\partial_x\right)+e^x,\\
\mathcal L_3 &= -\partial_x^2+V,\\
\mathcal L_4 &= -\partial_x\!\left(a_\epsilon\,\partial_x\right)+1,
\end{aligned}
\label{eq:benchmark_ops}
\end{equation}

where
\[
\begin{aligned}
c_a(x) &= \cosh(x/a); \quad \mathrm{for}\,\,\,a=4,\\
V(x) &= 1+\sin^2(2\pi x),\\
a_\epsilon(x) &= 1+\epsilon\cos(2\pi x),
\qquad 0<\epsilon<1 .
\end{aligned}
\]

Here, \(\mathcal L_1\) is a convection--diffusion--reaction operator, \(\mathcal L_2\) is a slowly varying Sturm--Liouville operator, \(\mathcal L_3\) is a Schr\"odinger-type elliptic operator with a smooth periodic potential, and \(\mathcal L_4\) is a periodic variable-coefficient diffusion operator. The operators \(\mathcal L_1\) and \(\mathcal L_2\) were also considered in Ref.~\cite{fspde}. The ellipticity of these benchmark operators is established in Appendix~\ref{app:ellip}, following the framework of Ref.~\cite{EvansPDE}.

The numerical results shown in Fig.~\ref{fig:slac_precond} indicate that the preconditioner defined in Eq.~\eqref{precondM} performs well for these operators when they are discretised using the SLAC representation. Moreover, linear combinations of such SLAC operators can be implemented within the framework described in Appendix~\ref{app:sumlcu}.

This ensures that $\mathrm{QLSA}$~\cite{Childshhl,HHL} can be applied without any loss of its exponential speedup due to ill-conditioning, as explained in the consequent section. 

\subsection{Cost of obtaining block-encoding of the preconditioned SLAC operators}
The cost of each step for implementing the preconditioned matrix in the multi-scale Shannon wavelet basis, as outlined in this section, is summarized below.
\begin{itemize}
   \item  The SLAC derivative operators, denoted as $A = \Delta^{(k)}_{\mathrm{SLAC}}$, are accessible via the block-encodings defined in Eqs.~\eqref{beslac1} and~\eqref{beslac2}. For an $n$-qubit system, these encodings are realized with a gate complexity of $\mathcal{O}(n^2)$ using the LCU architectures detailed in Sections~\ref{sec: slaclap} and~\ref{sec: slac1be}. Methods for implementing linear combinations of these operators are further elaborated in Section~\ref{sec: slaclc}.
   
   \item Leveraging the natural representation of the SLAC derivative in the Shannon wavelet basis, the operator is mapped to a multi-scale basis through recursive applications of the QSWT; see Section~\ref{sec:multisc}. The implementation requires one query to the block-encoding of $A$ and $2(n-1)$ applications of the $\mathrm{QSWT}$. Each primitive call, namely one application of the SLAC block-encoding or one application of $\mathrm{QSWT}$, has gate complexity $\mathcal{O}(n^2)$, as established in Sections~\ref{sec: slaclap},~\ref{sec: slac1be} and~\ref{sec:qswt} respectively.

   To optimise the controlled operations appearing in the recursive protocol, one can employ a sawtooth, or ladder, decomposition for the multi-controlled logic, using available garbage ancillae so that the control-specific Toffoli count scales linearly with $n$. Consequently, the multi-scale matrix $A_w = W A W^\dagger$ can be realised using $\mathcal{O}(n)$ calls to the $\mathrm{QSWT}$ protocol and one call to the block-encoding of the corresponding SLAC operator, with each primitive call implementable using $\mathcal{O}(n^2)$ gates.

   \item The multi-scaled matrix can be preconditioned by applying the diagonal wavelet preconditioner twice. Its block-encoding is given in Eq.~\eqref{preconbe} and is constructed using the diagonal unitaries defined in Eq.~\eqref{preU}, with gate cost $\mathcal{O}(n)$. Consequently, the block-encoding of the preconditioned multi-scale SLAC derivative operator, $A_p=P A_w P^\dagger$, can be obtained as described in Eq.~\eqref{preconAbe} by accessing the block-encoding of the SLAC derivative together with $\mathcal{O}(n)$ applications of the $\mathrm{QSWT}$ protocol. Each primitive call to either the SLAC block-encoding or the $\mathrm{QSWT}$ protocol has gate complexity $\mathcal{O}(n^2)$.
\end{itemize} 

\section{Inverting the SLAC derivative operator using QLSA}\label{sec: HHL}
In this section, we revisit the method introduced in~\cite{fspde} for efficiently solving PDEs involving SLAC derivative operators. We first formulate the problem for PDEs whose derivatives are discretised according to the SLAC prescription. We then show how the block-encoding of the preconditioned SLAC operator, constructed in the previous section, can be inverted using \(\mathrm{QLSA}\) to efficiently solve the resulting linear system, after projecting out the nullspace of the operator.
\begin{figure}[hbtb!]
            \includegraphics[width=\columnwidth]{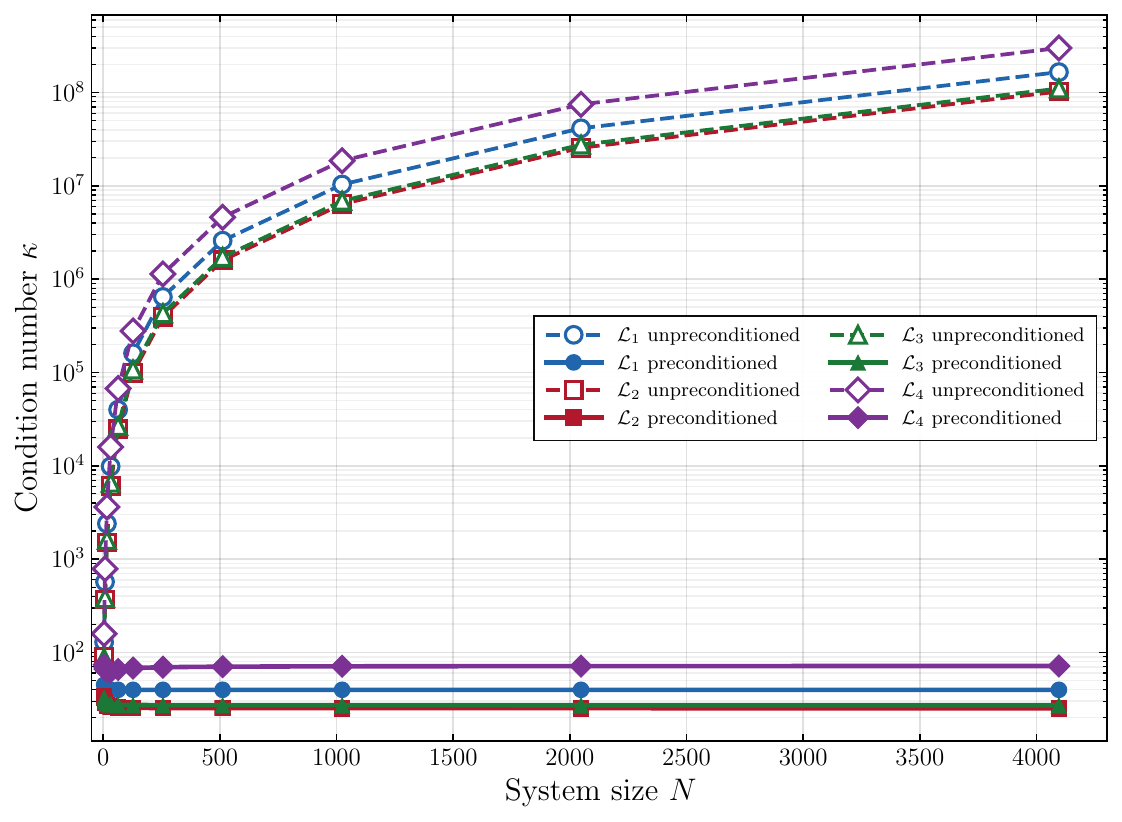}
\caption{
    The condition numbers for the four SLAC-discretised benchmark operators in Eq.~\eqref{eq:benchmark_ops}. Dashed curves with open markers denote the unpreconditioned systems, while solid curves with filled markers denote the preconditioned systems. The operators are discretised using periodic SLAC derivatives on the fixed domain \(\Omega=[0,1]\). Note the nullspace is trivial for all four operators. The unpreconditioned matrix condition numbers grow with the system size \(N\), whereas the preconditioned matrix condition numbers remain nearly constant, showing that the multiscale wavelet preconditioner suppresses the dominant scale dependence.
    }
    \label{fig:slac_precond}
\end{figure}
\subsection{Formulating the problem of SLAC derivative operator inversion}
We now outline the class of partial differential equations (PDEs) for which SLAC-based discretizations, combined with wavelet-based preconditioning, yield an efficient computational framework. Our formulation follows the notation and conventions of~\cite{fspde}.

Consider a $d$-dimensional linear elliptic PDE defined over a bounded domain \( \Omega \subset \mathbb{R}^d \), expressed as
\[
\mathcal{L} u(\vec{x}) = b(\vec{x}),
\]
where \( \vec{x} \in \Omega \), \( u(\vec{x}) \) denotes the unknown function, \( b(\vec{x}) \) represents the source term or inhomogeneity, and \( \mathcal{L} \) is a linear differential operator.

In our framework, \( \mathcal{L} \) is discretized using SLAC derivative operators. Specifically, for a linear PDE of order \( m \), the operator \( \mathcal{L} \) can be written as:
\[
\mathcal{L} = \sum_{\vec{\alpha}} c_{\vec{\alpha}}(\vec{x}) \, \partial^{\vec{\alpha}}.
\]
Here \( \vec{\alpha} = (\alpha_1, \ldots, \alpha_d) \) is a multi-index with \( |\vec{\alpha}| := \sum_{i=1}^d \alpha_i \le m \), and \( \partial^{\vec{\alpha}} := \partial_1^{\alpha_1}\cdots \partial_d^{\alpha_d} \). In our formulation, each partial derivative \( \partial_i^{\alpha_i} \) is replaced by its SLAC-discretized counterpart \( \partial_{\mathrm{slac};i}^{(\alpha_i)} \), yielding a discrete operator \( \mathcal{L}_{{slac}} \in \mathbb{C}^{N\times N} \). Here \( \partial_{\mathrm{slac};i}^{(\alpha_i)} \) denotes the SLAC partial derivative of order \( \alpha_i \) with respect to \( x_i \).

If the leading-order coefficients satisfy \( \sum_{|\vec{\alpha}| = m} c_{\vec{\alpha}}(\vec{x})\, \vec{\xi}^{\vec{\alpha}} \neq 0 \) for all nonzero \( \vec{\xi} \in \mathbb{R}^d \) and for all \( \vec{x} \in \Omega \), then the operator \( \mathcal{L} \) is said to be elliptic, and the corresponding PDE is classified as elliptic.

\paragraph{The solution state.} In quantum algorithms for solving linear systems \( A u = b \), the objective is typically to prepare a quantum state (not necessarily \( |u\rangle \)) that enables the estimation of expectation values \( \langle u| M |u\rangle \) for observables \( M \) of interest. In our approach, following the framework of~\cite{fspde}, we do not explicitly construct the state \( |u\rangle \); instead, we prepare an alternative state \( |\psi\rangle \) that encodes the same observable information. We refer to this state \( |\psi\rangle \) as the \emph{solution state} throughout.

To construct \( |\psi\rangle \), we first transform the linear system \( A u = b \) into a multi-scale wavelet basis, where \( A = \Delta^{(m)}_{\mathrm{slac}} \) denotes the SLAC derivative operator of order \( m = 1, 2 \). This basis transformation is performed through recursive applications of the quantum Shannon wavelet transform (QSWT), as outlined in Section~\ref{sec:multisc}, resulting in a multi-scale representation \( A_w \) of \( A \). Each block of \( A_w \) corresponds to a circulant matrix at a specific scale and can be diagonalized using a scale-dependent controlled QFT (QFT).
Let \( W \) denote the fully recursive QSWT circuit. The linear system in the multi-scale basis is then given by \( A_w u_w = b_w \), where \( A_w := W A W^\dagger \), \( u_w := W u \), and \( b_w := W b \). 

At this stage, the multi-scaled SLAC operator \( A_w \) of size \( N \) exhibits a condition number that scales as \( \mathcal{O}(N) \), making it ill-conditioned. To address this, we implement a diagonal wavelet preconditioner \( P \), which rescales each block of \( A_w \) by a factor of \( 1/2^r \) corresponding to its scale \( r \). This yields the preconditioned system \( A_p u_p = b_p \), with \( A_p := P A_w P \), \( u_p := P^{-1} u_w \), and \( b_p := P b_w \), as detailed in Section~\ref{sec:precon}.

Finally, to ensure numerical stability and avoid inversion over the nullspace of \( A_p \), we project the vector \( b_w \) onto the subspace orthogonal to the nullspace of \( A_p \). This projection guarantees that the preconditioned matrix retains a bounded, constant condition number, ensuring that the inverse operation executed by quantum linear solvers such as HHL remains well-defined, stable, and efficient. The details of this nullspace projection are provided in Appendix~\ref{app:nullspaceIR}.

We define the solution state as:
\begin{equation}\label{sol}
    \ket{\psi_{ab}} := W^\dagger U^a A_p^{-1} U^b W \; \ket{b} \quad \forall a,b \in \mathbb{B} := \{0,1\},
\end{equation}
where \( U^{0/1} = U^{+/-} \) are as defined in Eq.~\eqref{preU}. The corresponding observable is given by:
\begin{equation}
M' := \sum_{abcd \in \mathbb{B}} \ket{ab} \bra{cd} \otimes M\,.
\end{equation}
As shown in~\cite{fspde}, the expectation value \( \bra{\psi} M' \ket{\psi} \) approximates \( \bra{u} M \ket{u} \), where \( \ket{u} = W^{-1} P A_p^{-1} P W \ket{b} \) represents the true solution state of the PDE. The implementation of \( M' \) follows the procedure outlined in Appendix~A of~\cite{fspde}.

With this correspondence established, the following subsection focuses on obtaining the solution state defined in Eq.~\eqref{sol} by utilizing the block-encodings of the preconditioned SLAC derivative operators and analyzing the resource cost associated with each step required to construct this state for PDEs discretized using SLAC derivatives.

\subsection{Preparing the solution state and its cost}
As discussed earlier, the solution state for a PDE discretized using the multi-scale SLAC formalism is obtained by applying the inverse of the preconditioned matrix to the state \( \ket{b} \). Prior to performing this inversion, it is necessary to construct the block-encoding of the preconditioned matrix with the steps shown in previous section \ref{sec:precon}.

In~\cite{fspde}, the inverse of the preconditioned matrix was block encoded by first introducing an operation denoted as \( \textbf{QMI} \):
\begin{equation}
\textbf{QMI: }\;\ket{0}_{\text{flag}} \ket{\psi_{\text{in}}} \mapsto \ket{0}_{\text{flag}} \frac{1}{\alpha} \widetilde{A}_p^{-1} \ket{\psi_{\text{in}}} + \ket{\phi},
\end{equation}
where \( \widetilde{A}_p^{-1} \) represents an \( \varepsilon \)-approximation of \( A_p^{-1} \) in the operator norm.

As outlined in the previous subsection, it is necessary to project out the nullspace of the preconditioned matrix using the controlled swap test described in Appendix~\ref{app:nullspaceIR}, ensuring that the eigenvalues of \( A_p \) lie within the bounded range \( [1/\kappa, 1] \) with \( \kappa = \mathcal{O}(1) \). From Eq.~\eqref{applyv}, we observe that the subspace orthogonal to the vector \( \vec{b} \) is indicated by the flag \( F = 1 \). To perform matrix inversion within this subspace, we define a controlled unitary built from the \( \textbf{QMI} \) operation, which applies \( \textbf{QMI} \) selectively on the orthogonal component.
\begin{equation}
  \widetilde{QMI}
  \;=\;
  \ket{0}\bra{0}_{F} \otimes I_{R} \otimes I_{S}
  \;+\;
  \ket{1}\bra{1}_{F} \otimes I_{R} \otimes QMI .
\end{equation}
Thus, to invert the SLAC derivative operator only within the subspace orthogonal to its nullspace, we employ the above operation and subsequently trace out the register labeled \( R \).

As described in Appendix~E2 of~\cite{fspde}, the $\mathbf{QMI}$ procedure consists of a quantum phase estimation ($\mathrm{QPE}$) routine applied to the preconditioned matrix $A_p$, followed by controlled rotations on an ancilla qubit that encode the reciprocals of the estimated eigenvalues. The gate complexity of the $\mathrm{QPE}$ and $\mathrm{CROT}$ steps in the matrix-inversion routine is $\mathcal{O}\bigl(\log(\kappa_p / \varepsilon)\bigr)$. According to~\cite{fspde}, the query complexity of implementing $\mathbf{QMI}$ is $\mathcal{O}\bigl((\kappa_p / \varepsilon)\log(\kappa_p / \varepsilon)\bigr)$. When $A_p$ is accessed through a block-encoding with subnormalisation factor $\alpha^{(k)}$, this query cost acquires an additional multiplicative overhead proportional to the normalisation ratio $\alpha^{(k)} / \|A_p\|$. Accordingly, the overall query complexity for inverting the preconditioned matrix becomes $\mathcal{O}\bigl(\alpha^{(k)} (\kappa_p / \varepsilon)\log(\kappa_p / \varepsilon)\bigr)$ as $\|A_p\|=\mathcal{O}(1)$. 

By instead employing Quantum Singular Value Transformation (QSVT) (see~\cite{qsvt}), one obtains an exponential improvement in the dependence on the precision parameter $\varepsilon$. In this setting, the inverse of the preconditioned matrix can be block-encoded using $\mathcal{O}\bigl(\alpha^{(k)}\kappa_p \log(\kappa_p / \varepsilon)\bigr)$ queries to the block-encoding of $A_p$, together with $\mathcal{O}(n)$ applications of the $\mathrm{QSWT}$ protocol required to reach the multi-scale basis.

However, after projecting out the nullspace, the preconditioned matrix in the multi-scale basis satisfies $\kappa_p=\mathcal{O}(1)$. Therefore, after nullspace projection, the inverse of the preconditioned matrix can be block-encoded using $\mathcal{O}(\alpha^{(k)}\log(1/\varepsilon))$ queries to the block-encoding of $A_p$. Each query to this block-encoding consists of one call to the SLAC block-encoding together with $\mathcal{O}(n)$ applications of the $\mathrm{QSWT}$ protocol. Each primitive call to either the SLAC block-encoding or the $\mathrm{QSWT}$ protocol can be implemented using $\mathcal{O}(n^2)$ gates.

The formalism described above can be naturally extended to a \( d \)-dimensional ( \( d\mathrm{D} \) ) PDE system by following the procedure outlined in ~\cite{fspde}. This involves first defining a \( d\mathrm{D} \) preconditioner (refer Appendix D of~\cite{fspde}) that incorporates the \textbf{MAX} operation introduced in Eq.~(13) of~\cite{fspde}. Similarly, the \( d\mathrm{D} \) Shannon wavelet transform, which applies the wavelet transform across all spatial dimensions, can be realized as a tensor product of the one-dimensional recursive QSWT protocol. Consequently, the overall cost of obtaining the solution state for a \( d\mathrm{D} \) PDE scales linearly with the number of dimensions \( d \).

In the $d$-dimensional setting, let $N=2^{nd}$ denote the total number of lattice points.
These points are distributed over a $d$-dimensional grid, so that the number of lattice points along each spatial direction scales as $N^{1/d}$. 
Therefore, the finite-lattice truncation error should be estimated with respect to the resolution per dimension, giving a scaling of order $\mathcal{O}(N^{-1/d})$. 
For example, this scaling is attained for product-state solutions whose one-dimensional components each have truncation error $\mathcal{O}(1/N^{1/d})$.

\subsection{Complexity of solving dD PDEs using SLAC}
Altogether, we summarise the cost of generating a quantum solution for a $d$-dimensional PDE with derivative operators discretised using the SLAC formalism in the following result.

\begin{result}\label{re:SLACPDE}
Consider a $d$-dimensional inhomogeneous linear PDE of the form
$\mathcal{L}u(x)=b(x)$ defined on the domain $[0,1]^d$ with
periodic boundary conditions, where $\mathcal{L}$ is an elliptic
operator discretised using the SLAC formalism. Let
$A\in\mathbb{R}^{N\times N}$ denote the resulting SLAC matrix with
$N=2^{nd}$ points on a $d$-dimensional grid, and let
$Au=b$ be the corresponding linear system. For $\varepsilon>0$,
assume access to:
\begin{itemize}
    \item an $(\alpha,a,\varepsilon)$ block-encoding of $A$ constructed via the Linear Combination of Unitaries ($\mathrm{LCU}$) framework, where the state-preparation circuit is implemented using inequality tests; and
    \item a procedure $\mathcal{P}_b$ that prepares a quantum state $\ket{b}$ encoding the right-hand-side vector $b$.
\end{itemize}
Then, after projecting out the nullspace, an approximate solution state
$\ket{\tilde{u}}$ to the linear system $Au=b$, as defined in~Eq.~\eqref{sol}, with error at most $\varepsilon$, can be prepared using $\mathcal{O}(1)$ invocations of $\mathcal{P}_b$ and $\mathcal{O}(\alpha^{(k)}\log(1/\varepsilon))$ queries to the multiscale preconditioned block-encoding of $A$. Each such query uses one call to the SLAC block-encoding together with $\mathcal{O}(n)$ applications of the $d$-dimensional $\mathrm{QSWT}$ protocol. Each primitive call to the SLAC block-encoding and to the $d$-dimensional $\mathrm{QSWT}$ can be implemented using $\mathcal{O}(dn^2)$ gates. 
\end{result}
In particular, combining this result with Proposition~\ref{corl1} shows that, when the PDE involves an $n$-qubit SLAC Laplacian, we have $\alpha^{(2)}=\mathcal{O}(1)$, which makes the query cost $\mathcal{O}(\log(1/\varepsilon))$. Likewise, Proposition~\ref{corl2} implies that if the PDE involves an $n$-qubit first-order SLAC derivative, the subnormalisation constant scales as $\alpha^{(1)}=\mathcal{O}(n)$, increasing the query cost to $\mathcal{O}(n\log(1/\varepsilon))$. More generally, Result~\ref{re:lcuolcu} shows that, for a PDE involving a linear combination of SLAC operators, the resulting block-encoding has a subnormalisation factor $\alpha^{(k)}$ with the same asymptotic scaling as that of the first-order SLAC derivative.

Each $d$-dimensional $\mathrm{QSWT}$ can be implemented using $\mathcal{O}(d n^2)$ gates. 
Since the multiscale construction requires $\mathcal{O}(n)$ calls to the $\mathrm{QSWT}$, the gate cost of one multiscale block encoding scales as $\mathcal{O}(d n^3)$. 
In the preconditioned $\mathrm{QLSA}$, this block encoding is invoked $\mathcal{O}(\alpha^{(k)}\log(1/\varepsilon))$ times. 
Therefore, the overall gate complexity for solving a $d$-dimensional PDE in the SLAC representation scales as
$
\mathcal{O}\!\left(d n^3 \alpha^{(k)} \log(1/\varepsilon)\right).
$

\section{Conclusions}
In this work, we developed efficient quantum protocols for implementing SLAC derivative operators and their multi-scale representations on a finite lattice of $N = 2^n$ sites. The main challenge addressed here is that SLAC derivatives are highly non-local in position space, resulting in fully dense $N \times N$ matrices, even though they preserve the continuum momentum dispersion exactly within the Brillouin zone. We showed that this non-locality can nevertheless be handled efficiently within the block-encoding framework by combining $\mathrm{LCU}$ methods, nested-box inequality-test state preparation, and modular arithmetic.

We first constructed an explicit quantum implementation of the Shannon wavelet transform. Since the Shannon wavelet transform achieves exact momentum separation, it provides a natural basis for multi-scale representations of SLAC operators. Our construction uses the $\mathrm{QFT}$, modular adders, and simple sparse unitaries, giving an implementation whose gate complexity is polynomial in $n=\log N$. This contrasts with the classical FFT-based cost, which scales polynomially in the lattice size $N$.

The central technical contribution of this work is the explicit block-encoding of the SLAC Laplacian and the first-order SLAC derivative. Both operators are realised through the $\mathrm{LCU}$ framework, where the dense coefficient states are prepared using nested-box inequality tests. The resulting block-encodings have gate complexity $\mathcal{O}(n^2)$ up to error $\varepsilon=\mathcal{O}(1/N)$, with the dominant cost arising from the inequality-test arithmetic and the controlled modular subtractors used in the SELECT oracle. The SLAC Laplacian admits an optimal block-encoding, with a constant normalisation overhead, while the first-order SLAC derivative admits a good block-encoding with subnormalisation scaling as $\mathcal{O}(n)$.

We then showed how these block-encodings can be combined with the Quantum Shannon Wavelet Transform to construct multi-scale SLAC derivative operators. A single call to the block-encoding, together with $\mathcal{O}(n)$ applications of the $\mathrm{QSWT}$ protocol, generates the multi-scale SLAC operator up to scale $n-1$. We also discussed how linear combinations of SLAC derivative operators can be implemented using standard $\mathrm{LCU}$-of-block-encodings techniques, allowing sums of different derivative orders to be treated within the same framework.

Finally, we incorporated diagonal wavelet preconditioning into the SLAC framework. After projecting out the nullspace, the preconditioned multi-scale SLAC operators have improved conditioning and remain efficiently block-encodable. This allows differential equations discretised using SLAC derivatives to be solved using $\mathrm{QLSA}$, provided the right-hand-side state has support on the invertible subspace of the operator. In this way, the block-encodings developed here serve as efficient quantum primitives for solving PDEs involving non-local derivative operators.

The constructions developed in this work are particularly relevant for physical systems where preserving continuum dispersion properties after discretisation is important. This includes condensed-matter models such as Luttinger liquids and topological insulators ~\cite{cmp1,cmp2,cmp3}, as well as lattice gauge-theory settings where SLAC-type derivatives are useful for avoiding fermion doubling. More broadly, the nested-box state-preparation methods used here can be adapted to other dense coefficient functions that vary smoothly with the computational basis index, provided the associated arithmetic can be implemented efficiently.

Overall, this work shows that dense non-local derivative operators arising from continuum-preserving discretisations can be made compatible with efficient quantum block-encoding techniques. By combining analytic structure, tailored state preparation, and wavelet-based multi-scale methods, the approach developed here extends the range of differential and field-theoretic operators that can be implemented efficiently on a quantum computer.

A natural direction for future work is to use these block-encodings as primitives for qubitization-based simulation of continuum condensed-matter systems. This would allow the dynamics of SLAC-discretised Hamiltonians to be simulated directly while retaining the continuum dispersion properties encoded by the SLAC operators. Another important step is to perform a detailed resource estimate for a fault-tolerant implementation of these block-encodings. The main non-Clifford bottlenecks arise from the inequality tests, for which we have provided explicit Toffoli counts scaling quadratically in the register size. This makes the constructions well suited for assessing the resources required to implement SLAC-based simulations on future fault-tolerant quantum architectures.

\section*{Acknowledgments}
GKB thanks Mohsen Bagherimehrab, Daniel J. George, Nicholas Funai, Dominic G. Lewis, Nicolas C. Menicucci, and \v{S}imon Vedl for helpful conversations. DWB worked on this project under a sponsored research agreement with Google Quantum AI. DWB is also supported by Australian Research Council Discovery Project DP220101602. GKB received support from the Australian Research Council Discovery Project DP260104888. RMG thanks Soumya Sarkar for helpful discussions. RMG is supported by the Sydney Quantum Academy. GM thanks Abhijeet Alase for insightful discussions.

\bibliography{references}\bibliographystyle{utphys}

\providecommand{\href}[2]{#2}\begingroup\raggedright\begin{thebibliography}{10}

\bibitem{NIELSEN198120}
H.~Nielsen and M.~Ninomiya, ``Absence of neutrinos on a lattice: (i). proof by homotopy theory,'' \href{http://dx.doi.org/https://doi.org/10.1016/0550-3213(81)90361-8}{{\em Nuclear Physics B} {\bfseries 185} no.~1, (1981) 20--40}. \url{https://www.sciencedirect.com/science/article/pii/0550321381903618}.

\bibitem{finiteswt}
P.~Fries, I.~Reyes, J.~Erdmenger, and H.~Hinrichsen, ``{Renormalization of lattice field theories with infinite-range wavelets},'' \href{http://dx.doi.org/10.1088/1742-5468/ab14d8}{{\em J. Stat. Mech.} {\bfseries 1906} no.~6, (2019) 064001}, \href{http://arxiv.org/abs/1811.05388}{{\ttfamily arXiv:1811.05388 [hep-th]}}.

\bibitem{slac1}
S.~D. Drell, M.~Weinstein, and S.~Yankielowicz, ``Strong-coupling field theories. ii. fermions and gauge fields on a lattice,'' \href{http://dx.doi.org/10.1103/PhysRevD.14.1627}{{\em Phys. Rev. D} {\bfseries 14} (Sep, 1976) 1627--1647}. \url{https://link.aps.org/doi/10.1103/PhysRevD.14.1627}.

\bibitem{Quinn:1986mzs}
H.~R. Quinn and M.~Weinstein, ``Lattice theories of chiral fermions,'' \href{http://dx.doi.org/10.1103/PhysRevD.34.2440}{{\em Phys. Rev. D} {\bfseries 34} (1986) 2440--2450}.

\bibitem{slac2}
J.~P. Costella, ``{A New proposal for the fermion doubling problem. 2. Improving the operators for finite lattices},'' \href{http://arxiv.org/abs/hep-lat/0207015}{{\ttfamily arXiv:hep-lat/0207015}}.

\bibitem{shannon1}
S.~Mallat, {\em A Wavelet Tour of Signal Processing, Third Edition: The Sparse Way}.
\newblock Academic Press, Inc., USA, 3rd~ed., 2008.

\bibitem{shannon4}
M.~Bagherimehrab and A.~Aspuru-Guzik, ``{Efficient quantum algorithm for all quantum wavelet transforms},'' \href{http://dx.doi.org/10.1088/2058-9565/ad3d7f}{{\em Quantum Sci. Technol.} {\bfseries 9} no.~3, (2024) 035010}, \href{http://arxiv.org/abs/2309.09350}{{\ttfamily arXiv:2309.09350 [quant-ph]}}.

\bibitem{qft1}
G.~K. Brennen, P.~Rohde, B.~C. Sanders, and S.~Singh, ``{Multiscale quantum simulation of quantum field theory using wavelets},'' \href{http://dx.doi.org/10.1103/PhysRevA.92.032315}{{\em Phys. Rev. A} {\bfseries 92} no.~3, (2015) 032315}, \href{http://arxiv.org/abs/1412.0750}{{\ttfamily arXiv:1412.0750 [quant-ph]}}.

\bibitem{qft2}
S.~Singh and G.~K. Brennen, ``{Holographic Construction of Quantum Field Theory using Wavelets},'' \href{http://arxiv.org/abs/1606.05068}{{\ttfamily arXiv:1606.05068 [quant-ph]}}.

\bibitem{qft3}
M.~Bagherimehrab, Y.~R. Sanders, D.~W. Berry, G.~K. Brennen, and B.~C. Sanders, ``{Nearly Optimal Quantum Algorithm for Generating the Ground State of a Free Quantum Field Theory},'' \href{http://dx.doi.org/10.1103/PRXQuantum.3.020364}{{\em PRX Quantum} {\bfseries 3} no.~2, (2022) 020364}, \href{http://arxiv.org/abs/2110.05708}{{\ttfamily arXiv:2110.05708 [quant-ph]}}.

\bibitem{qft4}
D.~J. George, Y.~R. Sanders, M.~Bagherimehrab, B.~C. Sanders, and G.~K. Brennen, ``{Entanglement in quantum field theory via wavelet representations},'' \href{http://dx.doi.org/10.1103/PhysRevD.106.036025}{{\em Phys. Rev. D} {\bfseries 106} no.~3, (2022) 036025}, \href{http://arxiv.org/abs/2201.06211}{{\ttfamily arXiv:2201.06211 [quant-ph]}}.

\bibitem{qft5}
{\v{S}}.~Vedl, D.~J. George, and G.~K. Brennen, ``{Scale limited fields and the Casimir effect},'' \href{http://dx.doi.org/10.1103/PhysRevD.109.016018}{{\em Phys. Rev. D} {\bfseries 109} no.~1, (2024) 016018}, \href{http://arxiv.org/abs/2310.04089}{{\ttfamily arXiv:2310.04089 [quant-ph]}}.

\bibitem{fspde}
M.~Bagherimehrab, K.~Nakaji, N.~Wiebe, G.~K. Brennen, B.~C. Sanders, and A.~Aspuru-Guzik, ``{Fast quantum algorithm for differential equations},'' \href{http://arxiv.org/abs/2306.11802}{{\ttfamily arXiv:2306.11802 [quant-ph]}}.

\bibitem{hh1}
A.~W. Harrow, A.~Hassidim, and S.~Lloyd, ``{Quantum Algorithm for Linear Systems of Equations},'' \href{http://dx.doi.org/10.1103/physrevlett.103.150502}{{\em Phys. Rev. Lett.} {\bfseries 103} no.~15, (2009) 150502}, \href{http://arxiv.org/abs/0811.3171}{{\ttfamily arXiv:0811.3171 [quant-ph]}}.

\bibitem{hhl2}
C.~Andrew~M., K.~Robin, and S.~Rolando~D., ``{Quantum Algorithm for Systems of Linear Equations with Exponentially Improved Dependence on Precision},'' \href{http://dx.doi.org/10.1137/16M1087072}{{\em SIAM J. Comput.} {\bfseries 46} no.~6, (2017) 1920--1950}, \href{http://arxiv.org/abs/1511.02306}{{\ttfamily arXiv:1511.02306 [quant-ph]}}.

\bibitem{hhl3}
B.~D. {Clader}, B.~C. {Jacobs}, and C.~R. {Sprouse}, ``{Preconditioned Quantum Linear System Algorithm},'' \href{http://dx.doi.org/10.1103/PhysRevLett.110.250504}{{\em \prl} {\bfseries 110} no.~25, (June, 2013) 250504}, \href{http://arxiv.org/abs/1301.2340}{{\ttfamily arXiv:1301.2340 [quant-ph]}}.

\bibitem{Childshhl}
C.~Andrew~M., K.~Robin, and S.~Rolando~D., ``{Quantum Algorithm for Systems of Linear Equations with Exponentially Improved Dependence on Precision},'' \href{http://dx.doi.org/10.1137/16M1087072}{{\em SIAM J. Comput.} {\bfseries 46} no.~6, (2017) 1920--1950}, \href{http://arxiv.org/abs/1511.02306}{{\ttfamily arXiv:1511.02306 [quant-ph]}}.

\bibitem{Ambainis:2010wfy}
A.~Ambainis, ``{Variable time amplitude amplification and a faster quantum algorithm for solving systems of linear equations},'' \href{http://arxiv.org/abs/1010.4458}{{\ttfamily arXiv:1010.4458 [quant-ph]}}.

\bibitem{alves2024}
D.~S.~M. Alves, ``Entanglement renormalization for quantum field theories with discrete wavelet transforms,'' \href{http://dx.doi.org/10.1007/JHEP07(2024)081}{{\em Journal of High Energy Physics} {\bfseries 2024} no.~7, (2024) 81}. \url{https://doi.org/10.1007/JHEP07(2024)081}.

\bibitem{PhysRevD.109.016018}
{\v{S}}.~Vedl, D.~J. George, and G.~K. Brennen, ``Scale limited fields and the casimir effect,'' \href{http://dx.doi.org/10.1103/PhysRevD.109.016018}{{\em Phys. Rev. D} {\bfseries 109} (Jan, 2024) 016018}. \url{https://link.aps.org/doi/10.1103/PhysRevD.109.016018}.

\bibitem{Childs:2012gwh}
A.~M. Childs and N.~Wiebe, ``{Hamiltonian Simulation Using Linear Combinations of Unitary Operations},'' \href{http://dx.doi.org/10.26421/QIC12.11-12-1}{{\em Quant. Inf. Comput.} {\bfseries 12} no.~11\&12, (2012) 0901--0924}, \href{http://arxiv.org/abs/1202.5822}{{\ttfamily arXiv:1202.5822 [quant-ph]}}.

\bibitem{nestedboxes}
R.~Babbush, D.~W. Berry, J.~R. McClean, and H.~Neven, ``{Quantum simulation of chemistry with sublinear scaling in basis size},'' \href{http://dx.doi.org/10.1038/s41534-019-0199-y}{{\em npj Quantum Inf.} {\bfseries 5} (2019) 92}.

\bibitem{PDOpaper}
H.~Li, H.~Ni, and L.~Ying, ``{On efficient quantum block encoding of pseudo-differential operators},'' \href{http://dx.doi.org/10.22331/q-2023-06-02-1031}{{\em Quantum} {\bfseries 7} (2023) 1031}, \href{http://arxiv.org/abs/2301.08908}{{\ttfamily arXiv:2301.08908 [quant-ph]}}.

\bibitem{addergidney}
C.~Gidney, ``Halving the cost of quantum addition,'' \href{http://dx.doi.org/10.22331/q-2018-06-18-74}{{\em Quantum} {\bfseries 2} (June, 2018) 74}. \url{http://dx.doi.org/10.22331/q-2018-06-18-74}.

\bibitem{be5}
D.~Camps, L.~Lin, R.~Van~Beeumen, and C.~Yang, ``{Explicit Quantum Circuits for Block Encodings of Certain Sparse Matrices},'' \href{http://dx.doi.org/10.1137/22M1484298}{{\em SIAM J. Matrix Anal. Appl.} {\bfseries 45} no.~1, (2024) 801--827}, \href{http://arxiv.org/abs/2203.10236}{{\ttfamily arXiv:2203.10236 [quant-ph]}}.

\bibitem{be6}
C.~S{\"u}nderhauf, E.~Campbell, and J.~Camps, ``{Block-encoding structured matrices for data input in quantum computing},'' \href{http://dx.doi.org/10.22331/q-2024-01-11-1226}{{\em Quantum} {\bfseries 8} (2024) 1226}, \href{http://arxiv.org/abs/2302.10949}{{\ttfamily arXiv:2302.10949 [quant-ph]}}.

\bibitem{domyuv}
Y.~R. {Sanders}, G.~{Hao Low}, A.~{Scherer}, and D.~W. {Berry}, ``{Black-box quantum state preparation without arithmetic},'' \href{http://dx.doi.org/10.48550/arXiv.1807.03206}{{\em arXiv e-prints} (July, 2018) arXiv:1807.03206}, \href{http://arxiv.org/abs/1807.03206}{{\ttfamily arXiv:1807.03206 [quant-ph]}}.

\bibitem{costella2002newproposalfermiondoubling}
J.~P. Costella, ``A new proposal for the fermion doubling problem. ii. improving the operators for finite lattices,'' 2002.
\newblock \url{https://arxiv.org/abs/hep-lat/0207015}.

\bibitem{Su:2021lut}
Y.~Su, D.~W. Berry, N.~Wiebe, N.~Rubin, and R.~Babbush, ``{Fault-Tolerant Quantum Simulations of Chemistry in First Quantization},'' \href{http://dx.doi.org/10.1103/PRXQuantum.2.040332}{{\em PRX Quantum} {\bfseries 2} no.~4, (2021) 040332}, \href{http://arxiv.org/abs/2105.12767}{{\ttfamily arXiv:2105.12767 [quant-ph]}}.

\bibitem{Lemieux:2024pmt}
J.~Lemieux, M.~Lostaglio, S.~Pallister, W.~Pol, K.~Seetharam, S.~Sim, and B.~{\c{S}}ahino{\u{g}}lu, ``{Quantum sampling algorithms for quantum state preparation and matrix block-encoding},'' \href{http://arxiv.org/abs/2405.11436}{{\ttfamily arXiv:2405.11436 [quant-ph]}}.

\bibitem{llyod}
Q.~T. Nguyen, B.~T. Kiani, and S.~Lloyd, ``{Block-encoding dense and full-rank kernels using hierarchical matrices: applications in quantum numerical linear algebra},'' \href{http://dx.doi.org/10.22331/q-2022-12-13-876}{{\em Quantum} {\bfseries 6} (2022) 876}, \href{http://arxiv.org/abs/2201.11329}{{\ttfamily arXiv:2201.11329 [quant-ph]}}.

\bibitem{qsvt}
A.~Gily\'en, Y.~Su, G.~H. Low, and N.~Wiebe, \href{http://dx.doi.org/10.1145/3313276.3316366}{``{Quantum singular value transformation and beyond: exponential improvements for quantum matrix arithmetics},''} in {\em {51st Annual ACM SIGACT Symposium on Theory of Computing}}.
\newblock 6, 2018.
\newblock \href{http://arxiv.org/abs/1806.01838}{{\ttfamily arXiv:1806.01838 [quant-ph]}}.

\bibitem{be2}
J.~Zylberman, U.~Nzongani, A.~Simonetto, and F.~Debbasch, ``{Efficient Quantum Circuits for Non-Unitary and Unitary Diagonal Operators with Space-Time-Accuracy Trade-Offs},'' \href{http://dx.doi.org/10.1145/3718348}{{\em ACM Trans. Quant. Comput.} {\bfseries 6} no.~2, (2025) 15}, \href{http://arxiv.org/abs/2404.02819}{{\ttfamily arXiv:2404.02819 [quant-ph]}}.

\bibitem{compgid}
T.~Khattar and C.~Gidney, ``{Rise of conditionally clean ancillae for efficient quantum circuit constructions},'' \href{http://dx.doi.org/10.22331/q-2025-05-21-1752}{{\em Quantum} {\bfseries 9} (2025) 1752}, \href{http://arxiv.org/abs/2407.17966}{{\ttfamily arXiv:2407.17966 [quant-ph]}}.

\bibitem{HHL}
A.~W. Harrow, A.~Hassidim, and S.~Lloyd, ``{Quantum Algorithm for Linear Systems of Equations},'' \href{http://dx.doi.org/10.1103/physrevlett.103.150502}{{\em Phys. Rev. Lett.} {\bfseries 103} no.~15, (2009) 150502}, \href{http://arxiv.org/abs/0811.3171}{{\ttfamily arXiv:0811.3171 [quant-ph]}}.

\bibitem{EvansPDE}
L.~C. Evans, {\em Partial differential equations}, vol.~19 of {\em Graduate Studies in Mathematics}.
\newblock American Mathematical Society, Providence, RI, 1998.

\bibitem{cmp1}
Z.~{Wang}, F.~{Assaad}, and M.~{Ulybyshev}, ``{Validity of SLAC fermions for the (1 +1 ) -dimensional helical Luttinger liquid},'' \href{http://dx.doi.org/10.1103/PhysRevB.108.045105}{{\em \prb} {\bfseries 108} no.~4, (July, 2023) 045105}, \href{http://arxiv.org/abs/2211.02960}{{\ttfamily arXiv:2211.02960 [cond-mat.str-el]}}.

\bibitem{cmp2}
C.~W.~J. {Beenakker}, A.~{Don{\'\i}s Vela}, G.~{Lemut}, M.~J. {Pacholski}, and J.~{Tworzyd{\l}o}, ``{Tangent Fermions: Dirac or Majorana Fermions on a Lattice Without Fermion Doubling},'' \href{http://dx.doi.org/10.1002/andp.202300081}{{\em Annalen der Physik} {\bfseries 535} no.~7, (July, 2023) 2300081}, \href{http://arxiv.org/abs/2302.12793}{{\ttfamily arXiv:2302.12793 [cond-mat.mes-hall]}}.

\bibitem{cmp3}
T.~C. Lang and A.~M. L{\"a}uchli, ``{Chiral Heisenberg Gross-Neveu-Yukawa criticality: honeycomb vs. SLAC fermions},'' \href{http://arxiv.org/abs/2503.15000}{{\ttfamily arXiv:2503.15000 [cond-mat.str-el]}}.

\bibitem{Wang:2024rer}
S.~Wang, X.~Li, W.~J.~B. Lee, S.~Deb, E.~Lim, and A.~Chattopadhyay, ``{A Comprehensive Study of Quantum Arithmetic Circuits},'' \href{http://arxiv.org/abs/2406.03867}{{\ttfamily arXiv:2406.03867 [quant-ph]}}.

\end{thebibliography}\endgroup
\newpage
\appendix
\section{The ingredients for the LCU based block-encoding}
In Sections~\ref{sec: slaclap} and \ref{sec: slac1be}, we adapt the nested-box inequality-test framework to prepare the required \(N\) amplitudes with the desired weighting on each basis state. Using the finally prepared state, this appendix briefly describes the components of the $\mathrm{LCU}$ required to construct the block-encoding, as defined analytically in Eq.~\eqref{beslac2} and Eq.~\eqref{beslac1}. These “ingredients” comprise, as usual, the preparation circuit (for the respective block-encoding) and the selector circuit. In addition, we introduce a copy oracle that labels the success and failure branches, ensuring that the desired branch is recovered after uncomputation. We also briefly describe the implementation of the sign oracle used to set the required phase factors in both \(\mathrm{LCU}\) constructions.

\subsection{The preparation circuit: PREP}\label{app: O_prep}
For both block-encoding—the first-order SLAC derivative and the Laplacian—we begin by preparing the desired scaling through an inequality test. 

From Eq.~\eqref{fis}, the result of implementing the state-preparation module underlying the $\mathrm{LCU}$-based block-encoding (for $n\gg0$) of the SLAC Laplacian can be formalised as follows:

\begin{result}\label{lem:prep}
The overall action of the state-preparation circuit $O_{\mathrm{prep}}^{(2)}$ used in the $\mathrm{LCU}$-based block-encoding of the SLAC Laplacian (for $n\gg0$) is:
\begin{widetext}
\begin{align}\label{eq:prep2}
\ket{[0]_{n-1}}\,&\ket{[0]_{n-1}}\ket{\bar0}_{d;a}\,\ket{[0]_{n_{ref}}}\ket{0}_{f}\,\ket{\psi}_{\mathrm{sys}}
\;\xrightarrow{O_{\mathrm{prep}}^{(2)}}\notag\\&\sqrt{\frac{\pi^2}{\pi^2+24}}\ket{\mu=0}\ket{j = 0}\ket{0}_d\ket{0}_a+\sqrt\frac{12}{\pi^2+24}
      \sum_{\mu=0}^{n-2}\ket{\mu}
      \sum_{j \in B_{\mu}}\frac{1}{j}\ket{j}
 \ket{+}_d\ket{1}_a\ket{0}_{f_{1}} + \sqrt{1-p^{(2)}}\ket\omega_{fail}\,.
\end{align}
\end{widetext}
where $\ket{\omega_{\mathrm{fail}}}$ is the failure state resulting from the inequality test. A high-level circuit diagram for this operation is shown in Fig.~\ref{fig:Opreplap}. This circuit uses $\mathcal{O}(n)$ ancilla qubits and has gate complexity $\mathcal{O}(n^2)$. 
\end{result}

As noted earlier in Section~\ref{subsec:prepfo}, the state-preparation routine for the first-order derivative follows the same sequence of steps as in the Laplacian case, except that it excludes the state $j=0$. Using Eq.~\eqref{fis1}, we obtain the following result describing the action of the corresponding preparation circuit for the first-order SLAC derivative.

\begin{result}\label{cor:prep1}
The overall action of the state-preparation circuit $O_{\mathrm{prep}}^{(1)}$ for the $\mathrm{LCU}$-based block-encoding of the first-order SLAC derivative operator is: 
\begin{widetext}
\begin{equation}\label{PREP1}
\begin{aligned}
&\ket{[0]_{n-1}}\,\ket{[0]_{n-1}}\;\ket{0}_d\,\ket{[0]_{n_{ref}}}\ket{0}_{f}\,\ket{\psi}_{\mathrm{sys}}\;\xrightarrow{O_{\mathrm{prep}}^{(1)}}\\&\sqrt{\frac{1}{(n-1)}}\sum_{\mu=0}^{n-2}\ket{\mu}\sum_{j \in B_{\mu}}\sqrt{\frac 1 j}\ket{[j]_{n-1}}
\ket{+}_d\ket{0}_{f}+\sqrt{1-p^{(1)}}\ket\omega_{fail}\ket{1}_{f}\,.
\end{aligned}
\end{equation}
\end{widetext}
where $\ket{\omega_{\mathrm{fail}}}$ denotes the failure state produced by the inequality test.  
The high-level circuit structure is identical to that of Fig.~\ref{fig:Opreplap}, except for the modified first step and the altered inequality test.  
This circuit requires $\mathcal{O}(n)$ ancilla qubits and has gate complexity $\mathcal{O}(n^2)$.
\end{result}

\subsection{The Copy oracle}\label{app:cpy}
As the preparation circuit constructed above had some overall contribution from irrelavant states $\ket\omega_{fail}$, we want to mark the relevant and undesirable branch with a marker qubit labelled $C$. This ensures that while uncomputing the flipped qubit $C$, the undesired branches get discarded.

To do this one needs to apply $\mathrm{CX}$ gate from flag $f$ to marker qubit $C$. Then the failure branch will have a flipped marker qubit i.e. $\ket1_C$. This oracle will not be uncomputed and hence will help us to project out and postselect on the relevant branch. 

After implementing this copy oracle, the prepared state for the SLAC Laplacian operator (from Eq.\eqref{eq:prep2}) now becomes:
\begin{widetext}
\begin{equation}\label{c1}
\sqrt{\frac{\pi^2}{\pi^2+24}}\ket{\mu=0}\ket{j = 0}\ket{0}_d\ket{0}_a\ket{0}_C+\sqrt\frac{12}{\pi^2+24}
      \sum_{\mu=0}^{n-2}\ket{\mu}
      \sum_{j \in B_{\mu}}\frac{1}{j}\ket{j}
 \ket{+}_d\ket{1}_a\ket{0}_{f_{1}}\ket{0}_C + \sqrt{1-p^{(2)}}\ket\omega_{fail}\ket{1}_C\,.
\end{equation}
\end{widetext}
Similarly if we apply the copy oracle as $\mathrm{CX}$ gate from flag $f$ to the marker qubit $C$, then the prepared state for the first order SLAC derivative (from Eq.\eqref{PREP1})  becomes: 
\begin{widetext}
\begin{equation}\label{cpy1}
\begin{aligned}
\sqrt{\frac{1}{(n-1)}}\sum_{\mu=0}^{n-2}\ket{\mu}\sum_{j \in B_{\mu}}\sqrt{\frac 1 j}\ket{[j]_{n-1}}
\ket{+}_d\ket{0}_{f}\ket{0}_C+\sqrt{1-p^{(1)}}\ket\omega_{fail}\ket{1}_C\,.
\end{aligned}
\end{equation}
\end{widetext}

\subsection{The Sign Oracle for SLAC Laplacian}\label{app: O_sgn}
It is clear from Eq.~\eqref{slaccf} that we must also apply alternating signs to the computational basis states $\ket{j}$ to obtain the correct weightings in the SLAC Laplacian. This can be done efficiently using a sign oracle, say $O_{sgn}$, which has the following action:
\begin{equation}
    O_{sgn} \ket{j}\ket{\psi}_{sys} = (-1)^{1+j} \ket{j} \ket{\psi}_{sys}\,.
\end{equation}
This operation can be realised by applying a $Z$ gate to the $LSB$ of the ancilla register ${j}$. This produces the state $(-1)^j\,\ket{j}$, which equals the desired phase factor up to a global factor of $-1$. Therefore, if we apply this sign oracle to the prepared state of Eq.~\eqref{c1}, we obtain:
\begin{widetext}
    \begin{equation}
\label{sgn}
\Biggl\{-\sqrt{\frac{\pi^2}{\pi^2+24}}\ket{\mu=0}\ket{j = 0}\ket 0_d\ket{0}_a+\sqrt{\frac{12}{\pi^2+24}}
      \sum_{\mu=0}^{n-2}\ket{\mu}
      \sum_{j \in B_{\mu}}\frac{(-1)^{1+j}}{j}\ket{j}\ket{+}_d\ket{1}_a\Biggl\}
 \ket{0}_{f_{1}}\ket{0}_C + \sqrt{1-p^{(2)}}\ket\omega_{fail}\ket{1}_C\,.
\end{equation}
\end{widetext}

\subsection{The Phase Oracle for first order SLAC derivative}\label{app: phase}
The first-order SLAC derivative operator also has $j$-dependent phases, as evident in Eq.~\eqref{slac1cf}. These phases can be realised, up to a global factor of $-1$, by applying a Phase oracle $O_p$ which acts on a computational basis state $\ket{j}$ as follows:
\begin{equation}\label{signfo}
O_p^{(1)}\ket{j} =
\begin{cases}
   -\,e^{i\alpha j}\,\ket{j}, & j\neq 0,
\end{cases}
\qquad \alpha = \pi\!\left(1+\tfrac{1}{N}\right).
\end{equation}
Hence, up to an overall global phase of $-1$, the required phase can be applied to each computational basis state $\ket{j}$.

To implement this efficiently, we write the integer $j$ in binary, $j = \sum_{k=0}^{m-1} b_k\,2^k ,\,\, b_k \in \{0,1\},$ so that the basis state $\ket{j}$ corresponds to the bitstring $(b_{m-1}\cdots b_0)$.
With this representation, the desired phase factor can be realised by applying appropriate single-qubit phase rotations to each qubit, conditioned only on the value of the corresponding bit $b_k$ as,
\begin{equation*}
    \;e^{i\alpha j} = \prod_{k=0}^{n-1} e^{i\alpha b_k2^k}\,.
\end{equation*}
That is, to each qubit $k$ we apply a diagonal phase unitary gate $P_k(\phi_k) = \mathrm{diag}(1,e^{i\phi_k})$ in parallel (with $\phi_k=\alpha 2^k$), with overall circuit depth 1. Thus, if the phase oracle is applied to the state obtained after implementing the $\mathrm{COPY}$ step of the $\mathrm{LCU}$, i.e., to the state in Eq.~\eqref{cpy1}, then we obtain:
\begin{widetext}
\begin{equation}\label{phasefo}
\begin{aligned}
&\sqrt{\frac{1}{3\,({n}-1)}}\Biggr\{\sum_{\mu=1}^{n-1}\ket{\mu}
   \sum_{j \in B_{\mu}}\exp\Big(\frac{i\pi j}{N}\Big)(-1)^{j+1}\sum_{j \in B_{\mu}}\sqrt{\frac1j}\ket{[j]_{n-1}}
\;\;\Big(\ket{0}_d+\ket{1}_d\Big)\Biggl\}\ket{0}_{f}\ket{0}_C \;+\; \ket{\omega}_{fail}\ket{1}_C\Biggl\}.
\end{aligned}
\end{equation}
\end{widetext}

\subsection{The SELECT oracle of the $\mathrm{LCU}$}\label{app: sel}
The SLAC Laplacian and the first-order SLAC derivative for a periodic lattice are translationally invariant circulant dense matrices. To construct the block-encoding using the state-preparation and sign oracles defined in the previous subsections, it is necessary to apply appropriate cyclic shift matrices to the states prepared in Eq.~\eqref{sgn} and Eq.~\eqref{phasefo}. 

Let $P^k$ denote the cyclic right-shift matrix of size $N \times N$, with a leading $1$ in column $k$ of the first row, which then shifts to the right in consecutive rows, such that each row is a cyclic right shift of the row above it and the first row has the $1$ in column $k$. We now summarise the action of the $\mathrm{SELECT}$ oracle and its realisation via modular subtraction.
\begin{result}\label{lem:sel}
The $O_{\mathrm{select}}$ oracle acts by applying the corresponding permutation matrix to the computational basis state:
\begin{equation}
\ket{j}\ket{\psi}\xrightarrow{O_\mathrm{select}} \ket{j}\hspace{3pt} P^j\;\ket{\psi},
\end{equation}
where $P^{j}$ denotes the size-$N$ right-shift permutation matrix.  This action can be realised using a modular subtractor, which implements the shift:
\begin{equation}
\ket{j}\ket{\psi}_{sys}\rightarrow\ket{j}\ket{(\psi-j)mod\;N}.
\end{equation}
The shift operation is implemented through controlled subtractor circuits, which in turn can be constructed from controlled adder circuits as suggested in~\cite{Wang:2024rer}.  Each adder can be implemented using the low~$T$-cost design of~\cite{addergidney}, while the controlled (on $n$ qubits) version of this adder incurs a gate complexity of $\mathcal{O}(n^2)$.
\end{result}
The select oracle for the $\mathrm{LCU}$ of both SLAC operators is constructed to apply $P^{j}$ to the first half of the lattice and $P^{N-j}$ to the second half. Moreover, the discardable part of the state is left untouched via an identity operation on the system. Such a SELECT circuit ensures that a translationally invariant circulant matrix is implemented on the system $\psi$. 

Specifically, the SELECT oracle for the SLAC Laplacian has controls on the register $j$, and the $\mathrm{MSB}$ qubit $d=0/1$ selects whether $P^{j}$ or $P^{N-j}$ is applied; this oracle for the Laplacian can then be expressed as:
\begin{equation}
\begin{aligned}
\Biggl\{\, &|j=0\rangle\langle j=0| \otimes |0\rangle_d\langle 0| \otimes|0\rangle_a\langle 0| \otimes P^0_{\mathrm{sys}} + \sum_{j\in B_\mu} |j\rangle\langle j| \\
&\otimes 
|1\rangle_a\langle 1|\Bigl(
   |0\rangle_d\langle 0| \otimes P^j 
 + |1\rangle_d\langle 1| \otimes P^{N-j}
\Bigr)
\Biggr\} \otimes |0\rangle_C\langle 0|.
\end{aligned}
\end{equation}

So, if we apply the SELECT oracle to the state generated in Eq.~\eqref{sgn} for the SLAC Laplacian then, we obtain the following state which implements the permutation matrix for each $j$: 
\begin{widetext}
    \begin{align}
\label{sel2}
\Biggl\{-\sqrt{\frac{\pi^2}{\pi^2+24}}&\ket{\mu=0}\ket{j = 0}\ket{0}_d\ket{0}_a\otimes P^0_{sys}+\sqrt\frac{6}{\pi^2+24}
      \sum_{\mu=0}^{n-2}\ket{\mu}\notag
      \sum_{j \in B_{\mu}}\frac{(-1)^{1+j}}{j}\ket{j}\bigg(\ket{0}_d\otimes P^{j}_{sys}+\ket{1}_d\otimes P^{N-j}_{sys}\bigg)\ket{1}_a\Biggr\}\\&\otimes\ket{0}_{f}\ket{0}_C\ket{\psi}_{sys}+ \sqrt{1-p^{(2)}}\ket\omega_{fail}\ket{1}_C\,.
\end{align}
\end{widetext}
For the SELECT oracle applied to the prepared state in the $\mathrm{LCU}$ of the first-order SLAC derivative operator, we only need to apply the right shift cyclic permutation matrices via adder circuits for the branch $\mu;j\ne0$, using the following oracle:
\begin{align}
&\sum_{j\in B_\mu} \ket{j}\bra{j} \otimes 
     \Big( \ket{0}_{d}\bra{0} \otimes P^{j}_{sys}\nonumber\\& + \ket{1}_{d}\bra{1} \otimes P^{N-j}_{sys}\Big)\Biggr\}\otimes\ket{0}_C\bra{0}\,.
\end{align}
After applying this step using modular subtractor circuits, the state prepared for the $\mathrm{LCU}$ of the first order SLAC derivative operator, given in Eq.~\eqref{signfo} becomes: 
\begin{widetext}
\begin{equation}\label{selfo}
\begin{aligned}
\sqrt{\frac{1}{2\,(n-1)}}\Biggr\{\sum_{\mu=0}^{n-2}\ket{\mu}
   \sum_{j \in B_{\mu}} \exp\Big(\frac{i\pi j}{N}\Big)(-1)^{1+j}\sqrt{\frac1j}\ket{[j]_{n-1}}\Bigg(\ket{0}_d\otimes P^{j}_{sys}+\ket{1}_d\otimes P^{N-j}_{sys}\Bigg)\Biggl\}\ket{0}_{f}\ket{0}_C+\ket\omega_{fail}\ket 1_C\,.
\end{aligned}
\end{equation}
\end{widetext}
The implementation of the SELECT operation is depicted as a black box in Fig.~\ref{fig:lcu-architecture}, consisting of adder circuits with low $T$-cost, as described in~\cite{addergidney}.  
The final step for both LCUs—the first-order SLAC derivative and the Laplacian—is to uncompute the operations in $O_\mathrm{prep}^{(1;2)}$. However, the sign/phase oracle and the copy oracle are not uncomputed in either case. 

As stated in Eq.~\eqref{slac2be} and Eq.~\eqref{slac1be}, the postselected state has all flag qubits in the $\ket{0}$ state, with the desired amplitudes. The rescaling factor of the block-encoding is then proportional to the success probability of implementing the inequality test on the desired branch.

\section{Linear Combination of block-encodings of SLAC derivatives}\label{app:sumlcu}
Suppose we want to implement an $s$ qubit linear combination of SLAC derivative operators of form:
\begin{equation}\label{lcslac}
    A = \sum_{j=0}^{m-1} y_j A_j^{(k)},
\end{equation}
where $A_j^{(k)}$ is a $k^{th}$ order SLAC derivative operator, and $y\in \mathbf{C}^m \;s.t.\;||y||_1\leq \beta$, then we can revisit the Lemma in~\cite{qsvt} by first defining a state preparation pair\\
\begin{df1}\label{df2}
(State preparation pair; see~\cite{qsvt}): If there are pairs of unitaries implementing $P_L\ket{0}^{\otimes b}=\sum_{j=0}^{2^b-1}c_j\ket{j}$ and $P_R\ket{0}^{\otimes b}=\sum_{j=0}^{2^b-1}d_j\ket{j}$ such that $\sum_{j=0}^{m-1} |\beta c_j^* d_j - y_j| \leq \varepsilon_1 \quad \text{and for all } j \in \{m, \dots, 2^b - 1\}, \quad c_j^* d_j = 0.$ then $P_L,P_R$ prepare the amplitudes of Eq.~\eqref{lcslac} rescaled by a factor $\beta$ upto error $\varepsilon_1$ using $b$ ancilla qubits. Hence, $P_L,P_R$ is referred to as a $(\beta,b,\varepsilon_1)$- state preparation pair.
\end{df1}
One would also need to redesign the $\mathrm{SELECT}$ operation of the $\mathrm{LCU}$ in order to implement the linear combination of slac derivative operators, as defined below,
\begin{df1}\label{df3}
(SELECT operator for implementing B.E. of SLAC derivative; see~\cite{qsvt}): We define an $s+a+b$- qubit unitary operator which implements, the $(\alpha, a , \varepsilon_2)$ B.E., $U_j$ of the $k^{th}$ order SLAC derivative operator, $A_j^{(k)}$  (as constructed in Eq.~\eqref{slac1be} and Eq.~\eqref{slac2be}) if $j\leq m$ as:
\begin{equation}
    \mathcal{S} = \sum_{j=0}^{m-1}\ket{j}\bra{j} \otimes U_j + \left( I - \sum_{j=0}^{m-1} \ket{j}\bra{j}\right) \otimes I_a \otimes I_s.
\end{equation}
\end{df1}
Now, we can construct a $\mathrm{LCU}$ using the state preparation pair and the $\mathrm{SELECT}$ operator as follows:
\begin{equation}
\begin{aligned}
&P_L^\dagger \mathcal{S} P_R \, \ket{0}^{\otimes b} \ket{0}^{\otimes a} \ket{\psi}_{\text{sys}} \\
&\quad= P_L^\dagger \mathcal{S} \left( \sum_{j=0}^{2^b - 1} d_j \ket{j} \ket{0}^{\otimes a} \ket{\psi}_{\text{sys}} \right) \\
&\quad= P_L^\dagger \left( 
    \sum_{j=0}^{m-1} d_j \ket{j} U_j \ket{0}^{\otimes a} 
    + \sum_{j=m}^{2^b - 1} d_j \ket{j} \ket{0}^{\otimes a}\right) \ket{\psi}_{\text{sys}} 
 \\
&\quad= P_L^\dagger \left( 
    \sum_{j=0}^{m-1} d_j \ket{j} U_j 
    + \sum_{j=m}^{2^b - 1} d_j \ket{j} 
\right) \ket{0}^{\otimes a} \ket{\psi}_{\text{sys}} \\
&\quad= \sum_{j=0}^{m-1} c_j^* d_j U_j \ket{0}^{\otimes a} \ket{\psi}_{\text{sys}}\,.
\end{aligned}
\end{equation}
Since the state-preparation pair encodes the coefficients $y_j/\beta$ up to error $\varepsilon_1$, and each unitary $U_j$ provides an $(\alpha,a,\varepsilon_2)$ block-encoding of the corresponding SLAC derivative operator $A_j^{(k)}$, the construction described in this appendix yields the following result:
\begin{result}\label{re:lcuolcu}
We can implement an $(\alpha\beta,\, a + b,\, \alpha\varepsilon_1 + \alpha\beta \varepsilon_2)$-block-encoding of $A$ defined in Eq.~\eqref{lcslac} using a single use of $\mathcal{S}$, $P_R$, and $P_L^\dagger$, where each $U_j$ is implemented by calling the efficient block-encodings $U_{\mathrm{slac}}^{(1)}$ and $U_{\mathrm{slac}}^{(2)}$ defined in.~Eq.~\eqref{slac1be} and Eq.~\eqref{slac2be}.
\end{result}

For both the SLAC Laplacian and the first-order SLAC derivative, the block-encodings constructed in Sections~\ref{sec: slaclap} and~\ref{sec: slac1be} achieve precision scaling as $\mathcal{O}(1/N)$. Consequently, the linear combination of these operators in~Eq.~\eqref{lcslac} can be implemented via the $\mathrm{LCU}$ construction with an overall approximation error of order $\varepsilon_2 = \mathcal{O}(1/N)$. Since we are using a linear combination of two different matrices, we take the common rescaling factor to be $\alpha = \max(\alpha_j^k)$, which in this case corresponds to the rescaling factor of the first-order SLAC derivative. Thus, in order to implement this $\mathrm{LCU}$ of block-encoded SLAC matrices, we prepare the coefficients $y_j' = \alpha_j^k y_j / \alpha$ using Result~\ref{re:lcuolcu}. In particular, for $N = 2^n$, combining these block-encodings via Results~\ref{corl1}, \ref{corl2} and~\ref{re:lcuolcu} yields an implementation of the linear combination with gate complexity $\mathcal{O}(n^2)$ and total error $\mathcal{O}(1/N)$.

\section{Ellipticity of the benchmark operators}\label{app:ellip}

In Section~\ref{sec:precon}, we studied the condition numbers of the SLAC-discretised elliptic partial differential operators given in Eq.~\eqref{eq:benchmark_ops}. 
In this appendix, we provide a short justification for why these operators are elliptic. 
We follow the formulation presented in Chapter~6 of Ref.~\cite{EvansPDE}.

Consider the differential equation $\mathcal{L}u=f$, where $U$ is an open, bounded subset of $\mathbb{R}^n$, $u:\bar{U}\to\mathbb{R}$ is the unknown function, and $f:U\to\mathbb{R}$ is given. 
Here $\mathcal{L}$ denotes a second-order partial differential operator, which may be written in divergence form as
\[
\mathcal{L}u
=
-\sum_{i,j=1}^{n}
\left(a^{ij}(x)u_{x_i}\right)_{x_j}
+
\sum_{i=1}^{n} b^i(x)u_{x_i}
+
c(x)u,
\]
or in non-divergence form as
\[
\mathcal{L}u
=
-\sum_{i,j=1}^{n}
a^{ij}(x)u_{x_i x_j}
+
\sum_{i=1}^{n} b^i(x)u_{x_i}
+
c(x)u,
\]
for given coefficient functions $a^{ij}$, $b^i$, and $c$, with $i,j=1,\ldots,n$. 
The ellipticity of such an operator is determined by the highest-order coefficient matrix
\[
A(x)=\bigl(a^{ij}(x)\bigr).
\]
The operator is uniformly elliptic if there exists a constant $\theta>0$ such that
\[
\sum_{i,j=1}^{n} a^{ij}(x)\xi_i\xi_j
\ge
\theta |\xi|^2
\]
for almost every $x\in U$ and all $\xi\in\mathbb{R}^n$. 
Equivalently, the matrix $A(x)$ must be uniformly positive definite.

The benchmark operators considered in this work are one-dimensional. 
Therefore, the coefficient matrix $A(x)$ reduces to a scalar coefficient $a(x)$, and the ellipticity condition becomes simply
\[
a(x)\ge \theta>0
\]
for all $x$ in the domain.

For the operator $-\mathcal{L}_1$ in Eq.~\eqref{eq:benchmark_ops}, the scalar coefficient of the highest-order term is
\[
a_1(x)=1,
\]
which is uniformly positive. 
Thus, $\mathcal{L}_1$ is elliptic up to an overall sign convention.

For the operator $\mathcal{L}_2$, the principal coefficient is
\[
a_2(x)=c_a(x)=\cosh(x/4).
\]
Since $\cosh(x/4)\ge 1$, this coefficient is uniformly positive, and hence $\mathcal{L}_2$ is uniformly elliptic.

Similarly, the operator $\mathcal{L}_3$ has principal coefficient
\[
a_3(x)=1,
\]
and is therefore uniformly elliptic.

Finally, for the operator $\mathcal{L}_4$, the principal coefficient is
\[
a_4(x)=a_\epsilon(x)=1+\epsilon\cos(2\pi x).
\]
For $|\epsilon|<1$, we have
\[
a_4(x)\ge 1-|\epsilon|>0.
\]
Hence $\mathcal{L}_4$ is also uniformly elliptic.

\section{Projecting $\ket{b}$ outside the nullspace of preconditioned SLAC derivative} \label{app:nullspaceIR}
Here we discuss the projection of the state $\ket{b}$ outside the nullspace of the derivative operator which as discussed in Section~\ref{sec:precon} is $\ket{\psi_{null}}$, i.e. we want to implement the following operation, $\ket{b}_s\rightarrow (1-\ket{\psi_{null}}_s\bra{\psi_{null}})\;\ket{b}_s$. To do this we use $n+1$ ancilla qubits in registers $R,\ F$ initialised to the state $\ket{[0]_n}_{R}\;\ket{0}_F$ and first prepare $\ket{\psi_{null}}_{R}$ then perform a control-swap test on the system and reference register and in the end postselect on a particular state of the flag and reference qubits.

We start by preparing the null state in the reference register by applying Hadamard on the first qubit,
    \begin{equation*}
        \ket{b}_s\ket{[0]_n}_{R}\;\ket{0}_F \xrightarrow{H^0_{R}} \ket{b}_s\ket{\psi_{null}}_{R}\;\ket{0}_F
     \end{equation*}

Next, we implement the controlled swap test defined by the operation:
\begin{equation}
\begin{aligned}
V \;&=\;\bigl( H_{F}\!\otimes I_{SR} \bigr).\;\Bigl(\lvert 0\rangle\!\langle 0\rvert_{F}\!\otimes I_{SR}\;+\;\lvert 1\rangle\!\langle 1\rvert_{F}\!\otimes \operatorname{SWAP}_{SR}\Bigr).\;\\&\quad\quad\bigl( H_{F}\!\otimes I_{SR} \bigr)
\end{aligned}
\end{equation}
So, the ancilla and system registers now become:
\begin{align}\label{applyv}
\ket{b}_s\ket{\psi_{\text{null}}}_{R}\;\ket{0}_F
\xrightarrow{V} &\frac{1}{2} \Bigl(
   \ket{0}_{F}
   \bigl(
\ket{b}_{S}\,\ket{\psi_{\text{null}}}_{R}
      +\ket{\psi_{\text{null}}}_{S}\,\ket{b}_{R}
   \bigr) \notag \\
   &+
   \ket{1}_{F}
   \bigl(
       \ket{b}_{S}\,\ket{\psi_{\text{null}}}_{R}
      -\ket{\psi_{\text{null}}}_{S}\,\ket{b}_{R}
   \bigr)
\Bigr).
\end{align}

One can clearly see that the state in the above equation on $\ket{1}_F$ is actually the vector $b$, projected in the space orthogonal to the null vector. 
Now, if we post-select on flag being $\ket{1}_F$, then the state on the system, reference registers becomes:
\begin{equation*}
  \frac{1}{2}\Bigl(
        \ket{b}_{S}\,\ket{\psi_{\text{null}}}_{R}
      - \ket{\psi_{\text{null}}}_{S}\,\ket{b}_{R}
  \Bigr).
\end{equation*}
Since, $\ket{\psi_{\text{null}}}_{S}$ is an eigenstate of the SLAC derivative operator, we can re-express the state, $\ket{b}
  = \alpha\,\ket{\psi_{\text{null}}}
  + \beta\,\ket{\psi_{\perp}},
  \qquad
  \braket{\psi_{\text{null}} | \psi_{\perp}} = 0.$
Then the state after postselecting on flag in $1$, can be reexpressed as:
\begin{equation*}
  \frac{\beta}{2}
  \left(
    \ket{\psi_{\perp}}_{S}\,\ket{\psi_{\text{null}}}_{R}
    - \ket{\psi_{\text{null}}}_{S}\,\ket{\psi_{\perp}}_{R}
  \right)\,.
\end{equation*}
Now, if we trace out the register $R$, then we have the state:
\begin{equation*}
    \frac{\beta}{2}\ket{\psi_{\perp}} = 
  \frac{1}{2}
  \left(
    I - \ket{\psi_{\text{null}}}_s\bra{\psi_{\text{null}}}
  \right)
  \ket{b}_S.
\end{equation*}

Hence, by using $n+1$ ancillas, we can efficiently project the input state $\ket{b}_s$ out of the nullspace of the SLAC derivative operator, by using controlled swap test and postselecting on flag in 1 and tracing out the reference register. 

\section{SLAC derivative as a psuedo differential operator}\label{app:PDO}
A general pseudodifferential operator (PDO) $A$ acting on a periodic function $f$ defined on a discrete lattice $\Omega = [0,1]^d$ with resolution $P = 2^p$ is given by
\begin{equation}\label{PDO}
A f(\mathbf{x}) = \sum_{\boldsymbol{\xi} \in \{-P/2,\,\dots,\,P/2 - 1\}^d} e^{2\pi i \mathbf{x}\cdot\boldsymbol{\xi}/P}\,a\left(\frac{\mathbf{x}}{P},\boldsymbol{\xi}\right)\hat{f}(\boldsymbol{\xi}),
\end{equation}
where $\hat{f}(\boldsymbol{\xi})$ denotes the discrete Fourier transform of $f(\mathbf{x})$, and the function $a(\mathbf{x},\boldsymbol{\xi})$ is referred to as the symbol of the operator $A$. The symbol $a(\mathbf{x},\boldsymbol{\xi})$ characterizes the action of $A$ in momentum space and can generally depend both on the spatial coordinate $\mathbf{x}$ and the momentum variable $\boldsymbol{\xi}$.

In the special case where the symbol is independent of the spatial coordinates, i.e., $a(\mathbf{x},\boldsymbol{\xi})=a(\boldsymbol{\xi})$, the operator $A$ becomes a translation-invariant pseudodifferential operator, also known as a Fourier multiplier. 

\noindent{Connection to SLAC Derivatives:}  
The first-order and second-order SLAC derivative operators considered in this work (Eq.~\eqref{slaccf},Eq.~\eqref{slac1cf} are explicit examples of translation-invariant PDOs. 

Specifically, the first-order SLAC derivative has the symbol $a(\xi)=i\xi$, while the SLAC Laplacian has the symbol $a(\xi)=-\xi^2$. Both operators thus fit naturally within the general PDO framework, making the theoretical results regarding block-encodings of generic PDOs directly applicable to these SLAC derivative operators. In this section we give a theoretical comparison of the block-encodings constructed using PDOs in~\cite{PDOpaper} and the explicit optimal block-encodings constructed in this work using inequality tests. For more details regarding the block-encoding of such PDOs, we refer the reader to~\cite{PDOpaper}. 

Theorem 5 in~\cite{PDOpaper}, states the complexity for block-encoding $\left(2^{pd/2} C_a,\, \mathcal{O} \left( \mathrm{poly}(pd) + \mathrm{polylog} \left( \frac{1}{\varepsilon} \right) \right),\, \varepsilon \right)$ a PDO  with any generic symbol $a(\mathbf{x}, \boldsymbol{\xi})$ as $\mathcal{O} \left( \mathrm{poly}(pd) + \mathrm{polylog} \left( \frac{1}{\varepsilon} \right) \right),$ where $C_a \geq \sup |a|$ is a constant independent of $p$, $d$, and $\varepsilon$. The block-encoding of the first-order SLAC derivative $\widetilde{\Delta}^{(1)}_{\mathrm{slac}}$ defined in Proposition~\ref{corl2} and the SLAC Laplacian operator $\widetilde{\Delta}^{(2)}_{\mathrm{slac}}$ defined in Proposition~\ref{corl1} serve as concrete instantiations of Theorem 5 in~\cite{PDOpaper} for the specific symbol $a(\xi) = i\xi$ and $a(\xi) = -\xi^2$ respectively. 


In particular, the construction shown in Eq.~\eqref{slac1be} achieves a $(p,\mathcal{O}(n), \varepsilon)$-block-encoding of $\widetilde{\Delta}^{(1)}_{\mathrm{slac}}$ where $p=\mathcal O(n)$. While the construction shown in Eq.~\eqref{slac2be} achieves an explicit $(q,\mathcal{O}(n), \varepsilon)$-block-encoding of $\widetilde{\Delta}^{(2)}$ where $q=\mathcal{O}(1)$, constructed using $\mathrm{LCU}$ over circulant permutations with state preparation via inequality tests. 

Compared to the generic normalization factor $\alpha = 2^{pd/2} C_a$ in Theorem 5 of~\cite{PDOpaper}, which simplifies to $\mathcal{O}(\sqrt{N})$ in the one-dimensional case, the explicit construction of our work achieves a favorable normalization of $p$ and $q$ and matches the polylogarithmic gate complexity, while avoiding the exponential amplitude overhead inherent to the generic case. 

As mentioned in~\cite{PDOpaper}, one can attempt to improve the exponential amplitude overhead if the symbol is separable by revisiting Theorem 6 of~\cite{PDOpaper}. Theorem 6 states that
if $a(\mathbf{x}, \boldsymbol{\xi}) = \alpha(\mathbf{x})\beta(\boldsymbol{\xi})$ is a separable symbol then the discretized PDO with separable symbols can be block encoded as $(C_\alpha C_\beta,\, \mathcal{O}(\mathrm{poly}(pd) + \mathrm{polylog}(1/\varepsilon)),\, \varepsilon)$ with gate complexity $\mathcal{O}(\mathrm{poly}(pd) + \mathrm{polylog}(1/\varepsilon))$ where $C_\alpha, C_\beta > 0$ are constants such that $C_\alpha \geq \sup|\alpha|$ and $C_\beta \geq \sup|\beta|$.

\paragraph{Comparison with Theorem 6.}
The SLAC derivative operators $\widetilde{\Delta}^{(1)}$, $\widetilde{\Delta}^{(2)}$ defined in Proposition~\ref{corl2} and Proposition~\ref{corl1} respectively are translation-invariant pseudodifferential operators and can be written in the form $a(x, \xi) = \alpha(x)\beta(\xi)$ with $\alpha(x) = 1$ and $\beta(\xi) = -\xi$ or $\beta(\xi) = -\xi^2$. This satisfies the assumptions of Theorem 6 in~\cite{PDOpaper}, which guarantees that any such separable-symbol PDO can be block encoded with normalization factor $\alpha = C_\alpha C_\beta$ and gate complexity $\mathcal{O} \left( \mathrm{poly}(pd) + \mathrm{polylog} \left( \frac{1}{\varepsilon}\right)\right)$. Although both Theorem 6 and Theorem 5 (from~\cite{PDOpaper}) are applicable to the SLAC derivative operators $\widetilde{\Delta}^{(1;2)}$, Theorem 6 is used here because it provides a tighter normalization and avoids the exponential factor $2^{pd/2}$ appearing in the generic bound from Theorem 5.

Based on the discussion in this appendix, we make the following comparison between the block-encodings constructed in this work and those obtained from the framework of~\cite{PDOpaper} when applied to SLAC operators:
\newtheorem*{remark}{Remark}
\begin{remark}[Comparison]
Although both SLAC derivative operators fall within the PDO block-encoding framework, the explicit constructions provided in Sections~\ref{subsec:analyticallapbe} and \ref{SLAC1LCU}:
\begin{itemize}
    \item Avoid oracle-based assumptions in Theorems 5 and 6 of~\cite{PDOpaper}.
    \item Exploit circulant structure to realize efficient $\mathrm{LCU}$ circuits using modular quantum adder circuits,
    \item Achieve $\mathcal{O}(\mathrm{polylog}(N/\varepsilon))$ scaling using inequality-based amplitude preparation and modular arithmetic.
    \item Achieve optimality as the normalisation factor in B.E. is $\mathcal{O}(1)$ for the SLAC Laplacian.
    \item The block-encoding of first order slac derivative Eq.~\eqref{beslac1} satisfies the conditions of being a good B.E. with a scaling factor of $\mathcal{O}(n)$.
\end{itemize}
\end{remark}
Hence, the constructions in Proposition~\ref{corl1} and~\ref{corl2} offer fully explicit, implementable realizations with optimal and good asymptotic performance, respectively.

\end{document}